%% file: main.tex
\documentclass[11pt]{article}
\input{preamble}
\title{Testing Closeness of Multivariate Distributions via Ramsey Theory}

\author{
Ilias Diakonikolas\thanks{Supported by NSF Medium Award CCF-2107079, NSF Award CCF-1652862 (CAREER), and a Sloan Research Fellowship.}\\
UW Madison\\
{\tt ilias@cs.wisc.edu}\\
\and
Daniel M. Kane\thanks{Supported by NSF Medium Award CCF-2107547, and NSF Award CCF-1553288 (CAREER),  and a grant from
CasperLabs.}\\
UC San Diego\\
{\tt dakane@ucsd.edu }\\
\and
Sihan Liu\\
UC San Diego\\
{\tt sil046@ucsd.edu}
}
\begin{document}
\maketitle
\begin{abstract}
We investigate the statistical task of closeness (or equivalence) testing for multidimensional 
distributions. Specifically, given sample access to two unknown distributions $\mathbf p, \mathbf q$ on 
$\R^d$, we want to distinguish between the case that $\mathbf p=\mathbf q$ versus $\|\mathbf p-\mathbf q\|_{\mathcal A_k} > \epsilon$, 
where $\|\mathbf p-\mathbf q\|_{\mathcal A_k}$ denotes the generalized $\mathcal A_k$ distance between $\mathbf p$ and $\mathbf q$ --- 
measuring the maximum discrepancy between the distributions over any collection of $k$ disjoint,
axis-aligned rectangles. Our main result is the first closeness tester for this problem 
with {\em sub-learning} sample complexity in any fixed dimension and a nearly-matching sample complexity lower bound.

In more detail, we provide a computationally efficient closeness tester with 
sample complexity $O\left((k^{6/7}/ \mathrm{poly}_d(\epsilon)) \log^d(k)\right)$. On the lower bound side, we establish a qualitatively matching sample complexity lower bound of $\Omega(k^{6/7}/\mathrm{poly}(\epsilon))$, even for $d=2$. These sample complexity bounds 
are surprising because the sample complexity of the problem in the univariate setting is $\Theta(k^{4/5}/\mathrm{poly}(\epsilon))$. This has the interesting consequence that the jump from one to two dimensions leads to a substantial increase in sample complexity, while increases beyond that do not.

As a corollary of our general $\mathcal A_k$ tester, 
we obtain $d_{\mathrm TV}$-closeness testers for 
pairs of $k$-histograms on $\R^d$ 
over a common unknown partition, and pairs of uniform distributions supported on the union of $k$ unknown disjoint axis-aligned rectangles.

Both our algorithm and our lower bound 
make essential use of tools from Ramsey theory.    
\end{abstract}

\setcounter{page}{0}
\thispagestyle{empty}

\newpage

\input{intro}
\input{technique}

\input{upper}

\input{lower}

% \section*{Acknowledgements}
% We thank Andrew Suk for references pertaining to Theorem \ref{thm:erdos}.

\bibliographystyle{alpha}
\bibliography{ref}

\appendix
\input{appendix}

\end{document}

%% file: preamble.tex
\usepackage{bbm}
\usepackage{dsfont}

\def\colorful{0}

\usepackage{amsthm,amsfonts,amsmath,amssymb,epsfig,color,float,graphicx,verbatim, enumitem}
\usepackage{multirow}
\usepackage{algorithm}
\usepackage[noend]{algpseudocode}

\definecolor{cpurple}{rgb}{0.6,0,0.6}

\usepackage{hyperref}
\hypersetup{
  colorlinks = true,
  urlcolor = {cpurple},
  citecolor = {cpurple}
}
\usepackage{cleveref}

\usepackage{enumitem}

\usepackage{framed}
\usepackage{nicefrac}

\def\nnewcolor{1}
\ifnum\nnewcolor=1

\fi
\ifnum\nnewcolor=0

\fi

\ifnum\colorful=1

\newcommand{\inote}[1]{\footnote{{\bf [[Ilias: {#1}\bf ]] }}}
\newcommand{\dnote}[1]{\footnote{{\bf [[Daniel: {#1}\bf ]] }}}

\else

\newcommand{\inote}[1]{}
\newcommand{\dnote}[1]{}
\fi

\newtheorem{theorem}{Theorem}[section]
\newtheorem{question}{Question}[section]

\newtheorem{lemma}[theorem]{Lemma}
\newtheorem{informal theorem}[theorem]{Theorem (informal statement)}

\newtheorem{proposition}[theorem]{Proposition}
\newtheorem{corollary}[theorem]{Corollary}
\newtheorem{claim}[theorem]{Claim}
\newtheorem{fact}[theorem]{Fact}

\theoremstyle{definition}
\newtheorem{definition}[theorem]{Definition}
\newcommand{\eqdef}{\stackrel{{\mathrm {\footnotesize def}}}{=}}

\newcommand{\p}{\mathbf{p}}
\newcommand{\q}{\mathbf{q}}

\newcommand{\R}{\mathbb{R}}

\newcommand{\Z}{\mathbb{Z}}

\newcommand{\E}{\mathbf{E}}
\newcommand{\eps}{\epsilon}
\newcommand{\dtv}{d_{\mathrm TV}}

\renewcommand{\Pr}{\mathbf{Pr}}
\newcommand{\poly}{\mathrm{poly}}

\newcommand{\D}{\mathcal{D}}

\newcommand{\Yes}{\text{Yes}}
\newcommand{\No}{\text{No}}
\newcommand{\lp}{\left}
\newcommand{\rp}{\right}
\newcommand{\supp}{\text{supp}}

\renewcommand{\r}{\mathbf{r}}

\renewcommand{\t}{\mathbf{t}}

\newcommand{\iid}{\text{i.i.d.}~}

\newcommand{\Poi}{\mathrm{Poi}}

\DeclareMathOperator{\Order}{Order}
\newcommand{\abs}[1]{\lvert#1\rvert}

\newcommand{\A}{\mathcal A}
\newcommand{\Ak}{\mathcal A_k}

\newcommand\snorm[2]{\left\| #2 \right\|_{#1}}

\usepackage{mathtools}

\DeclarePairedDelimiter\floor{\lfloor}{\rfloor}

\DeclareMathOperator{\Var}{\mathbf{Var}}

\newcommand{\GY}{\mathcal Y}
\newcommand{\GN}{\mathcal N}

%% file: intro.tex
\section{Introduction} \label{sec:intro}

\paragraph{Background and Motivation}
A fundamental statistical task is to ascertain whether a set of samples 
comes from a given model, where the model may consist
of either a single fully specified probability distribution or a family of probability
distributions. The study of this broad task was initiated in a field 
now known as {\em statistical hypothesis testing} over a century ago~\cite{Pearson1900, NeymanP}; 
see, e.g.,~\cite{lehmann2005testing} 
for an introductory textbook on the topic.
In the past three decades, hypothesis testing has been extensively studied 
by the theoretical computer science and information-theory communities --- under the name {\em distribution testing} --- in the framework of property testing~\cite{RS96, GGR98}. 
It is instructive to note that the TCS style definition of hypothesis testing 
is equivalent to the minimax testing definition introduced and 
studied by Ingster and coauthors~\cite{Ing94, Ing97, IS03}.

The paradigmatic problem in distribution testing is the following: 
given sample access to one or more unknown probability distributions,
we want to correctly distinguish (with high probability) between the cases that 
the underlying distributions satisfy some global property $\cal{P}$ 
or are ``far'' from satisfying the property.
The primary objective is to obtain a tester that is statistically  efficient, i.e., 
it has information--theoretically optimal sample complexity.
An additional important criterion is computational efficiency; that is, 
the testing algorithm should run in sample-polynomial time. 
After the pioneering early works formulating this field~\cite{GR00, BFR+:00} from a TCS perspective, 
there has been substantial progress on testing a wide range of properties; 
see, e.g.,~\cite{BFFKRW:01, BDKR:02, BKR:04,  Paninski:08, PV11sicomp, ValiantValiant:11,
DJOP11, LRR11, VV14, chan2014optimal, diakonikolas2015optimal, CDKS17, DaskalakisDK18, CanonneDKS18, DiakonikolasGKP21, CanonneJKL22, CDKL22}
for a sample of works,
and~\cite{Rub12, Canonne22} for surveys on the topic.

Here we study the problem of {\em closeness testing} (or equivalence testing) 
between two unknown probability distributions. Specifically, 
given independent samples from a pair of distributions $\p, \q$, we want to
determine whether the two distributions are the same versus $\eps$-far from each other.
Early work on this problem~\cite{BFR+:00} focused on the setting that 
$\p, \q$ are arbitrary discrete distributions of a given support size $n$, 
and the metric used to quantify ``closeness'' is the $\ell_1$-distance 
(equivalently, total variation distance). It is now known~\cite{chan2014optimal} 
that the optimal sample complexity of $\ell_1$-closeness testing 
for distributions with support of size $n$ is $\Theta(\max\{n^{2/3}/\eps^{4/3}, n^{1/2}/\eps^2 \})$.

In summary, it is known that the complexity measure 
determining the sample complexity of 
testing the equivalence (and a range of other related properties) 
of unstructured 
(i.e., potentially arbitrary) discrete distributions 
is the domain size of the underlying distributions. 
Unfortunately, this implies that if $\p, \q$ are (potentially arbitrary) 
continuous distributions (even in one dimension!), no closeness 
tester with finite sample complexity exists. There are two natural approaches
to circumvent this bottleneck. The first approach is to assume that $\p, \q$
have some nice structure, in which case the domain size may
not be the right complexity measure for the testing problem. 
The second approach is to make no assumptions on the underlying distributions, 
but relax the metric under which we measure closeness.

Interestingly, it turns out that these two seemingly orthogonal
approaches are intimately related to each other. 
In particular, for the important special case of {\em one-dimensional} distributions, 
a line of works, see, e.g.,~\cite{DDSVV13, DKN:15, diakonikolas2015optimal,DKN17}, 
developed a general framework that yields optimal testers (for closeness and other properties) 
for a range of structured distribution families. 
The key idea underlying these testers is to
design a {\em single} tester for {\em arbitrary} 
one-dimensional distributions 
{\em but under a different --- carefully selected --- metric}; 
and then appropriately use this metric as a proxy for the total variation distance 
(for each structured distribution family of interest). 

In more detail, for one-dimensional distributions $\p, \q: \R \to \R_+$, 
the appropriate metric is known as {\em $\Ak$-distance}~\cite{DL:01, CDSS14} 
and is defined as follows: 
The $\Ak$-distance between one-dimensional distributions $\p$ and $\q$, 
denoted by $\|\p-\q\|_{\Ak}$, is defined as the maximum $\ell_1$-distance between 
the reduced distributions\footnote{The reduced distribution obtained from $\p$ 
with respect to a partition of the domain into $k$ subsets $R_1, \ldots, R_k$ 
is the discrete distribution with support size $k$ assigning probability mass 
$\p(R_i)$ to the $i$-th point.} obtained from $\p, \q$ over all partitions of the domain 
in at most $k$ {\em intervals}. 
The motivation for this particular definition 
of the $\Ak$-distance~\cite{CDSS14, DKN:15} between
one-dimensional distributions comes from the VC-inequality (see, e.g., page 31 of~\cite{DL:01}). 

The positive integer $k$ in the definition of the $\Ak$-distance is a tunable parameter 
that is selected appropriately depending on the application. For $k=2$, the $\Ak$-distance 
amounts to the distance between the {\em cumulative} distribution functions 
(known as Kolmogorov distance). As $k$ increases, the metric becomes stronger 
and converges to the total variation distance when $k \to \infty$ 
(under mild assumptions on the distributions). Moreover, if the underlying distributions 
$\p, \q$ belong to some class of shape restricted densities (e.g., 
univariate histograms or log-concave distributions), a finite value of $k$ suffices 
so that the $\Ak$-distance closely approximates the total variation distance.

It is worth noting that, in addition to distribution testing, the one-dimensional
$\Ak$ distance has been has been a crucial ingredient in developing efficient {\em learning} algorithms for structured univariate distributions~\cite{CDSS13, CDSS14, ADLS17, Chen0M20}.

\paragraph{Testing Closeness of Multivariate Distributions} 
The main motivation behind this work is to generalize the aforementioned framework 
to the {\em multivariate} setting with a focus on the task of closeness testing. 
A first step to achieve this is an appropriate generalization of the notion of
$\Ak$-distance which applies to one-dimensional distributions)
for distributions on $\R^d$ for all $d \geq 1$. 
Here we study the following natural definition, 
that has been previously used in the context of learning~\cite{DLS18} 
and uniformity testing~\cite{diakonikolas2019testing} 
for multivariate distributions.

\begin{definition}[Multidimensional $\Ak$-distance]\label{def:Ak-gen}
For two probability distributions (with densities/mass functions) 
$\p, \q: \R^d \mapsto \R_{+}$ and $k \in \mathbbm Z^+$, 
we define the \emph{multi-dimensional $\mathcal A_k$-distance} between $\p$ and $\q$ as the maximum 
value of $\sum_{i=1}^k \abs{\p(R_i) - \q(R_i)}$ for $k$ arbitrarily chosen non-overlapping axis-
aligned rectangles $\{R_i\}_{i=1}^k$ in $\R^d$.
\end{definition}

\noindent {\bf \em Motivation for \Cref{def:Ak-gen}}
Recall that the total variation distance between two distributions 
$\p, \q$ on $\R^d$ is defined as 
$\dtv(\p, \q) = \sup_{A \in \cal{S}}|\p(A) - \q(A)|$, 
where $\cal{S}$ is the collection of all measurable subsets on $\R^d$.
Since learning or testing under the total variation distance may be too 
strong a goal if the underlying distributions lack structure, 
a reasonable compromise is to consider alternative metrics. 
The VC-inequality states the following: Let $\cal{A}$ be any collection of 
subsets of $\R^d$ with VC-dimension $d$. Then for {\em any} distribution $\p$ on $\R^d$ it holds that 
$\E[\sup_{A \in \cal{A}} |\widehat{\p}_n(A) - \p(A)|] = O(\sqrt{d/n})$, 
where $\widehat{\p}_n$ is the empirical distribution obtained after drawing
$n$ i.i.d.\ samples from $\p$. In other words, for $n \gg d/\eps^2$, the empirical distribution is $\eps$-close to $\p$ with respect to the $\cal{A}$-metric, 
defined as 
$\|\p-\q\|_{\cal A} \eqdef \sup_{A \in \cal{A}} |\p(A) - \q(A)|$. 
For the univariate case, the $\Ak$-distance defined in the aforementioned 
works~\cite{CDSS14, DKN:15, diakonikolas2015optimal} is obtained from the 
$\A$-metric by considering the family of all unions of at most $k$ intervals 
(which has VC-dimension $2k$). 

Our \Cref{def:Ak-gen} is a natural generalization of the one-dimensional 
definition, where we consider the family of all unions of at most $k$ rectangles,
which has VC-dimension $\tilde{\Theta}(k d)$. Since learning an arbitrary 
distribution on $\R^d$ under this metric requires $\tilde{\Theta}(k d)/\eps^2$ 
samples, it is natural to ask whether the distribution testing problem has 
qualitatively lower sample complexity. We also note that the multidimensional 
$\Ak$-distance is a strengthening of the Kolmogorov-Sminov (KS) metric and converges to 
the total variation distance as $k \rightarrow \infty$ (under mild assumptions). 
It should be noted that a line of work in 
mathematical statistics --- see, e.g.,~\cite{Bic69, FR79, Henze88, JPZ97} 
for some classical works  --- 
has developed two-sample testers (aka closeness testers) for {\em non-parametric} 
multivariate distributions under the KS metric. Our work can be viewed as a 
strengthening and generalization of these results in the minimax setting.

We believe that, in addition to being a potential tool for performing multivariate $\dtv$-closeness testing for structured 
distributions, the $\Ak$ distance is an interesting metric on its own merits. 
To see this, we recall that one of the main motivations 
for considering the total variation distance is the following property:
If a decision algorithm is run twice on different inputs that follow two distributions that are close in total variation distance, 
then the acceptance probabilities will also be 
approximately the same in the two cases. 
Hence, for two distributions $\p, \q$ that have passed 
the $\dtv$-closeness testing, we can be confident 
that running some downstream decision algorithm on inputs 
drawn from $\p$ and from $\q$ should give similar results.
For the $\Ak$ distance, we have an analogous property 
if one restricts the algorithm in the above statement 
to be an axis-aligned decision tree (i.e., a 
decision tree whose leaf nodes follow the branching rule 
of $ x_i < b $ for some coordinate $i \in [d]$ and some real number $b \in \R$) with at most $k$ leaves.
Though being a restricted family of algorithms, 
axis-aligned decision trees are commonly used in machine learning 
applications due to their exceptional interpretability; 
see, e.g.,~\cite{yildiz2001omnivariate, bengio2010decision, bruch2020learning}.
This suggests that testing in $\Ak$ distance, even though being 
a weaker test compared to its $\dtv$-counterpart for arbitrary distributions 
(which is provably impossible without structural information),
may be sufficient for certain structured downstream decision-making tasks.

\medskip

We return to our closeness testing task.
One approach to solve the multidimensional\footnote{We will henceforth omit the 
term ``multidimensional'' when it is clear from the context, 
and use the term $\Ak$-
distance for multivariate distributions as well.} 
$\Ak$-closeness testing problem
is to learn $\p$ and $\q$ up to $\Ak$-distance $\eps/4$,
and then check whether the hypotheses are $\eps/4$-close to each other. 
Thus, the sample complexity of closeness testing is bounded above by the sample 
complexity of learning (within constant factors). 
Since $\tilde{\Theta}(k d / \eps^2)$ samples suffice to learn 
an arbitrary distribution on $\R^d$ up to $\Ak$-distance $\eps$, 
the naive ``testing-by-learning'' approach requires $\Omega(k)$ samples 
(even in one dimension and for constant $\eps$).

It is natural to ask whether a better sample size could be achieved for testing, 
since closeness testing is, in some sense, less demanding than learning. 
That is, the goal is to develop a closeness tester with sample complexity {\em strongly sublinear} 
in $k$, namely $O(k^c)$ for some constant $c<1$. The aforementioned line of work on univariate 
distributions~\cite{DKN:15, diakonikolas2015optimal,DKN17} developed identity 
and closeness testers under the $\Ak$-distance with strongly sublinear sample complexity. 
These testers were also applied to give 
total variation distance testers for classes of 
``shape constrained'' distributions~\cite{BBBB:72, GJ:14}, 
including histograms and logconcave distributions.

Concretely, for the problem of closeness testing of univariate distributions,~\cite{diakonikolas2015optimal} developed a sample-optimal $\Ak$-closeness tester
with sample complexity of $\Theta(k^{4/5}/\poly(\eps))$ (for not too small $\eps$). 
Interestingly, this bound differs from the sample complexity of closeness testing discrete distributions on $k$ points, which is $\Theta(k^{2/3}/\poly(\eps))$)~\cite{chan2014optimal}.

This discussion motivates the following natural question: 
\begin{center}
{\em What is the sample complexity of $\Ak$-closeness testing for multivariate distributions?}
\end{center}
Prior to this work, no closeness tester with sub-learning sample complexity was known even for 
$d=2$. The main contribution of this work is a 
{\em sample near-optimal and computationally efficient}
$\Ak$-closeness tester in any {\bf fixed}\footnote{
We emphasize here that the focus of our work is 
in closeness testing of non-parametric families 
of distributions. In non-parametric 
estimation/testing, the sample complexity 
inherently scales exponentially with the dimension; hence, 
it is standard to consider the dimension as being fixed. For example, estimation/testing for the class of log-concave distributions on $\R^d$ 
is known to require $2^{\Omega(d)}$ samples (see \cite{kim2016global}).} dimension. 
Moreover, we show that the sample complexity of our 
tester is optimal as a function of $k$, within logarithmic factors. 
As an immediate corollary, we obtain the first closeness tester 
for multivariate histogram distributions (with respect to the 
same unknown set of axis-aligned rectangles) under the total variation distance.

Specifically, our main result (Theorem~\ref{thm:main-intro}) establishes the following:
{\em For any $k, d \in \Z_+, \eps>0$, and sample access to arbitrary distributions $\p, \q$ on $\R^d$,
there exists a closeness testing algorithm under the $\mathcal{A}_k$-distance
using $O\left( (k^{6/7}/\poly_d(\eps)) \log^{d}(k) \right) $ samples.
Moreover, this bound is information-theoretically optimal as a function of $k$, even for $d=2$.} We remark that our $\mathcal{A}_k$-testing algorithm applies to {\em any} pair of distributions (over both continuous and discrete domains). 

As a corollary, we obtain the first closeness tester 
(with sub-learning sample complexity) 
between $k$-histograms with respect to the total variation distance. 
A probability distribution on $\R^d$ with density $\p$
is called a \emph{$k$-histogram} if there exists a partition of the support 
into $k$ axis-aligned rectangles $R_1, \ldots, R_k$ such 
that $\p$ is constant on $R_i$, for all $i = 1,\ldots, k$.
This is one of the most basic non-parametric distribution families and
have been extensively studied in 
statistics~\cite{Scott79, FreedmanD1981, Scott:92, LN96, Devroye2004, WillettN07, Klem09} and computer science --- including database 
theory~\cite{JPK+98,CMN98,TGIK02,GGI+02, GKS06, ILR12, ADHLS15} 
and theoretical ML~\cite{DDS12soda, CDSS13, CDSS14, CDSS14b, ADLS17, AcharyaDK15, 
DDSVV13, diakonikolas2015optimal, DKN:15, DKN17, diakonikolas2019testing, CDKL22}. 
Prior to this work, no closeness testing algorithm with sub-learning sample 
complexity was known for $k$-histograms, even for $d=2$. As a corollary of our 
main result, we provide such an algorithm (see Corollary~\ref{cor:md-hist}) for 
the case that the two histograms are supported on the same unknown partition.
In addition, we also obtain $\dtv$-closeness 
tester for uniform distributions supported on 
some unknown $k$ disjoint axis-aligned 
rectangles (see \Cref{cor:rectangle-hypothesis}).
We remark that though histograms and uniform distributions 
over unions of axis-aligned rectangles
are conceptually similar, these two families of 
distributions are orthogonal to each other.

\subsection{Our Results}

We study the complexity of closeness testing between two (arbitrary) 
distributions $\p, \q$ on $\R^d$ with respect to the 
$\Ak$ distance.
Our main result is the following.

% When the distributions are 1-dimensional, the work of \cite{diakonikolas2015optimal} gives that the sample complexity of the problem 
% We give a natural generalization of this distance metric in the multi-dimensional case 

\begin{theorem}[Main Result] \label{thm:main-intro}
Given $\eps > 0$, integer $k \geq 2$, and sample access to distributions with density 
functions $\p, \q : \R^d \rightarrow \R_{+}$, there exists a computationally 
efficient algorithm which draws
$C \, 2^{d/3} \, k^{6/7} \log^{3d}(k)/\eps^{\alpha_d}$
samples from $\p, \q$, for a sufficiently large universal constant $C>0$, 
where $\alpha_d = O(d^2 2^{2^{d+1}})$, and with probability at least $2/3$ correctly 
distinguishes whether $\p = \q$ versus $\snorm{\Ak}{\p - \q} \geq \eps$.
Moreover, $\Omega \lp( \min\{ k^{6/7} / \eps^{8/7}, k\} \rp)$ many samples are information-theoretically necessary for this hypothesis testing task, even if $\p, \q$ are two-dimensional discrete distributions on a sufficiently large domain.
\end{theorem}

\paragraph{Discussion} 
To interpret \Cref{thm:main-intro}, 
some comments are in order. We reiterate that the focus of our work 
is on the {\em non-parametric} setting and consequently we view the 
dimension $d$ as a fixed constant. In this regime, the sample complexity of our algorithm is 
$\tilde{O}_d (k^{6/7}) / \poly_d(\eps)$.

The one-dimensional special case of our closeness testing result was solved
in~\cite{diakonikolas2015optimal}, where the authors established a tight sample complexity 
bound of $ \Theta(k^{4/5}/\eps^{6/5}+ k^{1/2}/\eps^2)$.
Prior to our work, no $o(k)$ sample upper bound was known for this testing 
problem even for $d=2$ and $\eps=0.99$.

For the regime of fixed dimension that we focus on, 
our upper and lower bounds are essentially 
optimal in terms of their dependence on $k$ --- the main parameter of interest. 
For simplicity, let us fix $\eps$ to be a universal constant. 
Examining the exponent of $k$ in the dominant term of the sample complexity, 
we observe a surprising pattern: the exponent begins at $4/5$ when $d = 1$ 
(as follows from the prior work~\cite{diakonikolas2015optimal}), 
jumps to $6/7$ when $d=2$, and then stays at $6/7$ as $d$ increases 
(as follows from \Cref{thm:main-intro})!
This suggests that the $d=1$ case is a degenerate case 
and the essence and complexity of the problem is not entirely revealed until $d=2$.

Some remarks are in order regarding the dependence of the sample complexity 
on the parameters $\eps$ and $d$. 
First, we briefly comment on the $\log^{d}(k)$ term. 
Perhaps surprisingly, prior work~\cite{diakonikolas2019testing} 
has shown a sample complexity {\em lower bound} of $(\sqrt{k}/\eps^2) \Omega(\log(k) / d)^{d-1}$ 
for the {\em easier} problem of $\Ak$-uniformity testing.
This suggests that the $\log^{d}(k)$ factor is necessary for closeness testing as well, 
assuming that $k$ is sufficiently large. Finally,
we conjecture that the correct dependence on $\eps$ 
in the sample complexity of this task should be a fixed degree polynomial, 
independent of $d$. We leave this as an interesting technical question for future work (see \Cref{q:eps-dependency}).

Regarding our sample complexity lower bound,
\Cref{thm:main-intro} does not specify how large the domain size of 
the hard distributions needs to be. 
Due to the application of Ramsey-theoretic arguments 
in the proof of our lower bound, 
we need it to be extremely large in terms of $k$ (a tower function of $k$).
In Section~\ref{sec:domain-optimize}, we show that the domain size 
can be optimized to be (at most) doubly exponential in $k$ --- 
using a significantly more sophisticated construction (\Cref{thm:lower-bound-refined}).

As immediate corollaries of our main theorem, we obtain $\dtv$-closeness testers 
(with strongly sub-learning sample complexities) for multivariate structured distributions. 
In particular, we highlight here the $\dtv$-closeness tester for distributions in $\R^d$ 
that are $k$-histograms, i.e., piecewise constant over (the same) $k$ unknown disjoint axis-aligned 
rectangles. Notably, the sample complexity of this tester is the same as that of our 
$\Ak$ closeness testing. This implication and additional applications are given 
in Section~\ref{sec:application}.

\subsection{Overview of Techniques} \label{sec:techniques}

Here we provide a detailed overview of our technical approach
to establish our upper and lower bounds. 

\paragraph{Closeness Tester}
By definition of the $\Ak$ distance, 
there exist $k$ 
disjoint axis-aligned rectangles $\{R_i\}_{i=1}^k$ on $\R^d$ 
which witness the $\Ak$ discrepancy between $\p$ and $\q$; that is, 
$\sum_{i=1}^k \abs{ \p(R_i) - \q(R_i) } = \snorm{\Ak}{\p - \q} $.
If we knew what these rectangles were, the testing task would be easy.
Indeed, we could simply consider the reduced measures of $\p$ and $\q$ 
over $\{R_i\}_{i=1}^k$ (recall that these measures, after normalization, become distributions with support size $k$ that we can simulate access to) 
and then use an optimal $\ell_1$-closeness tester as a black-box.
Given the optimal $\ell_1$-closeness tester of~\cite{chan2014optimal}, such 
an approach would lead to a sample complexity upper bound of $O(k^{2/3})$ 
(for constant $\eps$). Of course, the difficulty is that 
we are not given these rectangles a priori, which intuitively could
make the problem require more samples than $\ell_1$-closeness testing 
on a domain of size $k$\footnote{In hindsight, given our sample complexity lower bound of $\Omega(k^{6/7})$, the fact that the witnessing rectangles 
are unknown implies that the $\Ak$ closeness testing problem {\em provably} 
requires more samples.}. 

The lack of a priori knowledge of the witnessing rectangles is {\em the} major obstacle towards developing a closeness tester with sub-learning sample complexity. 
Overcoming this bottleneck necessitates the bulk of the new technical ideas developed here.
To achieve this, at a very high-level, 
we will proceed to compute {\em some} small set of rectangles 
that capture a ``non-trivial''\footnote{Quantitatively, 
the term ``non-trivial'' here means ``a function of the form $\poly_d(\eps)$''.} 
fraction of the discrepancy (i.e., $\Ak$-distance) between $\p$ and $\q$. 

A simple but important observation in this context is the following: 
one should not expect that an {\em obliviously} selected 
(i.e., without drawing samples from the underlying distributions)
set of rectangles suffices for this purpose. 
Indeed, this holds even for the one-dimensional setting:
as was noted in \cite{diakonikolas2015optimal}, 
any obliviously chosen set of intervals 
may capture \emph{no} discrepancy between a pair of 
adversarially chosen one-dimensional distributions even though they have large $\Ak$ distance.

That is, it appears necessary to select rectangles 
using samples from the tested distributions.
Note that, in any dimension $d$, one needs at least two 
points in $\R^d$ to define an axis-aligned rectangle.
In particular, given two sample points $x,y \in \R^d$, 
we consider the following natural rectangle 
defined by these points, namely
\[ R_{x,y} \eqdef \{  z = (z_1, \ldots, z_d) \in \R^d \mid \min(x_i, y_i) \leq z_i \leq \max(x_i, y_i) \text{ for all } i \in [d] | \} \;.\]
The main intuition behind this definition is the following.  
Suppose that we draw two samples $x,y$ from the mixture 
$(1/2)(\p + \q)$ (the uniform mixture of $\p$ and $\q$), 
and they both happen to land in some rectangle $R$
such that the discrepancy $\abs{\p(R) - \q(R)}$ is non-trivial. 
Then, intuitively, the rectangle $R_{x,y}$ 
will capture (in expectation) a non-trivial fraction of the rectangle $R$, 
and therefore also a non-trivial fraction 
of the discrepancy between $\p$ and $\q$ within $R$.
The latter statement turns out to be true (see \Cref{lem:rectangle-discrepancy}) and 
its proof makes essential use of tools from Ramsey theory.

Before we provide an overview of the ideas required to prove \Cref{lem:rectangle-discrepancy}, 
we explain how to leverage this statement to develop our closeness tester.
Suppose that the $\Ak$-distance between $\p, \q$ is $\eps$. 
Then at the cost of increasing $k$ and decreasing $\eps$ by at most a constant factor, 
we can without loss of generality assume 
that there exist $k$ rectangles $\{R_i\}_{i=1}^k$, 
each of which has probability mass approximately $1/k$  
and witnesses roughly $\eps/k$ discrepancy.
If we draw $m$ samples from each of $\p, \q$, 
approximately $m^2/k$ of these rectangles will contain two samples.
Given \Cref{lem:rectangle-discrepancy}, we know that each pair 
of samples landing in some $R_i$ can be used to define a rectangle 
that, with some non-trivial probability, captures a non-trivial fraction of the discrepancy 
between $\p$ and $\q$ within $R_i$.

A potential concern is how one would find the right set of rectangles 
defined by the sample points
(i.e., that capture enough discrepancy). 
The statement of \Cref{lem:rectangle-discrepancy} only ensures 
the existence of such rectangles, but offers no clues on how 
one could reliably identify them.
Perhaps the most natural approach is to to try 
all possible sets of $\Theta(m^2/k)$ many rectangles 
defined by the coordinates of the sample points, 
and then run a standard $\ell_1$-closeness tester (on the corresponding
reduced distributions) to compare the probability mass of $\p$ and $\q$ 
on the selected rectangles.
Unfortunately, in addition to its computational intractability, 
it is not even clear whether this method can lead to \emph{any} sample complexity sublinear in $k$.
In particular, the standard analysis of the above strategy 
will apply the union bound on the failure probabilities of running
the $\ell_1$-closeness tester on each possible reduced distribution 
(defined by each set of rectangles).
Since there are at least $m^{ \Omega(m^2/k) }$ many different ways 
to select the set of rectangles, this increases the sample complexity 
of the $\ell_1$-closeness tester by a factor of $\Omega(m^2/k)$, 
making it hopeless to achieve any sublearning sample complexity 
(even balancing the quantities $m^2/k$ and $m$ directly will give us $m = k$).

To circumvent this obstacle, 
we leverage an idea from \cite{diakonikolas2019testing}, that
we term \emph{Grid Covering} (see 
Definition~\ref{def:grid-cover}).
At a high level, we show that
we can cover the set of all possible rectangles 
that can be defined by the sample coordinates --- 
which we refer to as $\mathcal S$ --- 
by a carefully chosen subset of these rectangles --- 
which we refer to as $\mathcal F$ --- 
such that each rectangle from $\mathcal S$ can be expressed 
as the union of at most {\em polylogarithmically many} 
rectangles from $\mathcal F$.
Moreover, $\mathcal F$ will be constructed to have the subtle property 
that any point in $\R^d$ is contained in at most polylogarithmically many 
rectangles within the subset (in sharp contrast, in the worst case, 
a point may be included in a constant fraction of $\mathcal S$.).

To take advantage of this property, we consider the notion 
of \emph{induced distributions} (see \Cref{def:ind-d}), 
$\p^{\mathcal F}, \q^{\mathcal F}$, on $\mathcal F$: 
to sample from $\p^{\mathcal F}$, we first draw 
a sample point $x \in \R^d$ from $\p$ and 
return uniformly at random some rectangle from $\mathcal F$ 
that includes $x$ (and similarly for $\q^{\mathcal F}$). 
As a consequence of the aforementioned properties of $\mathcal F$, 
the discrepancy (under some appropriate metric) between $\p^{\mathcal F}$ 
and $\q^{\mathcal F}$ will shrink by at most a polylogarithmic 
factor compared to the discrepancy between $\p$ and $\q$ 
captured by the best $\Theta(m^2/k)$ rectangles 
from our original collection of rectangles $\mathcal S$ 
(defined using the sample points); see Lemma~\ref{lem:induced-discrepancy}.
Importantly, the new pair of distributions are both discrete, 
and the discrepancies between them will be supported on a small number 
of domain elements. 
Therefore, one could hope to apply techniques from ``standard" 
$\ell_1$-closeness testing of discrete distributions from there on. 
While this turns out to be manageable, we emphasize that the induced
distributions still have very large support size. Hence, a direct application 
of $\ell_1$-closeness testing on an arbitrary discrete domain is not sufficient for our purposes.
We will return to this issue when we analyze the sample complexity of our tester in detail.

It remains to show correctness of this scheme. That is, we want to 
establish that there exists a small set of rectangles
defined by the sample points which capture a non-trivial amount 
of discrepancy between $\p$ and $\q$ with high constant probability. 
To show this, we return to $\{R_i\}_{i=1}^k$ --- a set of $k$ rectangles 
which witness the $\Ak$ distance between $\p$ and $\q$. 
We will prove that for each of these rectangles $R_i$, 
if two samples $x,y \in \R^d$ are drawn from $R_i$, 
there is a non-negligible probability that $R_{x,y}$ --- 
the rectangle defined by $x$ and $y$  --- 
captures a non-trivial fraction of the discrepancy in $R_i$. 
(see \Cref{lem:rectangle-discrepancy} and 
and its proof in Section~\ref{sec:discrepancy}).

As a starting point to achieve this, we show 
that if two sample points $x,y$ are drawn 
from the restriction of $(1/2)(\p+\q)$ to a rectangle $R$, 
there is a decent probability that $R_{x,y}$ will capture 
a non-trivial fraction of the mass of $(1/2)(\p+\q)$ in $R$ 
(see \Cref{lem:constant-mass}). 
This statement turns out to be {\em essentially equivalent} 
to a result in Ramsey theory shown by De Bruijn that can be viewed
as a generalization of the classical Erd\H{o}s-Szekeres theorem
(see Fact~\ref{thm:erdos}). 
In particular, De Bruijn showed that given $N$ points in $\R^d$ (for $N$ at least doubly exponential in $d$), 
there exists a triplet $(x,y,z)$ of these points such that 
one of the points $z$ is inside the rectangle $R_{x,y}$
%with corners being  
defined by the other two. 
This statement provides us with the desired discrepancy result 
for the special case that one of $\p, \q$ has non-trivial probability 
mass in the rectangle $R$ while the other has mass zero.

To prove the desired discrepancy result for the general case, 
we introduce and leverage the notion of \emph{discrepancy density} of a set $S \subset \R^d$, 
defined to be the discrepancy between $\p, \q$ in $S$ 
divided by the total mass assigned by $\p$ and $\q$ in $S$ 
(see Definition~\ref{def:density}).
At a high level, our analysis proceeds as follows. 
We define an iterative process that selects rectangles 
with increasing discrepancy density. As the discrepancy density approaches one, 
the situation qualitatively resembles the case that only one of $\p, \q$ 
assigns non-zero mass to the rectangle.
We now provide some further details of the process.
Note that if in expectation the rectangle $R_{x,y}$ --- 
defined by random points $x,y$ from $R$ --- 
captures a non-trivial amount of discrepancy 
between $\p, \q$ in $R$, we are done. 
Otherwise, 
there exists a rectangle $R_{x^*,y^*}$
such that the probability masses of $\p$ and $\q$ in $R_{x^*,y^*}$ differ by a negligible amount.
As a result, 
since the probability masses of $\p$ and $\q$ within $R_{x^*,y^*}$ are approximately the same,
the {\em complement} of $R_{x^*,y^*}$, 
which we denote by $S := R \setminus R_{x^*,y^*} $ , 
must have higher discrepancy density between $\p$ and $\q$. 
Since the complement $S$ can be shown to be a union of a small number 
of axis-aligned rectangles (see Claim~\ref{clm:carve-complement}), we can select 
one of these rectangles to restart the process.
By iterating this procedure, we obtain a sequence of rectangles 
whose discrepancy densities increase monotonically 
until we reach the case that a random pair of points drawn from one of these rectangles can capture a non-trivial amount of discrepancy between $\p$ and $\q$ in expectation. 

Up to this point, we have summarized the key ideas needed for the correctness
analysis of our closeness tester. We now proceed to describe the tester in more detail 
and provide a sketch of its sample complexity.
Using \Cref{lem:rectangle-discrepancy} and
(an adaptation of) the grid-covering approach of \cite{diakonikolas2019testing}, 
we obtain a pair of discrete induced distributions 
(that we can simulate access to based on the samples drawn) 
such that they have $\poly_d(\eps) \; m^2/k^2$ discrepancy concentrated over approximately $m^2/k$ domain elements (up to polylogarithmic factors). 
Leveraging the guarantees of the pair of induced distributions we have constructed, 
it is tempting to apply the so called $\ell_{1, k}$-tester from \cite{DKN17}~\footnote{The original $\ell_{1,k}$-tester is for identity testing; one can adapt these techniques to derive 
an $\ell_{1,k}$-tester for closeness testing.}. 
In particular, given samples from a pair of discrete distributions, 
such a tester aims at distinguishing between the cases 
that the underlying distributions are equal versus 
far in $\ell_{1, k}$-distance --- i.e., there exist $k$ domain elements 
such that the $\ell_1$-distance restricted to these elements is large.
Due to the ``sparsity assumption'' on the discrepancies,
the sample complexity of $\ell_{1, k}$-closeness testing is 
comparable to that of standard $\ell_1$-closeness testing 
on a domain of size $k$, even though the actual domain size 
of the input distributions may be much larger.
In particular, using the guarantees of the $\ell_{1, k}$ tester in a black-box manner, we can detect the 
existing discrepancy between the pair of induced distributions obtained with sample size approximately 
\[ (m^2/k)^{2/3}/ \lp(\poly_d(\eps) \; m^2/k^2\rp)^{4/3}
+ (m^2/k)^{1/2}/ \lp(\poly_d(\eps) \; m^2/k^2\rp)^{2} \;,\]
where $m$ is the initial number of samples drawn to construct the rectangles.
Balancing the number of samples used for defining the rectangles 
and the number of samples used for detecting the discrepancy, 
we obtain that 
$m = \tilde \Theta_d \lp( k^{7/8}/\poly_d(\eps) \rp)$ suffices. 
This sample upper bound is strongly sub-linear in $k$, 
but it turns out (in hindsight)
to provide a sub-optimal dependence on $k$. 

Intuitively, 
the reason that the above guarantee turns out to be sub-optimal is the following. 
There would be a key property of our underlying discrete distributions 
left unused if we were to apply the guarantees of $\ell_{1, k}$-testing in a black-box manner.
Specifically, since the $k$ rectangles 
that witness the $\Ak$-distance between $\p$ and $\q$ 
are themselves each of probability mass at most $O(1/k)$,  
the rectangles defined by our sample points (that capture non-trivial discrepancies) 
will also each be of mass at most $O(1/k)$.
This in turn implies that in the end we only need to detect 
discrepancies supported on a few number of \emph{light} domain elements (i.e., domain elements with small probability masses in the constructed discrete distributions). 
By carefully incorporating this additional property (i.e., that the bins witnessing 
discrepancies are themselves of small probability mass) 
into the analysis of the $\ell_{1, k}$-tester, 
we obtain an improved sample complexity upper bound of
\[ (m^2/k)^{2/3}/ \lp(\poly_d(\eps) \; m^2/k^2\rp)^{4/3} \;. \] 
See \Cref{lem:small-support-l2} for the new tester and its analysis. 
We believe that this tester --- customized for detecting discrepancies supported on a small number of \emph{light} domain elements --- 
may be applicable in other scenarios, 
as it allows us to escape from some worst-case scenarios of $\ell_{1,k}$ testing.
Finally, balancing the number of samples used for defining the rectangles 
and that of samples used for detecting the discrepancy gives 
us a sample bound of approximately $m = \tilde \Theta_d(k^{6/7}/\poly_d(\eps))$.

\paragraph{Sample Complexity Lower Bound}
Our sample complexity lower bound applies specifically for $2$-dimensional distributions. This suffices for us to conclude that our sample upper bound 
is nearly optimal as a function of $k$ for any constant dimension $d>1$.
 
The starting point of our sample lower bound technique 
is the lower bound for one-dimensional $\Ak$ closeness testing 
shown in \cite{diakonikolas2015optimal}. 
Specifically, we start by showing that it is no loss
of generality to establish a lower bound for ``order-based'' testers, 
and then prove a lower bound for such testers. 
In the proceeding discussion, we elaborate on each of these steps.

We start by noting that most reasonable testers seem to only be able 
to take advantage of the ordering of the $x$-coordinates and the $y$-coordinates of 
the points they observe --- and not the precise numerical values of these 
coordinates (see Definition~\ref{def:order-sampling}). 
We call such a tester an {\em order-based tester.} 
Intuitively, this holds because the $\Ak$ distance is invariant 
under applying a monotonic transformation to all of the $x$-coordinates 
or all of the $y$-coordinates, 
and only the ordering of these coordinates is invariant under all monotonic transformations. 
In fact, we show that if there exists a non order-based $\Ak$ closeness tester 
on a domain of size $N$, we can use it to construct an order-based tester 
that has almost the same guarantees --- albeit on a smaller domain 
(see \Cref{lem:order-reduction}). Hence, using our reduction, we can translate 
any sample complexity lower bound against order-based testers 
into one against general testers at the cost of increasing the domain size.
To obtain the reduction, we show that for any 
$2$-dimensional $\Ak$ tester 
on a sufficiently large domain there exists 
a large subset of its domain such that if the samples 
are drawn from the subdomain, 
the general tester's output will depend only on the order of the samples. 
In other words, restricted to this subdomain, the tester 
becomes exactly an order-based tester. 

The argument itself resembles the one in \cite{diakonikolas2015optimal}.
The key difference is that, due to the tester being 
$2$-dimensional, the structure of the order 
information becomes much more complicated. 
More specifically, there is now order information from both 
of the dimensions.   
To deal with this issue, we need to take a two-fold approach. 
Namely, we need to first select a subset of coordinates 
in the first dimension 
to make the tester's  output independent of the samples' order information 
in the first dimension, 
and then adaptively select the subsets of coordinates 
in the second dimension to hide the remaining order information (see \Cref{lem:order-reduction}).

For order-based testers, we construct families of distributions that are hard to distinguish.
Lying in the center of the construction are two small gadgets, 
each consisting of a pair of distributions. 
We denote the two gadgets as $\GY$ and $\GN$ respectively.
In the $\GY$ gadget, the two distributions are both uniform 
distributions supported on the edges of a square, 
whose diagonals are parallel to the $x$ and $y$ axis respectively 
(which we term a ``diagonal square''). 
In the $\GN$ gadget, one distribution 
is distributed uniformly over a randomly 
chosen pair of parallel edges of the square, 
and the other one 
is distributed uniformly over the remaining two edges. 
The key point is that though the two distributions 
in the $\GY$ gadget are identical 
and the distributions in the $\GN$ gadget 
have $\Ak$ distance equal to one 
(even for $k = 4$), 
we show that no order-based tester can distinguish 
between the two gadgets 
when fewer than three samples are drawn  
%and only the order information of these samples are made available  
(see Section~\ref{sec:square-edge}).
% This can already be seen as a mini-example of our lower bound construction for $k=4$.

To construct the full hard instance, we replicate the gadgets 
many times in a fairly standard way.
In particular, we let $\p, \q$ have their supports in several ``boxes''. 
If the tester draws $m$ samples, to introduce ``noise'', 
we produce roughly $m$ heavy boxes on which $\p, \q$ are identical. 
We also have $k$ light boxes each with mass approximately 
$\eps/k$ on which $\p, \q$ 
either use the construction of the $\GY$ gadget, 
and are therefore identical (if we want to construct $\p = \q$);  
or they use the construction of the $\GN$ gadget, 
and are therefore far from each other 
(if we want to construct $\snorm{\Ak}{\p - \q} = \eps$).
As we have discussed, observing up to three samples 
from any of the light boxes 
gives an order-based tester no information regarding 
which case one is in.  
In other words, one will only gain information 
from light boxes 
with at least four samples; note that 
there will only exist approximately $m^4/k^3$ 
such boxes if one draws $m$ samples. 
Additionally, the $m$ heavy boxes will ``add noise'' 
on the order of $\sqrt{m}$, 
and thus one can only distinguish between the two cases if 
$m^4/k^3 \gg m^{1/2}$ (or equivalently $m \gg k^{6/7}$). 
This heuristic argument can be made rigorous 
with an appropriate use of information theory 
(see Section~\ref{sec:construction}).

A disadvantage of the above proof technique is that 
the Ramsey theory argument (used in the first step) 
only applies if the domain is extremely large. 
Using an enhancement of the technique from \cite{DKN17}, 
we can reduce this to domains of 
doubly exponential size in $k$ (see \Cref{thm:lower-bound-refined}). 
To achieve this, we need to modify our square-diagonal construction 
so that three samples provides little information to the tester 
even when the numerical values of these samples are also revealed. 
To do this, we show that by applying carefully chosen random functions 
to the $x$- and $y$- coordinates, we can effectively obscure almost 
all non-order-based information contained in any set of three samples. 
For the univariate case, \cite{DKN17} showed that for two samples, 
applying a random {\em affine} transformation 
can obscure both the difference and the average of a pair of points. 
However, when there are three points $a < b < c$, 
applying an affine transformation preserves the value of $(a-c)/(b-c)$.
Hence, a non-trivial amount of information may be retrieved 
from the tester by computing this quantity, 
even if a random affine transformation is applied.
To address this issue in our two-dimensional setting,
we will apply an exponential function 
$x \mapsto \exp( \exp(\lambda) x )$, 
where $\lambda$ is a carefully chosen uniform variable. 
Then, if $a,b,c$ are not too close, 
$(a-c)/(b-c)$ will be exponentially close 
to $\exp ( \exp( \lambda) \; (a-b) ) = \exp ( \exp( \lambda + \log (a-b) )  )$. 
When $\lambda$ is large compared to $\log(a-b)$, 
the ratio $(a-c)/(b-c)$ will therefore have roughly 
the same distribution 
of outputs, independent of $a, b,c$. 
As a result, the transformation effectively hides any information 
encoded by the ratio $(a-b)/(b-c)$. 
Afterwards, we can mirror the analysis from \cite{DKN17} 
to apply a suitable random affine transformation to hide 
all of the remaining information.
The details of the construction and its analysis can be found in \Cref{sec:domain-optimize}.

%% file: technique.tex
\subsection{Basic Notation}
For $n \in \Z_+$, we denote $[n] \eqdef \{1, \ldots, n\}$.
We will use $\mathbb S_m$ for the set of all permutations over $m$ distinct 
elements. Given $m > 0$, we use $\Poi(m)$ to denote the Poisson distribution with mean $m$.

An axis-aligned rectangle $R$ is a set in $\R^d$ that can be represented as 
the product of $d$ intervals $I_1, \cdots, I_d$, i.e., $R = \prod_{i=1}^d I_i$.
Given $x, y \in \R^d$, the axis-aligned rectangle %with corners 
defined by $x,y$ is the set 
$R_{x,y} \eqdef \{  z \in \R^d \mid \min(x_i, y_i) \leq z_i \leq \max(x_i, y_i) \text{ for all } i \in [d]\}$.

We will use $\p, \q$ to denote the probability density functions of our
distributions (or probability mass functions for discrete distributions). 
For discrete distributions $\p, \q$ over $[n]$, 
their $\ell_1$ and $\ell_2$ distances are 
$\snorm{1}{\p -\q} \eqdef 
\sum_{i=1}^n \abs{\p(i) - \q(i)}$ and 
$ \snorm{2}{ \p - \q }
\eqdef 
\sqrt{ \sum_{i=1}^n \lp( \p(i) - \q(i)  \rp)^2}
$.
For density functions
$ \p, \q : \R^d \mapsto \R_{+} $, 
we have 
$
\snorm{1}{\p - \q}
\eqdef \int_{ \R^d } \abs{\p(x) - \q(x)} dx
$.
The total variation distance between
distributions $\p, \q$ is defined to be
$ \dtv(\p, \q) = \frac{1}{2} \snorm{1}{\p - \q}$.
Let $R \subset \R^d$ be a subset of the domain of $\p$. 
We denote by $\p_{|R}$ the conditional distribution of $\p$ restricted to 
$R$, i.e.,  $\p_{|R}( x ) = \p(x) / \int_{R} \p(x) dx$ for $x \in R$.
Let  $ \mathcal R = \{R_1, \cdots, R_k \}$ be a collection of disjoint sets
$R_i \subseteq \R^d$. The \emph{reduced measure} corresponding to $\p$ 
and $\mathcal R$, which we denote by $\p^{ \mathcal R}$, 
is a discrete measure on $[k]$ defined as
$\p^{\mathcal R}_i = \p(R_i)$ for $i \in [k]$.
% Let $\mathcal J_k$ 
% be the collection of all sets  of disjoint axis-aligned rectangles of size $k$ in $\R^d$.
% For two probability distributions $\p, \q: \R^d \mapsto \R_{\geq 0}$ and $k \in \mathbbm Z^+$, 
% we define the \emph{Generalized $\mathcal A_k$-distance} between $\p$ and $\q$ as
% $
% \snorm{\mathcal A_k}{ \p - \q }
% := \max_{ \mathcal R = \{ R_i \}_{i=1}^k \in \mathcal J_k } \snorm{1} {\p^{\mathcal R} - \q^{\mathcal R}}
% =
% \max_{ \mathcal R = \{ R_i \}_{i=1}^k \in \mathcal J_k }  
% \sum_{i=1}^k \abs{\p(R_i) - \q(R_i)}.$

\subsection{Organization} \label{ssec:org}
The structure of this paper is as follows: In Section~\ref{sec:alg} we develop the 
analysis tools required to design and analyze our closeneness tester. 
Section~\ref{sec:lb} contains our sample complexity lower bound. 
In Section~\ref{sec:open}, we provide some conclusions and open problems. 

%% file: upper.tex
\section{Closeness Testing Algorithm} \label{sec:alg}

In this section, we describe and analyze our multivariate $\Ak$-closeness tester.
The structure of this section is as follows: 
In~\Cref{sec:tester}, we present our 
algorithm and its analysis. The proof of our main structural result 
(\Cref{lem:rectangle-discrepancy}) which relies on Ramsey theory
is given in \Cref{sec:discrepancy}. In \Cref{sec:light-bin-l2-test}, 
we describe and analyze our new closeness tester for discrete distributions 
which detects discrepancies supported on a small number of light domain elements 
(\Cref{lem:small-support-l2}).
Finally, \Cref{sec:application} describes some applications of our $\Ak$ 
closeness tester to test closeness of 
structured distributions under the total variation distance.

\subsection{The Tester and its Analysis} \label{sec:tester}

We start with an overview of our algorithmic approach followed
by a detailed pseudo-code and analysis of our tester.

\paragraph{Overview of Algorithmic Approach} 
Let $\mathcal{R} = \{R_i\}_{i=1}^k$ be a collection of $k$ disjoint rectangles which 
witness the $\Ak$-distance between $\p, \q$\footnote{Note that such a collection is not necessarily unique.}.
The main technical obstacle of $\mathcal A_k$ closeness testing is 
that the algorithm does not know (a priori) such a collection of rectangles. To 
circumvent this issue, we draw samples from $\p, \q$ and use the obtained 
information to construct a set of rectangles that capture a non-trivial amount 
of discrepancy between the underlying distributions. A natural way to construct 
our rectangles is as follows. Given a collection of sample points from $\p$ and $\q$,
we group these points into disjoint pairs and make our rectangles be those 
defined by the corresponding pairs. 

Note that the number of ways to group the sample points into disjoint pairs scales exponentially with the number of samples drawn.
But before we discuss how the grouping is done in our algorithm,
we need to prove that this approach can work {\em in principle}, 
i.e., that if one draws sufficiently many samples,  
there exists a {\em small} set of rectangles (each defined by pairs of sample points)
that capture enough discrepancy between $\p$ and $\q$.

Let $x,y$ be two samples drawn from the mixture $(1/2)(\p + \q)$. Conditioned on the event 
that $x,y$ both land in some rectangle $R \in \mathcal{R}$ of the witnessing partition, 
we show that $R_{x,y}$ --- the rectangle defined by $x,y$ --- will in expectation 
capture a non-trivial amount of the discrepancy in $R$. The formal statement is 
specified in \Cref{lem:rectangle-discrepancy} and its proof is given in Section~\ref{sec:discrepancy}.

By applying \Cref{lem:rectangle-discrepancy} to each rectangle $R_i \in \mathcal{R}$, one can show the existence of a collection of $k' = O(k)$ rectangles 
defined by the sample points which capture enough discrepancy between $\p, \q$ 
(\Cref{lem:grid-point-discrepancy}).
It then remains to find these rectangles and 
invokes an appropriate closeness testing procedure 
to compare the probability mass of $\p, \q$ on them.
Trying all possible collections of rectangles defined by the sample points is certainly not
computationally feasible. 
Even worse, the natural analysis of this brute-force strategy 
would require one to union-bound the failure probabilities 
of the closeness testing steps executed on each possible 
collection of rectangles. 
As the number of possible collections scales exponentially 
with the size of the collection, i.e., $k'$, each individual 
closeness testing routine is only allowed to fail 
with exponentially small probability, 
making the sample complexity of this approach at least linear in $k$.

We instead follow an approach inspired by the idea 
of a \emph{Good Oblivious Covering} in \cite{diakonikolas2019testing}.
In particular, we consider a sub-collection of rectangles defined 
by the coordinates of the sample points that 
form a nice ``cover'' of all possible such rectangles. 
We then proceed to define the notion of ``induced'' distributions 
of $\p, \q$ on the cover such that the two corresponding induced distributions 
have large $\ell_2$-discrepancy supported on a small number of domain elements
if and only if there exists a collection of rectangles 
defined by the sample points over which the probability mass of $\p, \q$ differ significantly. 
Then, applying a novel variant of the $\ell_{1, k}$-tester from 
\cite{DKN17}
(see \Cref{lem:small-support-l2}) 
yields our final tester.

\bigskip

\noindent We are now ready to proceed with the details of the proof.

\paragraph{Discrepancy from Random Points}
Let $R$ be an axis-aligned rectangle such that $\p(R)$ and $\q(R)$ 
differ substantially and $x, y$ be sample points drawn 
from $ \frac{1}{2} (\p + \q)_{|R} $, the uniform mixture distribution 
between $\p, \q$ restricted to $R$.
We consider the rectangle defined by $x,y$, which we denote by $R_{x,y}$.
Our main structural result, serving as the direct motivation for our algorithm, shows that 
$R_{x,y}$ captures non-trivial amount of discrepancy between $\p$ and $\q$ with non-trivial probability.
\begin{proposition}[Random Point Discrepancy]
\label{lem:rectangle-discrepancy}
{Let $\p, \q$ be distributions on $\R^d$ and $R$ be an
axis-aligned rectangle $R \subset \R^d$ satisfying
$ \abs{\p(R) - \q(R)} \geq \eps ( \p(R) + \q(R) ) $}.
Let $x, y$ be random points sampled from $(1/2) (\p + \q)_{|R}$. 
Then there exists a number
$\alpha_d =  C d^2 2^{2^{d+1}}$,  
for some sufficiently large universal constant $C>0$,
such that
$ \E \lp[ \abs{\p(R_{x,y}) - \q(R_{x,y})} \rp] 
\geq \eps^{ \alpha_d }  ( \p(R) + \q(R) ) $.
\end{proposition}

The proof of~\Cref{lem:rectangle-discrepancy} makes essential use of 
Ramsey theory and is one of the main technical contributions of this 
work. We defer its proof to \Cref{sec:discrepancy}. 

Here we comment on the quantitative aspects of this result. 
Specifically, it is not clear whether the $\eps^{\alpha_d}$ multiplicative factor
in the right hand side of the final inequality is best possible. It is a plausible conjecture that 
the optimal dependence is $\poly(\eps)$ --- independent of the dimension $d$ (see \Cref{q:eps-dependency}). 
Such an improvement would directly improve the sample complexity of our closeness tester, 
as a function of $\eps$.

\paragraph{Existence of Witnessing Grid-aligned Rectangles}

We begin with an assumption that simplifies our analysis: 
the cumulative density function of each coordinate of
$\p$ or of $\q$ is continuous. 
We will eventually remove the assumption in the proof of our main theorem.
Suppose that $\snorm{\Ak}{\p - \q} \geq \eps$.
Then there exists a collection of $k$ disjoint axis-aligned rectangles $
R_1, R_2, \cdots, R_k \subseteq \R^d$ such that 
$\sum_i \abs{\p( R_i ) - \q(R_i)} \geq \eps$. 
By~\Cref{lem:rectangle-discrepancy}, if two sample points $x, y$ happen to land 
in the same rectangle $R_i$, the rectangle $R_{x,y}$ they define 
will capture a non-trivial fraction of discrepancy in $R_i$.
For this reason, we restrict our attention to rectangles 
lying on the \emph{sample-point grid} defined below.

\begin{definition}[Sample-Point Grid] \label{def:grid}
Let $S = \{ x^{(1)}, \cdots, x^{(m)} \} \subset \R^d$ 
be a set of sample points such that no two points overlap in any of their coordinates, 
i.e., $x^{(i)}_{\ell} \neq x^{(j)}_{\ell}$ for all $i \neq j \in [m]$ and $\ell \in [d]$.
The {\em sample-point grid} {$G_S$} (with respect to $S$)
is the set of all points $z \in \R^d$ such that the $i$-th coordinate $z_i$ 
is chosen from the set $\{ x^{(1)}_i, \cdots, x^{(m)}_i \}$. 
Given an axis-aligned rectangle $R$, we say that $R$ is a {\em grid-aligned rectangle} 
with respect to $G_S$ if all its vertices are grid-points from $G_S$.
\end{definition}

Let $G_S$ be a sample-point grid with respect to a collection 
of sufficiently many i.i.d.\ samples from $(1/2)(\p + \q)$.
We first show that, with high constant probability,
there exist $O(k)$ many rectangles aligned with $G_S$ 
that capture enough discrepancy between $\p, \q$ in $\ell_2$ distance.

\begin{lemma}[Existence of a Small Set of Witnessing Grid-aligned Rectangles]
\label{lem:grid-point-discrepancy}
Let $\alpha_d>0$ be as defined in \Cref{lem:rectangle-discrepancy}.
Let $\p, \q$ be distributions over $\R^d$ satisfying
$\snorm{\Ak}{\p - \q} \geq \eps$.
Let $S$ be a set of $\Poi(m)$ i.i.d.\ samples from $(1/2)(\p + \q)$, 
where $k > m \geq C \sqrt{k}/(\eps/4)^{2 \alpha_d}$, 
for some sufficiently large universal constant $C>0$, 
and $G_S$ be the sample-point grid defined by these points.
With probability at least $9/10$, there exist $ k' \leq 3k$ disjoint grid-aligned rectangles $\tilde R_1, \cdots, \tilde R_{k'}$ with respect to $G_S$ satisfying the following:
\begin{itemize}
\item[(i)] $ \p(\tilde R_i) + \q(\tilde R_i)   \leq O(1/k)$ for all $i \in [k']$, and
\item[(ii)] $\sum_{i=1}^{k'} \lp(\p( \tilde R_i) - \q(\tilde R_i)\rp)^2 
\geq \Omega( (\eps/4)^{2 \alpha_d} \; m^2 / k^3 )$.
\end{itemize}
\end{lemma}
\begin{proof}
Let $R_1, \cdots, R_k$ be a collection of $k$ axis-aligned rectangles 
which realize the $\Ak$-distance between $\p, \q$. Namely, it holds
$\snorm{\Ak}{\p - \q} = \sum_{i=1}^k \abs{ \p(R_i) - \q(R_i) }.$
For convenience, for each rectangle $R_i$, we will denote 
$v_i \eqdef ( \p(R_i) + \q(R_i) ), \eps_i \eqdef \abs{ \p(R_i) - \q(R_i) } / v_i$.
We first perform some preliminary simplifications 
to make sure that 
$v_i = O(1/k)$ and $\eps_i \geq \eps / 4$.
Given $ v_i > 1/k$, we can subdivide $R_i$ into $ \floor{v_i \; k}$ 
sub-rectangles evenly along the first coordinate 
according to {the cumulative density function of the first coordinate of $(1/2)(\p + \q)$.}
% As a result, 
% each of the sub-rectangle will have mass upper bounded by $2/k$ and we increase the number of rectangles by at most $2k$.
We next discard any rectangles $R_i$ such that $\eps_i < \eps / 4$, 
which leads to us losing 
at most $\sum_{i} (\p(R_i) + \q(R_i)) \eps/4 \leq \eps/2$ discrepancy. 
In summary, after these operations, we will have a collection of 
$\tilde k \leq 3k$ rectangles $R_1, \cdots, R_{\tilde k}$ such that 
for each rectangle $R_i$ in the collection 
we have that $v_i \leq 1/k$, $\eps_i \geq \eps/4$, and 
$\sum_{i=1}^{\tilde k} v_i \eps_i \geq \eps/2$.

Let $S$ be the set of $\Poi(m)$ many \iid samples drawn 
and $G_S$ be the corresponding sample-point grid.
We define the random variable $Y_i$ as follows:
if exactly two samples $x,y \in S$ fall in the same rectangle $R_i$ for $i \in [\tilde k]$, 
then $Y_i = \lp( \p(R_{x,y}) - \q(R_{x,y}) \rp)^2 $; otherwise, $Y_i = 0$.
By the definition of $Y_i$, we know that if $Y_i > 0$, then there exists some rectangle $\tilde R \subset R_i$ aligned with $G_S$ such that 
$Y_i = \lp( \p( \tilde R ) - \q(  \tilde R ) \rp)^2$.
Hence, $\sum_{i=1}^{\tilde k} Y_i$ is always a lower bound on 
the discrepancy collected by the best collection of at most 
$\tilde k \leq 3k$ rectangles aligned with the grid $G_S$ 
for any instance of the set $S$.
Consequently, to prove the lemma, it suffices to show that 
$\sum_{i=1}^{\tilde k} Y_i \geq
\Omega \lp( (\eps/4)^{2 \alpha_d } \; m^2 / k^3 \rp)$ 
with probability at least $9/10$.

Consider the event $E_i$ that exactly two sample points land inside $R_i$.
Then it is easy to see that
$$
\Pr[ E_i] = 
\Pr \lp[\Poi(m v_i / 2) = 2  \rp]
= \Theta(1) \, (m v_i)^2 \;.
$$
Conditioned on the event $E_i$, $Y_i$ is equal to 
$\lp( \p(R_{x,y}) - \q(R_{x,y}) \rp)^2$, 
where $x,y$ are two random points from $ \frac{1}{2} (\p + \q)_{|R_i}$.
By our preliminary simplification, we have that 
$\abs{\p(R_i) - \q(R_i)} \geq (\eps/4) \lp(\p(R_i) + \q(R_i) \rp)$.
Hence, applying \Cref{lem:rectangle-discrepancy}, we obtain 
$$
\underset{{ x,y \sim \frac{1}{2} (\p + \q)_{|R_i} }} \E \lp[ \abs{\p(R_{x,y}) - \q(R_{x,y})} \rp]
\geq 
\eps_i^{\alpha_d} v_i \;.
$$
Combining this with Jensen's inequality then gives that
$$
\underset{{ x,y \sim \frac{1}{2} (\p + \q)_{|R_i} }} \E \lp[ \lp(\p(R_{x,y}) - \q(R_{x,y}) \rp)^2 \rp]
\geq 
\lp( \underset{{ x,y \sim \frac{1}{2} (\p + \q)_{|R_i} }} \E \lp[ \abs{\p(R_{x,y}) - \q(R_{x,y})} \rp] \rp)^2
\geq
\eps_i^{2\alpha_d} v_i^2 \;.
$$
Since $Y_i$ conditioned on the event $E_i$ is distributed 
as $\lp( \p(R_{x,y}) - \q(R_{x,y}) \rp)^2$ and $Y_i$ is always non-negative, 
we thus have
\begin{align} \label{eq:single-yi-bound}
\E[Y_i] 
\geq
\E[Y_i | E_i] \; \Pr[E_i]
\geq \Omega(1) \; (m \; v_i)^2 \; \eps_i^{ 2  \alpha_d }  v_i^2
\geq \Omega(1) \; \eps_i^{ 2  \alpha_d }
\; m^2 \; v_i^4 \;.    
\end{align}
Summing over all $Y_i$'s, we obtain 
\begin{align*}
\sum_{i=1}^{\tilde k} \E[Y_i] 
\geq 
\Omega(1) \, 
\sum_{i=1}^{\tilde k}  \eps_i^{ 2  \alpha_d -4 }
\, m^2 \; (v_i \eps_i) ^4
\geq 
\Omega(m^2) \, (\eps/4)^{ 2  \alpha_d -4 }
\sum_{i=1}^{\tilde k}  (v_i \eps_i)^4
\geq 
\Omega(m^2) (\eps/4)^{ 2  \alpha_d } / k^3 \, ,
\end{align*}
where the first inequality uses (\Cref{eq:single-yi-bound}), 
in the second inequality we bound from below $\eps_i$ by $\eps/4$, 
and in the third inequality we use the fact that 
$\sum_{i=1}^{b} a_i^4$ subject to $\sum_i a_i = A$, $a_i \geq 0$, 
is minimized at $a_i = A/b$.

% and the axis-aligned rectangle $R_{x,y}$ satisfies $ \abs{ \p(R_{x,y}) - \q(R_{x,y}) }  \geq \eps^{ \alpha_d } \; \lp( \p(R_i) + \q(R_i) \rp)$.
% As a result of Poissonization, the events $E_i$ for $i \in [k]$ are independent of each other. 
% Moreover, if $E_i$ happens, we then have $Y_i \geq \alpha_{d, \eps}^2 v_i^2$.
% we know there is a rectangle $R \subseteq U_i$ satisfying
% (i) $\p(R) + \q(R) \leq \p(U) + \q(U) \leq 2/k$.
% (ii) $\abs{\p(R) - \q(R)}^2 \geq  \alpha_{d,\eps}^2 \; v_i^2$.
% We first lower bound the probability that $E_i$ happens, which would then give us a lower bound of $\E[Y_i]$.
% In particular, condition on the event that two samples $x,y$ all in $U_i$, by Lemma \ref{lem:rectangle-discrepancy}, we have that with probability at least $\alpha_{d, \eps}$, 
% $$ \abs{\p(R_{x,y}) - \q(R_{x,y})} \geq \alpha_{d, \eps} \; ( \p(U_i) + \q(U_i) )
% .$$
% Hence, 
% $$
% \Pr[E_i] \geq \Poi( m \; v_i/2, 2) \; \alpha_{d, \eps} \geq \Omega(1) \; \alpha_{d, \eps} \; (m \; v_i)^2.
% $$
% On the other hand, $E_i$ only happens if two samples fall in $U_i$, which easily gives $\Pr[E_i] \leq O( m^2 v_i^2 )$.

On the other hand, since $Y_i$ is defined to be non-zero 
only when there exist two points landing in $R_i$, 
and takes values at most $v_i^2$, we have that
$
\Var[Y_i] \leq \Pr[ E_i ] \; v_i^4
= O(1)   m^2  v_i^6 \,.
$
Furthermore, since the $Y_i$'s are independently distributed, 
it follows that
\begin{align*}
    \Var \lp[ \sum_{i=1}^{\tilde k} Y_i \rp]
    = \sum_{i=1}^{\tilde k}
    \Var[ Y_i ]
    \leq 
    O(1)  \; m^2  \sum_{i=1}^{\tilde k} v_i^6
    \leq O( m^2 / k^5 ) \;,
\end{align*}
where in the last inequality we use that $v_i \leq O(1/k)$.
% \begin{align*}
% \E \lp[ \sum_{i=1}^{k}  Y_i \rp] &\geq 
% \Omega(1) \; \eps^{2 \alpha_d}  \; m^2 \; \sum_{i=1}^k v_i^4 \, , \\
% \Var \lp[ \sum_{i=1}^{k} Y_i \rp] &\leq O(1) \; m^2 \; \sum_{i=1}^k v_i^6.
% \end{align*}
We then have that 
$ 
\Var[ \sum_{i=1}^{\tilde k} Y_i]  \leq (1/20) 
\lp(  
\E \lp[ \sum_{i=1}^{\tilde k} Y_i \rp] \rp)^2
$
as long as $m \geq C  \sqrt{k} / (\eps/4)^{2 \alpha_d}$ 
for some sufficiently large universal constant $C>0$. 
Then, by Chebyshev's inequality, it follows that
$$
\Pr \lp[ 
\sum_{i=1}^{\tilde k} Y_i \geq \Omega\left( (\eps/4)^{2 \alpha_d} \; m^2/k^3 \right)
\rp]
\geq 9/ 10 \;.
$$
% Since $v_i = O(1/k)$ and $\sum_{i} v_i \leq 1$, it is easy to see that $ \sum_{i=1}^k v_i^4 \geq O \lp( 1/k^3 \rp)$.
This concludes the proof of \Cref{lem:grid-point-discrepancy}.
\end{proof}

\paragraph{Existence of Good Grid Covering}
By \Cref{lem:grid-point-discrepancy}, 
there exist $O(k)$ grid-aligned rectangles that capture 
$\Omega_{d, \eps}(m^2/k^3)$ discrepancy between $\p, \q$.
A naive tester may proceed as follows: 
Choose a set of $k'= O(k)$ disjoint rectangles aligned with the sample-point grid, 
and then perform closeness testing between the reduced distributions of $\p, \q$ 
on the chosen rectangles. Then, with non-trivial probability, 
the chosen rectangles will capture enough discrepancy between $\p, \q$, 
and a standard closeness tester would suffice.
Unfortunately, the number of ways to choose $k'$ disjoint 
grid-aligned rectangles from a grid containing $m^d$ grid points 
is at least $m^{ \Omega(d \; k') }$.
If we were to try all possible collections of $k'$ disjoint grid-aligned rectangles, 
the resulting tester would likely be inefficient, 
as discussed in our techniques overview (\Cref{sec:techniques}), 
in terms of both sample complexity and computational complexity.
To circumvent this issue, we 
will instead consider a carefully chosen subset 
of all grid-aligned rectangles with respect to the sample-point grid 
such that any grid-aligned rectangle can be decomposed into the union 
of a small number of rectangles from the family. 
Moreover, the subset is carefully constructed 
to have the subtle property that any point $x \in \R^d$ 
is contained in a small number of rectangles from the subset.
This leads us to the concept of \emph{Grid Covering}, 
which is based on the idea of \emph{Good Oblivious Covering} 
(Definition 2 from \cite{diakonikolas2019testing}).

\begin{definition}[Grid Covering] \label{def:grid-cover}
Let $m$ be a power of $2$ and $S$ be a set of $(m+1)$ points in $\R^d$ and $G_S$ be the corresponding sample-point grid. 
A {\em grid covering} is a family of rectangles aligned with the sample-point grid, 
which we denote by $\mathcal F(G_S)$, satisfying the following:
\begin{itemize}[leftmargin=*]
    \item Any rectangle aligned with the grid can be represented as the union of at most 
    $2^d \log^d m$ disjoint rectangles from $\mathcal F(G_S)$.
    \item Any point in $\R^d$ is contained in exactly $\log^d m $ rectangles.
\end{itemize}
\end{definition}

\noindent With a construction similar to that in \cite{diakonikolas2019testing}, 
we show that a Grid Covering always exists.

\begin{lemma}[Existence of Grid Covering] \label{lem:gc-exist}
Let $m$ be a power of $2$, {$S$ be a set of $(m+1)$ points from $\R^d$ and $G_S$ be the corresponding sample-point grid.}
Then there exists a grid covering $\mathcal F(G_S)$.
\end{lemma}
\begin{proof}
For each coordinate $j \in [d]$, let $x^{(1)}_j, \cdots, x^{(m+1)}_j$ be the 
$j$-th coordinates of the samples collected sorted in increasing order.
We will refer to these numbers as the ``grid values''.
For each $i \in  [\log m]$, we will define $\mathcal I_{j, i}$ 
as the partition of the interval $[x^{(1)}_j, x^{(m)}_j]$ into $2^i$ many 
sub-intervals such that each sub-interval in the partition 
contains an equal number of grid values.
Then the rectangles in $\mathcal F(G_S)$ are those of the following form: 
for $j \in [d]$, an interval $I_j \in \bigcup_i \mathcal I_{j, i}$ 
is chosen and the rectangle is simply the product of the $d$ selected intervals $I_j$.

Then it is easy to see that for any value $ z \in [ x^{(1)}_j, x^{(m+1)}_j ] $, $z$ 
is within $\log m$ intervals from $\bigcup_i \mathcal I_{j, i}$ 
(one interval from each partition). As a result, any point in $\R^d$ 
is within $ \log^d m $ rectangles from $\mathcal F(G_S)$.

Let $R$ be a grid-aligned rectangle that is the product of the intervals 
$I_1, \cdots, I_d$. Notice that the interval $I_j$ can be decomposed 
into at most $2  \log m$ intervals from $\bigcup_i \mathcal I_{j, i}$ 
(at most $2$ intervals from each partition $\mathcal I_{j,i}$). 
Thus, $R$ can be decomposed into at most $2^d \log^d m$ rectangles from $\mathcal F(G_S)$.
This completes the proof. 
\end{proof}

We next define the notion of the \emph{induced distribution} 
of $\p, \q$ on $\mathcal F(G_S)$.
\begin{definition}[Induced Distribution] \label{def:ind-d}
Given a ditribution $\p$ on $\R^d$ and a family of sets $\mathcal F$ 
{whose elements are non-empty sets in $\R^d$ that are not necessarily disjoint}, 
the {\em induced distribution $\p^{\mathcal F}$} 
is defined as follows. To draw a random sample from $\p^{\mathcal F}$, 
one first draws a random sample $x$ from $\p$. 
{If $x$ does not belong to any set in $\mathcal F$, we return the special element 
$\emptyset$. Otherwise, we return a uniformly random set $S \in \mathcal F$ such that $x \in S$.}
\end{definition}
% Let $R$ be a rectangle 
% Suppose the grid $G$ captures $\gamma$ discrepancy between $\p, \q$ and 
% $\mathcal F(G_S)$ is its grid covering. Then, it is easy to see that $\snorm{1}{ \p^{\mathcal F(G_S)} - \q^{\mathcal F(G_S)} } \geq \gamma / \log^d m$ since $ \abs{ \p^{\mathcal F(G_S)} (R) - \q^{\mathcal F(G_S)}} (R) \geq \log^{-d} m \; \abs{ \p(R) - \q(R) }$ for any rectangle $R \in \mathcal F(G_S)$.
Notice that for a rectangle $R \in \mathcal F(G_S)$, we have that 
$\p^{\mathcal F(G_S)}(R) = \p(R) / \log^d m$, 
since each point appears in exactly $\log^d m$ rectangles from $\mathcal F(G_S)$. 
This then allows us to show that the $\ell_2$-discrepancy between 
the induced distributions $\p^{\mathcal F(G_S)}, \q^{\mathcal F(G_S)}$ 
must be non-trivial if the grid $G$ satisfies the 
conclusion in \Cref{lem:grid-point-discrepancy}. 
Specifically, we show:
\begin{lemma} \label{lem:induced-discrepancy}
{Let $m$ be a power of $2$, $S$ be a set of $(m+1)$ points in $\R^d$, 
and $G_S$ be the corresponding sample-point grid.}
Moreover, suppose that the conclusion of \Cref{lem:grid-point-discrepancy} holds for $G_S$.
Then there exists a subset of rectangles $H \subset \mathcal F(G_S)$ 
such that the following conditions hold: 
\begin{itemize}
\item[(i)] $\abs{H} \leq 3k \; 2^d \log^d m$, 
\item[(ii)] $\p^{\mathcal F(G_S)}(R) + \q^{\mathcal F(G_S)}(R) \leq  O\lp( \log^{-d}(m) / k \rp)$ for all $R \in H$, and 
\item[(iii)]
$\sum_{ R \in H } \lp( \p^{\mathcal F(G_S)}(R) - \q^{\mathcal F(G_S)}(R) \rp)^2 
\geq \Omega(1) \; 
2^{-d} \; \log^{-3d} (k) \; (\eps/4)^{2\alpha_d } \; m^2 / k^3 \;.
$
\end{itemize}
\end{lemma}
\begin{proof}
Since we assume that the conclusion 
in \Cref{lem:grid-point-discrepancy} is satisfied, 
there exist $k' \leq 3k$ 
many grid-aligned rectangles  $R_1, \cdots, R_{k'}$ (with respect to $G_S$) satisfying 
$  \p(R_i) + \q(R_i)   \leq O\lp( 1/k \rp)$ for all $i \in [k']$, and
\begin{align} \label{eq:initial-discrepancy}
\sum_{i=1}^{k'} \abs{\p(R_i) - \q(R_i)}^2 \geq \Omega( (\eps/4)^{2 \alpha_d}  \; m^2 / k^3 ).
\end{align}
% Suppose the algorithm does not add any extra points to $S$. The claim follows from \Cref{lem:grid-point-discrepancy}. 
% Now, it is easy to see that $R_1, \cdots, R_k'$ are still grid-aligned rectangles with the extra points. Hence, the claim follows.
% Let $\mathcal F(G_S)$ be a grid-covering.
% Recall that the induced distributions $\p^{\mathcal F(G_S)}, \q^{\mathcal F(G_S)}$ are distributions over rectangles in $\mathcal F(G_S)$. 
% We will show that there exists a set of rectangles $H \subset \mathcal F(G_S)$ with $\abs{H} \leq k \; 2^d \log^d m$ such that
% $ \p^{\mathcal F(G_S)}(R) + \q^{\mathcal F(G_S)}(R)   \leq O\lp( 1/k \rp)$ for $R \in H$ and
% \begin{align} \label{eq:induced-discrepancy}
% \sum_{R \in H} \lp( \p^{\mathcal F(G_S)}(R) - \q^{\mathcal F(G_S)}(R) \rp)^2
% \geq \Omega(1) \; 
% 2^{-d} \; \log^{-2d}(m) \; \eps^{2 \alpha_d}  \; m^2 / k^3.    
% \end{align}
% $\p^{\mathcal F(G_S)}, \q^{\mathcal F(G_S)}$ have non-trivial $\ell_2$ discrepancy.
By the definition of the grid covering, 
each $R_i$ can be decomposed into at most $2^d \; \log^d(m)$ rectangles from $\mathcal F(G_S)$.
Let $H_i$ be the set of rectangles in $\mathcal F(G_S)$ into which $R_i$ is decomposed.
We will consider $H = \bigcup_i H_i$.
It is clear that $\abs{H} \leq  3k \; 2^d \; \log^d(m)$, which shows (i).
Moreover, by the definition of the induced distribution, 
for any rectangle $R \in \mathcal F(G_S)$, 
we have $\p^{\mathcal F(G_S)}(R) = \log^{-d} (m) \; \p(R) $.
Therefore, for each $R \in H$, it holds 
$ \p^{\mathcal F(G_S)}(R) + \q^{\mathcal F(G_S)}(R) \leq   \log^{-d} (m) \; O (1/k)$, which shows (ii).

It remains to show (iii). 
For each $H_i$, we have
\begin{equation*}
% \label{eq:cover-spread}
\sum_{R \in H_i} \lp( \p(R) - \q(R) \rp)^2
\geq 2^{-d} \; \log^{-d}(m) \; \lp( \p(R_i) - \q(R_i) \rp)^2 \, ,
\end{equation*}
since $ \sum_{R \in H_i} \p(R) = \p(R_i) $ (and the same for $\q$) 
and $\abs{H_i} \leq 2^{d} \; \log^{d} (m) $.
Combining this with the fact that $\p^{\mathcal F(G_S)}(R) = \log^{-d} (m) \, \p(R)$ 
and~\Cref{eq:initial-discrepancy} gives (iii). 
This completes the proof. 
\end{proof}

Unlike the naive testing approach (running an $\ell_1$-closeness tester
on many different pairs of reduced distributions), 
we can now run a closeness tester just on the induced distributions 
$\p^{\mathcal F(G_S)}, \q^{\mathcal F(G_S)}$.
A technical issue is that the domain size of $\mathcal F(G_S)$ is still very large. 
This makes the black-box application of any $\ell_1$-closeness tester sample inefficient.
Instead, we need to leverage the fact that a non-trivial fraction 
of the discrepancy between the two distribution 
is supported on a small number of elements.
Interestingly, a tester with similar guarantees was developed in \cite{DKN17} (see Lemma 2.5). 
However, as is, that tester is not sufficient for our purposes.
More specifically, we essentially need to develop an $\ell_2$-version of it. 
The reason is that we need to distinguish between the cases $\p = \q$ 
versus the case that a non-trivial amount of $\ell_2$-discrepancy is supported 
on a few elements {\em that are themselves not too heavy}.
In particular, using tools developed in \cite{DiakonikolasK16} and \cite{chan2014optimal}, we show the following:
\begin{lemma}
\label{lem:small-support-l2}
Let $\p, \q$ be discrete distributions on $[n]$ and $s \in [n]$.
Given $\eps > 0$ and $\Poi(m)$ many \iid samples from $\p, \q$, 
for $ m = \Theta \lp(  \max \lp( \eps^{-4/3}, \eps^{-2} / \sqrt{s} \rp) \rp)$,
there exists a tester \textsc{ Flatten-Closeness } that distinguishes 
between the following cases with probability at least $9/10$: 
(a) $\p = \q$ versus (b) there exists a set of elements $H$ of size $s$ such that
(i)  $\sum_{i \in H} ( p_i - q_i )^2 \geq \eps^2$
and (ii)  $ \max_{i \in S} (p_i + q_i)/2 \leq 1/s $.
\end{lemma}

The proof of \Cref{lem:small-support-l2} builds on the approach of the 
$\ell_{1,k}$ tester. An important difference is that we 
now need to carefully incorporate the upper bound on the mass 
of the elements witnessing the discrepancy into the analysis. 
We defer the proof to \Cref{sec:light-bin-l2-test}.

We are now ready to present the pseudo-code of our testing algorithm and provide its 
proof of correctness. 

\begin{algorithm}[H]
\caption{Multidimensional $\mathcal{A}_k$ Closeness Tester} \label{alg:multi-ak-test}
    \begin{algorithmic}[1]
    \Require  sample access to $\p, \q$ on $\R^d$; %failure probability $\delta$; 
    accuracy $\eps$.
    \State Set $m \gets C' \; k^{6/7} \; \eps^{ -2 \alpha_d / 3 } \; \log^d(k) \; 2^{d/3}$, 
    where $C'$ is a sufficiently large constant and 
    $\alpha_d$ is defined in \Cref{lem:rectangle-discrepancy}.
    \State Draw $\Poi(m)$ samples from $(1/2) (\p + \q)$ and denote the set of samples by $S$.
    \State Add arbitrarily some distinct points to $S$ such that $|S|$ is a power of $2$. \label{line:addpoint}
    \State Construct the grid $G_S$ (\Cref{def:grid}) and the grid covering $\mathcal F(G_S)$ (\Cref{def:grid-cover}).
    % \State Take $\Poi(m)$ samples from $\lp( \p^{\mathcal F(G_S)} +  \q^{\mathcal F(G_S)} \rp) / 2$. Denote the set of samples as $S$ and $ \p^{\mathcal F(G_S)}_S, \q^{\mathcal F(G_S)}_S$ as the corresponding distribution flattened by the samples from $S$.
    \State Run the $\ell_2$-closeness tester of \Cref{lem:small-support-l2} 
    on the induced distributions $\p^{\mathcal F(G_S)}, \q^{\mathcal F(G_S)}$ 
    with accuracy parameter 
    $ \kappa = c \; 2^{-d} \log^{-3d} k \; (\eps/4)^{2 \alpha_d} \; m^2 / k^3 $ 
    for some sufficient small constant $c>0$.
    \State Accept if that closeness tester accepts; otherwise Reject.
    % (Or normalize it)
    \end{algorithmic}
\end{algorithm}

\begin{proof}[Proof of Upper Bound in Theorem~\ref{thm:main-intro}]
We first present the analysis assuming that 
$\p, \q$ are continuous distributions and in the end give a preprocessing step to make sure the algorithm works for general distributions.
Let $F(G_S)$ be defined as in Algorithm~\ref{alg:multi-ak-test}.
If $\p = \q$, we have $\p^{\mathcal F(G_S)} = \q^{\mathcal 
F(G)}$. Therefore, the tester will accept with probability at least $2/3$ by \Cref{lem:small-support-l2}.

Next, we consider the case $\snorm{\Ak}{\p - \q} > \eps$. 
We claim that with probability at least $9/10$ 
there exist $k' \leq k$ grid-aligned rectangles (with respect to $G_S$) 
such that the conclusion of \Cref{lem:grid-point-discrepancy} is satisfied. 
{Without the operation of adding extra points into $S$ in Line~\ref{line:addpoint} of 
Algorithm~\ref{alg:multi-ak-test}, the claim just follows 
from \Cref{lem:grid-point-discrepancy}.}
Now it is easy to see that $R_1, \cdots, R_{k'}$ 
are still grid-aligned rectangles with the extra points. 
Hence, the claim follows. 

Condition on the event that the conclusion of \Cref{lem:grid-point-discrepancy} holds. 
We can then apply \Cref{lem:induced-discrepancy}, 
which gives us that there exists a set of elements $H \subset \mathcal F(G_S)$ 
such  that 
(i) $\abs{H} \leq 2^d \log^d m \; k$
(ii) $\p(R) + \q(R) \leq  O(1/k)$ for all $R \in H$, and 
(iii)
$$\sum_{ R \in H } \lp( \p^{\mathcal F(G_S)}(R) - \q^{\mathcal F(G_S)}(R) \rp)^2 
\geq \Omega(1) \; 
2^{-d} \; \log^{-3d} (k) \; (\eps/4)^{2\alpha_d } \; m^2 / k^3.
$$
Then, applying \Cref{lem:small-support-l2}, 
gives that the tester rejects with probability at least $9/10$ 
given that 
$$
m \geq C  \max \lp(  {\kappa}^{-2/3}
,{\kappa}^{-1} / \sqrt{k}   \ \rp) \;,
$$
where 
$\kappa =  \Theta(1) \;
2^{-d} \; \log^{-3d} (k) \; (\eps/4)^{2\alpha_d } \; m^2 / k^3 $ 
and $C>0$ is a sufficiently large constant.
% Substituting the bound $m = k^{6/7} \; \eps^{-c_d'} \; \log^{2d}(k)$ then gives us $m' = C' \; m$ for some sufficiently large constant $C$.
One can verify that
$
m = 
C' \; k^{6/7} \; \eps^{ -2 \alpha_d / 3 } \; \log^d(k) \; 2^{d/3}
$
suffices, 
where $C'$ is a sufficiently large constant. 
% solving out $m, m'$ then gives us the entire tester succeeds with probability at least $2/3$ if we set
% $$
% m = m' = 
% k^{6/7} \; \eps^{- c_d'} \; \log^{2d}(k)
% $$
% for some number $c_d'$ that depends only on $d$.

Now let us relax the assumption that the marginal distributions 
of $\p, \q$ in each coordinate have continuous cumulative density functions.
We begin with the observation that the algorithm's output 
essentially depends only on the order information of the sample points. 
That is, given two different sets of samples 
$S = \{ x^{(1)}, \cdots, x^{(m)} \}, \tilde S = \{ \tilde x^{(1)}, \cdots,   \tilde x^{(m)}\}$ 
such that the relative orders of $x^{(1)}_j, \cdots, x^{(m)}_j$ and 
$\tilde x^{(1)}_j, \cdots, \tilde x^{(m)}_j$ are the same for each coordinate $j \in [d]$, 
the output of the algorithm will always be the same.

Based on this observation, we know that the algorithm will satisfy 
the same guarantee if we give it only the ``rank'' information of the samples. 
For $j \in [d]$, we sort $x^{(1)}_j, \cdots, x^{(m)}_j$ in increasing order.  
% When there are ties, we break ties uniformly at random.
We will denote by $\pi(j)_i$ the rank of $x^{(i)}_j$ in the sorted sequence.
Then, for each $i \in [m]$, we replace the original sample 
with the new sample $\hat x^{(i)}$ defined as $\hat x^{(i)}_j = \pi(j)_i$. 

If the marginal distributions of $\p$ or $\q$ are not continuous, 
we may observe multiple samples sharing the same value at some coordinates.
Then, when computing the rank information of the samples, 
we will break ties uniformly at random.
Now consider the distributions $\p', \q'$ obtained 
by stretching any point-mass of their marginal distributions 
at any coordinate into an interval.
If the algorithm takes samples from $\p', \q'$ instead, 
the guarantees are satisfied, 
since $\p', \q'$ are both continuous distributions and 
$\snorm{\Ak}{\p - \q} = \snorm{\Ak}{\p' - \q'}$. 
On the other hand, the order of samples taken from $\p', \q'$ 
has the same distribution as the order of samples taken from $\p, \q$ 
after we break ties uniformly at random. This then concludes the proof.
\end{proof}

% Let $F: \R \mapsto [0,1]$ be the cumulative mass function of the marginal distribution of $(\p + \q)/2$ along the first coordinate. If $F$ has a jump-point at $x \in \R$, it means that the marginal distribution must have a point mass at $x$. This means the algorithm may observe multiple sample points from $(\p + \q)/2$ that has their first coordinate being $x$. 
% % Notice that our algorithm depends only on the order information of these points. 
% % What we can do is to simply randomize the order of these ``colliding'' points in the first coordinate. 
% Then, the behavior of the algorithm will be identical to the scenario where the point mass at $x$ is stretched into an interval. 
% Hence, we reduce back to the case where the cumulative mass function is continuous.

\subsection{Proof of \Cref{lem:rectangle-discrepancy}} \label{sec:discrepancy}

Let $R$ be an axis-aligned rectangle such that 
$ \abs{\p(R) - \q(R)} \geq \eps ( \p(R) + \q(R) )$. 
Let $x, y$ be samples from $(1/2) (\p + \q)_{|R}$ -- the uniform mixture
of $\p, \q$ restricted to $R$. We want to show that in expectation over $x, y$ the discrepancy $\abs{\p(R_{x,y}) - \q(R_{x,y})}$ is large.

\paragraph{Warm-up: Special case $\p(R) > 0$ and $\q(R) = 0$.}
Towards establishing the desired statement, we first analyze the special 
case that $\p(R) > 0$ and $\q(R) = 0$. The proof for this case also serves as 
intuition regarding why selecting the interval $R_{x,y}$ is a good choice. 

In this case, the discrepancy between $\p$ and $\q$ is simply $\p(R_{x,y})$ --- 
the probability mass of $R_{x,y}$ with respect to $\p$.
Therefore, whether $R_{x,y}$ captures enough discrepancy boils down 
to the following question: {\em Let $x,y$ be random points drawn 
from an \emph{arbitrary} distribution $D$ over $\R^d$. 
What is the minimum amount of mass captured by the rectangle $R_{x,y}$ in expectation?} 
We show the quantity is indeed non-trivial.

Interestingly, the proof of this statement relies on a certain 
generalized version of the famous Erd\H{o}s-Szekeres theorem.
In particular, the generalized Erd\H{o}s-Szekeres theorem  bounds from above 
the minimum length of a sequence consisting of points in $\R^d$ 
such that there exists a subsequence of points that 
is monotonic in each coordinate.
\begin{lemma}
\label{lem:constant-mass}
Let $x,y$ be random samples independently drawn from a distribution $D$ on $\R^d$. 
Then it holds
$
\E_{x,y \sim D} [ D( R_{x,y} ) ] \geq 
\beta_d \, ,
$
where $\beta_d = \lp( 2^{2^{d-1}} + 1 \rp)^{-3}$.
Moreover, there exists a distribution $D$ such that
$\E_{x,y \sim D} [ D( R_{x,y} ) ] \leq \frac{2}{ 2^{2^{d-1}} }$.
\end{lemma}

\noindent We note that the above statement is  
qualitatively nearly tight as a function of $d$.

\begin{proof}[Proof of \Cref{lem:constant-mass}]
To prove the lemma, we make essential use of the following 
generalized version of the Erd\H{o}s-Szekeres theorem proved by De Brujin.
\begin{fact}[De Brujin's Generalized Erd\H{o}s-Szekeres Theorem, see  \cite{kruskal1953monotonic}]
\label{thm:erdos}
Let $\psi(n, d)$ denote the least integer $N$ such that every sequence of points 
$x^{(1)}, \cdots, x^{(N)}$ in $\R^d$ contains a monotonic subsequence 
$x^{(i_1)}, \cdots, x^{(i_n)}$ of length $n$ satisfying the following:
for each coordinate $j \in [d]$, we have either that 
$x^{(i_1)}_j \leq \cdots \leq x^{(i_n)}_j$
or that $x^{(i_1)}_j \geq \cdots \geq x^{(i_n)}_j$. 
Then it holds $\psi(n, d) = (n-1)^{2^d} + 1$.
\end{fact}

\noindent As an immediate corollary, we obtain the following: 
\begin{corollary} \label[corollary]{cor:xyz}
Let $S \subset \R^d$ be a set of points with size $|S| \geq 2^{2^{d-1}}+1$. 
Then there exists a triple $x,y,z \in S$ such that $z \in R_{x,y}$.
Furthermore, there exists a set of points $S$ of size $|S| = 2^{2^{d-1}}$ such that there is no triple $x,y,z \in S$ satisfying $z \in R_{x,y}$.
\end{corollary}
\begin{proof}
Let $S$ be an arbitrary set of $m$ points in $\R^d$ .
Let $x^{(1)}, \cdots, x^{(m)}$ be a sequence of points in $\R^{d-1}$ 
obtained by (1) sorting the points in $S$ based on their first coordinates, 
and (2) throwing away their first coordinates.
Applying \Cref{thm:erdos} with $n = 3$ gives us that we will have 
a sub-sequence
$x^{(i_1)}, x^{(i_2)}, x^{(i_3)}$ such that the points 
are either monotonically increasing or monotonically decreasing 
in each of the $(d-1)$ coordinates \emph{if and only if $m \geq 2^{2^{d-1}}+1$}.
Furthermore, by our construction of the sequence of $x^{(i)}$'s, 
the first coordinates of the corresponding points in $S$ 
are always monotonically increasing.
This concludes the proof.
\end{proof}

It is worth noting that for a set of points $ S \subset \R^d$ 
to contain a triple $x,y,z$ such that $z \in R_{x,y}$, 
the size of $S$ needs to be doubly exponential in $d$; 
this bound is 
tight 
since the corollary 
is essentially equivalent to \Cref{thm:erdos}, 
which is itself quantitatively tight.
Now let $S$ be the set of $2^{2^{d-1}}$ points such that there 
is no triple $x,y,z \in S$ satisfying $z \in R_{x,y}$ 
(\Cref{cor:xyz} ensures the existence of such a set of points).
Let $D$ be the uniform distribution over $S$.
One can see that $D(R_{x,y}) \leq  2 / 2^{2^{d-1}}$ for any $x,y \in S$.
It hence follows that
$
\E_{x,y \sim D} \lp[ D(R_{x,y})  \rp]
\leq 2 / 2^{2^{d-1}} \;,
$
showing that our lower bound is qualitatively tight.

To relate \Cref{lem:constant-mass} to the Generalized Erd\H{o}s-Szekeres 
theorem (\Cref{thm:erdos}), we make the following observations: 
(i) the probability mass of $R_{x,y}$ under $D$ is equal to 
the probability that a third random point $z$ drawn from $D$ 
happens to land in $R_{x,y}$, and 
(ii) drawing three random samples from $D$ is equivalent to 
first drawing $N$ random samples from $D$ and then choosing 
$3$ \emph{distinct} points from these $N$ points uniformly at random.

Let $D_N$ be the empirical distribution obtained after drawing 
$N$  \iid samples from $D$.
The observations above allow us to conclude that
\[
\underset{x,y\sim D}{\E}[D(R_{x,y})] = \underset{x,y,z\sim D}{\Pr}[z \in R_{x,y}] \geq
\underset{ D_N  }\E \lp[
\underset{\underset{\text{without replacement}}{x,y,z\sim D_N}}{\Pr}[z \in R_{x,y}] ] \rp] \;.
\]
If we have $N \geq 2^{2^{d-1}}+1$, \Cref{cor:xyz} guarantees the existence 
of a triple $x,y,z \in D_N$ such that $z \in R_{x,y}$.
Hence, the probability in the last equation above is at least $1/N^3$. 
This completes the proof of \Cref{lem:constant-mass}.
\end{proof}

\paragraph{General Case.} 
We are now ready to handle the general case and complete 
the proof of \Cref{lem:rectangle-discrepancy}.
To do so, we leverage 
the concept of the \emph{discrepancy density} defined below.

\begin{definition}[Discrepancy Density] \label{def:density}
Let $\p,\q$ be distributions over $\R^d$. 
For a set $S \subseteq \R^d$, we define the discrepancy density of $S$ 
with respect to $\p,\q$ as follows: 
\begin{align*}
    \rho(S;\p, \q) \eqdef 2 \; \abs{ \p(S) - \q(S) } / (\p(S) + \q(S) ) \;.
\end{align*}
\end{definition}

The high-level intuition is the following. 
Let $R_{x,y}$ be the axis-aligned rectangle %with corners 
defined by $x,y$, where $x, y$ are independent random samples 
drawn from $(1/2) (\p + \q)_{|R}$. 
By \Cref{lem:constant-mass}, the probability mass of $R_{x,y}$ 
(with respect to the mixture distribution $(1/2)(\p + \q)$)
is a non-trivial fraction of the mass of $R$ in expectation.
If the discrepancy between $\p, \q$ within $R_{x,y}$ is a non-trivial 
fraction of the mass of $R_{x,y}$, we are done. 
Otherwise, if we were to ``remove'' the region $R_{x,y}$ from $R$, 
we would discard about approximately equal amounts of $\p$ mass 
and $\q$ mass. Therefore, the discrepancy density of the remaining space, 
$\rho(R \backslash R_{x,y};\p, \q)$, must have increased. 
We can then carve the remaining space into at most $2d$ many 
axis-aligned sub-rectangles and pick a sub-rectangle 
with significantly higher discrepancy density to restart the process.
When the discrepancy density approaches one, the situation 
qualitatively resembles the special case where $\q(R) = 0$ 
(or $\p(R) = 0$); and if the mass of $R_{x,y}$ is non-trivial, 
the discrepancy captured will also be non-trivial.
The formal proof follows.

% The key is to show the existence of some sub-rectangle $R^* \subset R$ such that (i) $R^*$ has high enough ``discrepancy density'', i.e. $\abs{ \p(R') - \q(R') } / \left( \p(R') + \q(R') \right)$ is closed to $1$ (ii) $R^*$ has non-trivial mass with respect to $(\p + \q)/2$. 
% Without loss of generality, assume $\p(R') > \q(R')$.
% Then, the situation becomes very similar to the case $\q(R') = 0$ and applying \Cref{lem:constant-mass} with $x,y$ drawn from the distribution $D = \frac{1}{2} \lp( \p + \q\rp)_{|R'}$ then gives us $\p \lp( R_{x,y} \rp) - \q \lp( R_{x,y} \rp)$, which is approximately $\p \lp( R_{x,y} \rp)$, must be at least a non-trivial fraction of $\p(R')$.
% Then, $\p(R_{x,y}) + \q(R_{x,y})$ must be a non-trivial fraction of $\p(R')$. 

\begin{proof}[{\bf \em Proof of \Cref{lem:rectangle-discrepancy}}]
For notational convenience, we will denote $D := (1/2) \lp( \p + \q \rp)$.
Let $x,y$ be two sample points drawn from
$D_{|R}$, the restriction of $D$ to $R$. 
Then, by \Cref{lem:constant-mass}, it holds
$$
\E_{ x,y \sim D_{|R} } \lp[  D ( R_{x,y} ) \rp] 
\geq \beta_d \; D(R) \;,
$$
for some $\beta_d$ depending only on $d$.
We will use $E$ to denote the event 
$\lp\{ D ( R_{x,y} ) \geq \beta_d \; D(R)  / 2 \rp\}$.
Then we must have that 
$\underset{x,y\sim  D_{|R}} \Pr \lp[ E \rp] \geq \beta_d / 2$, 
since otherwise $\E_{ x,y \sim D_{|R} } \lp[  D ( R_{x,y} ) \rp]$ 
will be no more than $\beta_d \; D(R)$.

We consider two complementary cases.
First, if
\begin{align} \label{eq:good-case}
\underset{x,y \sim D_{|R} } \E \lp[ \abs{ \p (R_{x,y}) - \q(R_{x,y}) }   \big| E\rp] > \frac{\eps}{2}  \; 
\underset{x,y \sim D_{|R} }
\E \lp[  D ( R_{x,y} ) \big| E \rp] \;,     
\end{align}
we will have
\begin{align*}
\underset{x,y \sim D_{|R} } \E \lp[ \abs{ \p_{|R} (R_{x,y}) - \q_{|R}(R_{x,y}) } \rp]
&\geq 
\underset{x,y \sim D_{|R} } \E \lp[ \abs{ \p_{|R} (R_{x,y}) - \q_{|R}(R_{x,y}) } \big | E \rp] \; \underset{x,y \sim D_{|R} } \Pr[ E ] \\
& \geq \eps/2 \; \E_{ x,y \sim D_{|R} } \lp[  D_{|R} ( R_{x,y} ) \big| E \rp] \; \underset{x,y \sim D_{|R} } \Pr[ E ] \\
& \geq \eps \; \beta_d^2 \; D(R) / 8
\end{align*}
and we are done.

Otherwise, it holds
\begin{align}
\label{eq:other-case}
\underset{x,y \sim D_{|R} } \E
\lp[ \eps/2 \; D \lp( R_{x,y}\rp) - \abs{ \p\lp( R_{x,y} \rp) -\q \lp( R_{x,y} \rp) } \big| E \rp] \geq 0  \;.  
\end{align}
Since we also condition on the event $E$, we know that there exists a 
rectangle $\tilde R \subset  R$  such that
\[
(\eps/2) \; D ( \tilde R ) - \abs{ \p ( \tilde R ) -\q ( \tilde R ) } \geq 0 \;, \,
D(\tilde R) \geq \beta_d D(R)/2 \;.
\]
% Without loss of generality, we will assume that $ \frac{1}{2^{2^d+1}} \leq D_{|R}(\tilde R) \leq  \frac{1}{2^{2^d}} $ since we can always truncate the rectangle to remove some extra mass.
We consider the remaining space $R \backslash \tilde R$. 
We have that its discrepancy density satisfies
$$ \abs{ \p(R \backslash \tilde R) - \q(R \backslash \tilde R) } / D(R \backslash \tilde R) \geq 
\frac{\eps \; D(R) - \eps/2 \; D(\tilde R) }
{ D(R) - D(\tilde R) }
= \eps \; (1 + \gamma) \, ,
$$
where we denote 
$\gamma =  \frac{ D(R) - D(\tilde R) / 2   }{  D(R) - D(\tilde R) } - 1$ 
for convenience. 
{
Notice that
\begin{align*}
\gamma = \frac{ D(R) - D(\tilde R) / 2   }{  D(R) - D(\tilde R) } - 1
= 
\frac{  D(\tilde R) / 2   }{  D(R) - D(\tilde R) }
\geq 
\frac{  (\beta_d /4) \; D(R)   }{  D(R)  } = \beta_d/4 \;,
\end{align*}
where in the inequality above we bound below 
$D(\tilde R)$ by $\beta_d D(R)/2$ by our choice of $\tilde R$. 
This then gives that
$\gamma \geq \beta_d/4$.}
% Thus, in the end, we will have some rectangle $R^*$ such that
% Equation~\Cref{eq:good-case} holds.
% It then remains for us to argue that the mass of $R^*$ is not too small relative to $R$ under $D$.
% We do so by coupling the increment in discrepancy density with the decrement in mass of the target rectangle in every iteration.

It turns out that remaining space $R \backslash \tilde R$, 
can be carved into $2d$ many axis-aligned rectangles.
An illustration of the $d=2$ case is given in \Cref{fig:2d-carve}. 
\begin{figure}
    \centering
    \includegraphics[width=20em]{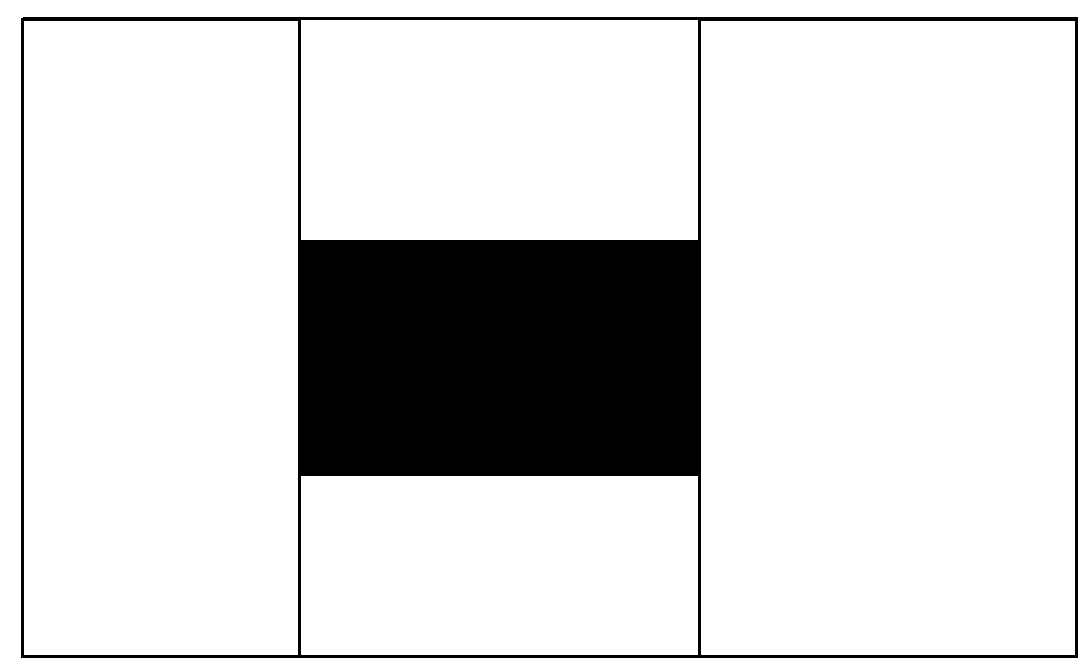}
    \caption{{The black rectangle represents $\tilde R \subset R$ in $\R^2$. One can see that there is a natural way to carve the remaining space $R \backslash \tilde R$ into four axis-aligned rectangles.}}
    \label{fig:2d-carve}
\end{figure}

Specifically, we show the following:
\begin{claim}
\label{clm:carve-complement}
Let $\tilde R \subseteq R$ be an axis-aligned rectangle. 
The set $R \backslash \tilde R$ can be decomposed into 
$2d$ axis-aligned rectangles $R_1, \cdots, R_{2d}$.
\end{claim}
\begin{proof}
The proof proceeds via induction.
The base case ($d=2$) is clear, as shown in \Cref{fig:2d-carve}.
Assume that the statement holds for $d = k$.
We proceed to show that it still holds for $d = k+1$.
Suppose that $R$ is defined by points
$x,y \in \R^{k+1}$ and $\tilde R$ 
is defined by points $\tilde x, \tilde y \in \R^{k+1}$.
We let $R_{2k+1}$ be the rectangle that occupies 
the interval $[x_1, \tilde x_1]$ in the first dimension 
and occupies the same intervals as $R$ in the other dimensions;  
similarly, let $R_{2k+2}$ be the rectangle that occupies the interval 
$[\tilde y_1, y_1]$ in the first dimension and occupies the same 
intervals as $R$ in the other dimensions.
Then the remaining space $R \backslash (R_1 \cup R_2)$ 
lies entirely in the interval $[\tilde x_1, \tilde y_1]$ 
in the first dimension.
We can then discard the first dimension. 
We denote the projection of $R \backslash (R_1 \cup R_2)$ 
into the remaining subspace $\R^k$ as $R'$, 
and the projection of $\tilde R$ as $\tilde R'$.
We can then apply our inductive hypothesis on $R'$ and $\tilde R'$ 
to obtain $2k$ rectangles $R_1', \cdots, R_{2k}'$ 
that live in the subspace of $\R^k$.
Then let $R_i$ to be the product of $R_1'$ 
and the interval $[\tilde x_1, \tilde y_1]$ for $i = [2k]$. 
It is easy to verify that $R_1, \cdots R_{2k+2}$ 
partitions the space $R \backslash \tilde R$.
\end{proof}

We denote the rectangles obtained by applying the above claim to 
$R \backslash \tilde R$ as $R_1, \cdots, R_{2d}$. 
Furthermore, we will denote 
$\gamma_i = \frac{\abs{ \p( R_i ) - \q(R_i) }}{ D(R_i) \; \eps } - 1$.
In other words, 
$\abs{ \p(R_i) - \q(R_i) } / D(R_i) = (1 + \gamma_i) \; \eps_i$.

We claim that there exists a rectangle $R_{i^*}$ in the remaining space such that
\begin{align} \label{eq:i-condition}
D(R_{i^*}) \geq  \frac{ \gamma \; \lp(    D(R) - D(\tilde R) \rp)}{ 2d \; \gamma_{i^*} } \;, 
\gamma_{i^*}
\geq  \frac{ \gamma }{2d}.    
\end{align}
By definition of $D(R_{i})$ and $\gamma_i$, we have
\begin{align}
    \sum_{i=1}^{2d} D(R_{i}) &=  D(R) - D(\tilde R) \;, \label{eq:total-mass} \\
    \sum_{i=1}^{2d} D(R_{i}) \; \eps \; ( 1 + \gamma_i ) &\geq  \lp( D(R) - D(\tilde R) \rp) \; \eps \; (1 + \gamma) \;, \label{eq:discrepancy-sum}
\end{align}
where the first equality follows from the fact
that the rectangles form a partition of the remaining space, 
and the second equality follows from the fact 
that the sum of discrepancies in each rectangle must be 
at least the total discrepancy in the remaining space.

Substituting \Cref{eq:total-mass} into 
\Cref{eq:discrepancy-sum} and simplifying the result gives
$
\sum_{i=1}^{2d} D(R_{i}) \gamma_i \geq 
\lp( D(R) - D(\tilde R) \rp) 
\; \gamma \;.
$
Therefore, there exists $i^*$ such that 
$$ 
D(R_{i^*}) \gamma_{i^*} \geq 
\lp( D(R) - D(\tilde R) \rp)  \gamma / (2d) \;.
$$
{Since $R_{i^*} \subseteq R \backslash \tilde R$}, 
we must have that
$D(R_{i^*}) \leq  D(R) - D(\tilde R) $, 
which implies that
$ \gamma_{i^*} \geq   \gamma / 2d$.
On the other hand, we also have that
$$
D(R_{i^*}) \geq 
\frac{\gamma}{2d}
 \; 
 \frac{ D(R) - D(\tilde R) }{ \gamma_{i^*} } \;.
$$
This then establishes the existence of an $i^*$ such that 
\Cref{eq:i-condition} is satisfied.

We can inductively restart the process with $R' = R_{i^*}$
and $\eps' = \eps (1 + \lambda_{i^*})$.
In each iteration, 
the discrepancy density must increase 
by at least a multiplicative factor of 
$ (1 + \gamma / 2d) \geq 
(1 + C \; \beta_d/2d)$, for some universal constant $C>0$.
Since the discrepancy density is at most one, 
the process must terminate in 
$O \lp(   d \; \beta_d^{-1} \; \log(1/\eps ) \rp)$ many iterations,
and we will eventually find some rectangle $R^*$ such that 
\Cref{eq:good-case} is satisfied.

It remains to show that the mass $D(R^*)$ is {bounded below}.
% We will use super-script to denote the variables in the $t$-th iteration. 
Suppose that in the $t$-th iteration, we start from the rectangle 
$R^{(t)}$ with discrepancy density 
$\eps^{(t)} := \frac{ \abs{\p(R^{(t)}) - \q(R^{(t)})} }{ D(R^{(t)}) }$ and end with the rectangle $R^{(t+1)}:=R^{(t)}_{i^*}$ 
with discrepancy density
$\eps^{(t+1)} := \eps^{(t)} \; \lp(  1 + \gamma_{i^*}^{(t)} \rp)$.
Denote by $\tilde R^{(t)}$ the rectangle discarded 
and $\eps^{(t+1/2)} := \eps^{(t)} \; \lp(  1 + \gamma^{(t)} \rp) $ 
the discrepancy density of the remaining space 
$R^{(t)} \backslash \tilde R^{(t)}$.
We analyze how much the mass of the rectangle 
can shrink in each iteration, as follows: 
\begin{align*}
\frac{ D(R^{(t)}_{i^*}) }{ D(R^{(t)}) }
&\geq 
\frac{ \gamma^{(t)} }{ 2d } \; \lp( 1 - D(\tilde R^{(t)} ) /  D(R^{(t)}) \rp) \; \frac{1}{ \gamma_{i^*}^{(t)} }  \\
&\geq
\frac{ \gamma^{(t)} }{ 2d } \; 
\lp( \frac{ 1 -  D(\tilde R^{(t)} ) /  D(R^{(t)})  }
{1 - { D(\tilde R^{(t)} ) / \lp( 2D(R^{(t)}) \rp) }} \rp)^2
\; \frac{1}{ \gamma_{i^*}^{(t)} } \\
& =
\frac{  \gamma^{(t)} }{ 2d } \; 
\lp( \frac{1}{1 +  \gamma^{(t)}} \rp)^2
\; \frac{1}{ \gamma_{i^*}^{(t)} }  \\
&\geq
\Omega(1) \; \frac{ \beta_d }{ d^3 } \; \lp( \frac{1}{1 + \gamma_{i^*}^{(t)}} \rp)^3
\, ,
\end{align*}
where the first line uses our choice of $R^{(t)}_{i^*}$ such that 
\Cref{eq:i-condition} is satisfied, 
the second line uses the elementary inequality
$ (1-x) \geq (1-x)^2/(1-x/2)^2 $ for any $x \leq 1$, 
the third line uses the definition of $\gamma^{(t)}$, 
and the last line uses the facts 
$\gamma_{i^*}^{(t)} \geq \gamma^{(t)}/2d$ by our choice of 
$R^{(t)}_{i^*}$ such that \Cref{eq:i-condition} is satisfied
and $\gamma^{(t)} \geq \Omega( \beta_d )$ 
by our choice of $\tilde R^{(t)}$.

Notice that the discrepancy density increases 
by a multiplicative factor of $1 + \gamma_{i^*}^{(t)}$ 
in the $t$-th iteration.
Thus, we have
\begin{align} \label{eq:density-bound}
\eps \; \prod_{t} \lp( 1 + \gamma_{i^*}^{(t)} \rp) \leq 1 \;.
\end{align}
Therefore, $D(R^*)$ is at least
$$
\prod_t C \; \frac{ \beta_d }{ d^3 } \; \lp( \frac{1}{1 + \gamma_{i^*}^{(t)}} \rp)^3
\geq 
\lp( C \; \frac{ \beta_d }{ d^3 } \rp)^{ O \lp(   d \; \beta_d^{-1} \; \log(1/\eps ) \rp) }
\; \eps^3 
\geq 
\eps^{ \tilde O \lp( d \beta_d^{-1} \rp) } \;,
$$
where we used the fact that the process terminates 
in at most $O \lp(   d \; \beta_d^{-1} \;  \log(1/\eps ) \rp)$ 
iterations and \Cref{eq:density-bound}.
This then shows that there exists a rectangle $R^* \subseteq R$ 
such that $D(R^*) \geq \eps^{\alpha_d}$ for some 
$\alpha_d = \tilde O \lp( d \beta_d^{-1} \rp)$ 
and \Cref{eq:good-case} is satisfied.
Then it holds
\begin{align*}
\underset{x,y \sim D_{|R^*} } \E \lp[ \abs{ \p (\tilde R) - \q(\tilde R) } \rp]
&\geq 
\eps/2 \; 
\frac{1}{2^{2^d}} \; D(R^*)^2
\geq 
\eps/2 \; 
\frac{1}{2^{2^d}} \; \eps^{2\alpha_d}
\; D(R)
\geq \eps^{\alpha_d'} \; D(R) \;,
\end{align*}
{where $\alpha_d' = \tilde O \lp( d \beta_d^{-1} \rp) \leq C \; d^2 \; 2^{2^{d+1}}$, for some sufficiently large universal constant $C$.}
This concludes the proof of \Cref{lem:rectangle-discrepancy}.
\end{proof}

\subsection{Proof of \Cref{lem:small-support-l2}}
\label{sec:light-bin-l2-test}
Given a discrete distribution $\p$, 
flattening~\cite{DiakonikolasK16} is the technique of using a small 
set of samples from $\p$ to appropriately subdivide its bins (domain elements) 
aiming to reduce the $\ell_2$-norm of the distribution. 
Formally, the flattening technique yields what was described 
in \cite{DiakonikolasK16} as a \emph{split distribution}.
\begin{definition}[Definition 2.4 from \cite{DiakonikolasK16}]
Given a distribution $\p$ on $[n]$ and a multiset $S$ of elements of $[n]$, 
define the \emph{split distribution} $p_S$ on $[n + |S|]$ as follows: 
For $1 \leq i \leq n$, let $a_i$ denote $1$ plus the number of elements of $S$ 
that are equal to $i$. Thus, $\sum_{i=1}^n a_i = n + |S|$. 
We can therefore associate the elements of $[n + |S|]$ to elements 
of the set $B = \{ (i,j) : i \in [n], 1 \leq j \leq a_i \}$. 
We now define a distribution $\p_S$ with support $B$, 
by letting a random sample from $\p_S$ be given by $(i,j)$, 
where $i$ is drawn randomly from $p$ and $j$ is drawn randomly from $[a_i]$.
\end{definition}

We will use the following basic facts 
about split distributions.

\begin{fact}
[Fact 2.5 and Lemma 2.6 from \cite{DiakonikolasK16}]
\label{lem:flatten}
Let $\p$ and $\q$ be probability distributions on $[n]$, and $S$ a given multiset of $[n]$. Then:
(i) We can simulate a sample from the split distributions
$\p_S$ or $\q_S$ by taking a single sample from $\p$ or $\q$, respectively. 
(ii) It holds $\snorm{1}{\p_S - \q_S} = \snorm{1}{\p - \q}$. 
(iii) For any multisets $S \subseteq S' \subseteq [n]$,
$\snorm{2}{\p_{S'}} \leq \snorm{2}{\p_{S}}$. 
(iv) If $S$ is obtained by drawing $\Poi(m)$ samples from $\p$, 
then $\E \lp[ \snorm{2}{\p_S}^2 \rp] \leq 1/m$.
% Moreover, it holds:
% (i) , and
% (ii) If $S$ is obtained by taking $\Poi(m)$ samples from $p$, then $\E \lp[ \snorm{2}{p_S}^2 \rp] \leq 1/m$.
\end{fact}

% \begin{lemma}
% [Lemma 2.6 from \cite{DiakonikolasK16}]
% \label{lem:flatten}
% Let $p$ be a distribution on $[n]$. Then: (i) For any multisets $S \in S'$ of $[n]$,
% $\snorm{2}{p_{S'}} \leq \snorm{2}{p_{S}}$, and
% (ii) If $S$ is obtained by taking $\Poi(m)$ samples from $p$, then $\E \lp[ \snorm{2}{p_S}^2 \rp] \leq 1/m$.
% \end{lemma}
% \paragraph{$\ell_2$-Closeness Tester.}
We will also leverage the following $\ell_2$-distance estimator
to develop our final tester. 
\begin{lemma}
[Proposition 6 from \cite{chan2014optimal}]
\label{lem:robust-l2-estimate}
Let $\p$ and $\q$ be unknown distributions on $[n]$.
There exists an algorithm that on input $n, \eps > 0$, and 
$b \geq \max \lp(  \snorm{2}{\p}^2, \snorm{2}{\q}^2 \rp)$, 
it draws $\Poi(m)$ samples from $\p, \q$, 
where $m = \Theta \lp( \sqrt{b} / \eps^2 
+ \sqrt{b} \; \snorm{4}{\p - \q}^2 / \eps^4
\rp) $,
and with probability $3/4$ estimates $\snorm{2}{\p-\q}$ 
up to accuracy $\pm \eps$.
% distinguishes between the cases $\snorm{2}{\p-\q} = 0$ and
% $\snorm{2}{\p-\q} > \eps$.
\end{lemma}

We can easily convert the above $\ell_2$-distance estimator 
to an $\ell_2$-closeness tester, which is more applicable to our setting.

\begin{corollary}
\label{lem:robust-l2}
Let $\p$ and $\q$ be unknown distributions on $[n]$.
There exists an algorithm that on input $n, \eps > 0$, and 
$b \geq \max \lp(  \snorm{2}{\p}^2, \snorm{2}{\q}^2 \rp)$, 
the algorithm draws $\Poi(m)$ samples from $\p, \q$, 
where $m = \Theta \lp( \sqrt{b} / \eps^2 \rp) $,
and with probability $3/4$ distinguishes between 
the cases $\p = \q$ versus $\snorm{2}{\p-\q} > \eps$.
\end{corollary}

Let $S$ be a multiset of $\Poi(m)$~\iid samples from 
$(1/2) \lp(\p +\q \rp)$ for some $m \leq s / 100$. 
First, we argue that the $\ell_2$-distance 
between $\p, \q$ will not decrease by too much 
after flattening, 
by taking advantage of the fact that the $\ell_2$-discrepancy between $\p, \q$ 
is  supported on a few light elements.

Let $X_i$ be the random variable denoting the number of samples from $S$ landing
in the $i$-th element, i.e., $X_i \sim \Poi(m \; \lp( \p_i +\q_i\rp) / 2  )$.
Then the expected discrepancy restricted to the elements from 
the set of elements witnessing the discrepancy, i.e., $H$, 
after flattening  is at least 
\begin{align*}
\sum_{i \in H} \E \lp[  \frac{\lp(\p_i - \q_i \rp)^2}{ X_i +1 } \rp]
= 
\sum_{i \in H} \lp(\p_i - \q_i \rp)^2 \; (1 - e^{- \lambda_i})/ \lambda_i \, ,
\end{align*}
where $\lambda_i \eqdef m \; \lp( \p_i +\q_i\rp) / 2$.
Notice that $(1 - e^{-x})/x$ is a decreasing function with respect to $x$. 
Since $m \leq s / 100$ and $(\p_i +\q_i)/2 \leq s$, it follows that 
$(1 - e^{- \lambda_i})/ \lambda_i$
is bounded below by
$ 1 - e^{ -0.01 } / 0.01 \geq 0.99$.
This then gives us
\begin{align} \label{eq:expect-lower-bound}
\E\lp[ \sum_{i \in H}  \frac{\lp(\p_i - \q_i \rp)^2}{ X_i +1 } \rp]
\geq 
0.99 \sum_{i \in H} \lp( \p_i - \q_i \rp)^2.
\end{align}
On the other hand, the variance
of the discrepancy restricted to the elements in $H$, after flattening,
is bounded above by
\begin{align*}
\Var \lp[ \sum_{i \in H}  \frac{\lp(\p_i - \q_i \rp)^2}{ X_i +1 } \rp]
&= 
\E \lp[ \lp(\sum_{i \in H}  \frac{\lp(\p_i - \q_i \rp)^2}{ X_i +1 } \rp)^2 \rp]
- \E^2 \lp[ \lp(\sum_{i \in H}  \frac{\lp(\p_i - \q_i \rp)^2}{ X_i +1 } \rp)^2 \rp] \\
&\leq 
 \lp(\sum_{i \in H} \lp( \p_i - \q_i \rp)^2\rp)^2
 - (0.99)^2 \lp(\sum_{i \in H} \lp( \p_i - \q_i \rp)^2 \rp)^2 \, ,
\end{align*}
where in the last inequality we use the fact that $\frac{1}{1 +X_i}$ 
is at most $1$ and \eqref{eq:expect-lower-bound}.
This then gives us
\begin{align} \label{eq:var-upper-bound}
    \sqrt{ \Var \lp[ \sum_{i \in H}  \frac{\lp(\p_i - \q_i \rp)^2}{ X_i +1 } \rp] }
    \leq 0.15 \;  \sum_{i \in H} \lp( \p_i - \q_i \rp)^2.
\end{align}
Combining \eqref{eq:expect-lower-bound} and~\eqref{eq:var-upper-bound}, 
we obtain
\begin{align} \label{eq:good-discrepancy}
    \Pr \lp[ \snorm{2}{ \p_S - \q_S }^2 <  \frac{1}{3} \sum_{i \in H} \lp( \p_i - \q_i \rp)^2  \rp] 
    \leq 1/4 \;.
\end{align}
On the other hand, by \Cref{lem:flatten} and Markov's inequality, it holds
\begin{align} \label{eq:good-l2}
    \Pr \lp[ \max \lp( \snorm{2}{ \p }^2, \snorm{2}{ \q }^2 \rp) > 40 / m  \rp]  \leq 1/10 \;.
\end{align}
By the union bound, \eqref{eq:good-discrepancy},~\eqref{eq:good-l2} and 
\Cref{lem:robust-l2}, it follows that the $\ell_2$-closeness tester
of \Cref{lem:robust-l2} 
succeeds with probability at least $2/3$, if we take  
$m' = C \sqrt{ \frac{1}{m} } \; \eps^{-2}$ many samples, for a sufficiently large constant $C$. 
Balancing $m$ and $m'$ (with the restriction that $m \leq s / 100$ in mind) 
then gives us that 
the overall tester succeeds with probability at least $2/3$
if we draw $\Poi(m)$ many~\iid samples with 
$$
m = \Theta \lp(  \max \lp( \eps^{-4/3}, \eps^{-2} / \sqrt{s} \rp) \rp) \;.
$$
{This concludes the proof of Lemma~\ref{lem:small-support-l2}.} \qed

\subsection{
Applications: Closeness Testing of Multivariate Structured Distributions under Total Variation Distance} \label{sec:application}
%In this section, we provide some applications of $\Ak$ closeness testing of multivariate distributions.
The most direct application of our multivariate $\Ak$-closeness tester 
is for the problem of testing closeness of multivariate histogram 
distributions --- distributions that are piecewise constant over (the same) \emph{unknown} 
collection of axis-aligned rectangles --- with respect to the total variaton distance.

This follows directly from our main theorem, since for any pair of $k$-histogram distributions
$\p, \q$ with respect to the same set of rectangles,  we have 
$\dtv(\p, \q) = \frac{1}{2} \sum_{i=1}^k \abs{ \p(R_i) - \q(R_i) } = \frac{1}{2} \snorm{\Ak}{\p - \q}$. Formally, we have the following: 

\begin{corollary} \label{cor:md-hist}
Let $\{R_i\}_{i=1}^k$ be a set of axis-aligned rectangles in $\R^d$.
Suppose $\p, \q$ are distributions over $\R^d$ that are piecewise constant over each of $\{R_i\}_{i=1}^k$, i.e., $\p(x) = \p(y)$ for any $x,y \in R_i$ and the same for $\q$.
Then there exists a tester which distinguishes between 
$\p = \q$ and $\dtv(\p, \q) > \eps$ with sample complexity
$
C \; k^{6/7} \; \eps^{ -2 \alpha_d / 3 } \; \log^d(k) \; 2^{d/3}
$, 
where $C$ is a sufficiently large universal constant and $\alpha_d = O(d^2 2^{2^{d+1}})$.
\end{corollary}

We now proceed with our second application. 
We consider the binary hypothesis class $H$ consisting 
of all possible $k$-unions of 
axis-aligned rectangles within the unit cube $[0,1]^d$. Given two 
hypotheses $h_1, h_2 \in H$, we can test whether $h_1$ is 
equivalent to $h_2$ or they are far from each other under the 
uniform distribution over the unit cube $[0,1]^d$.

\begin{corollary}
\label{cor:rectangle-hypothesis}
Let $H$ be the class of all possible $k$-unions of axis-aligned 
rectangles within the unit cube $[0,1]^d$, i.e., $$H = \lp\{h |
h = \bigcup_{i=1}^k R_i \text{ where } \{R_i\}_{i=1}^k \subset 
[0,1]^d \text{ are disjoint axis-aligned rectangles over } [0,1]^d \rp\}.
$$
Let $h_1, h_2$ be two unknown hypotheses from $H$.
Given $\eps > 0$ and sample access to $(x, h_i(x))$, 
where $x$ follows the uniform distribution over $[0,1]^d$, 
there exists an efficient algorithm 
which distinguishes with probability at least $2/3$
between (i) $h_1(x) = h_2(x)$ for all $x$, and 
(ii) $ \E_{ x \sim U } \lp[ \mathbbm 1 \{ h_1(x) \neq h_2(x) \} \rp] > \eps $, 
where $U$ is uniform distribution over $[0,1]^d$. 
Moreover, the algorithm has sample complexity
$
C \; k^{6/7} \; \eps^{ -2 \alpha_d / 3 } \; \log^d(k) \; 2^{d/3}
$, 
where $C$ is a sufficiently large constant and $\alpha_d = O(d^2 2^{2^{d+1}})$.
\end{corollary}
\begin{proof}
Consider the distributions $\p, \q$ defined as follows. 
To draw a sample from $\p$, we take a sample $(x, h_1(x))$ where $x \sim U$.
If $h_1(x) = 1$, we return $x$. 
Otherwise, we return some arbitrarily chosen point $s \not \in [0, 1]^d$.
We define $\q$ similarly based on $h_2$.
If $h_1$ and $h_2$ are identical, it is easy to see that $\p = \q$.
If $\E_{ x \sim U } \lp[ \mathbbm 1 \{ h_1(x) \neq h_2(x) \} \rp]$, 
we claim that $\snorm{\Ak}{\p- \q} \geq \eps/2 $.
Suppose that $h_1$ is the union of the rectangles $\{R_i\}_{i=1}^k$ 
and $h_2$ is the union of the rectangles $\{R_i'\}_{i=1}^k$.
Then we have that 
\begin{align*}
&\E_{ x \sim U } \lp[ \mathbbm 1 \{ h_1(x) \neq h_2(x) \} \rp] \\
&= \int_{x \in [0,1]^d} U(x) \; \mathbbm 1 \lp\{ x \in \bigcup_{i=1}^k R_i  \backslash \bigcup_{i=1}^k R_i' \rp\} 
+ \int_{x \in [0,1]^d} U(x) \; \mathbbm 1 \lp\{ x \in \bigcup_{i=1}^k R_i'  \backslash \bigcup_{i=1}^k R_i \rp\} 
\\
&\leq 2 \max \lp( 
\sum_{i=1}^k \int_{x \in [0,1]^d} U(x) \; \mathbbm 1 \lp\{ x \in R_i \backslash \bigcup_{i=1}^k R_i' \rp\} \, , \, 
\sum_{i=1}^k \int_{x \in [0,1]^d} U(x) \; \mathbbm 1 \lp\{ x \in R_i' \backslash \bigcup_{i=1}^k R_i \rp\}
\rp).
\end{align*}
Without loss of generality, we assume that the first term is larger. 
Then we have that
\begin{align*}
\sum_{i=1}^k \int_{x \in [0,1]^d} U(x) \; \mathbbm 1 \lp\{ x \in R_i \backslash \bigcup_{i=1}^k R_i' \rp\}
\geq \frac{1}{2} \eps \;,
\end{align*}
if $\E_{ x \sim U } \lp[ \mathbbm 1 \{ h_1(x) \neq h_2(x) \} \rp] > \eps$.
On the other hand, 
we also have
\begin{align*}
\sum_{i=1}^k \int_{x \in [0,1]^d} U(x) \; \mathbbm 1 \lp\{ x \in R_i \backslash \bigcup_{i=1}^k R_i' \rp\}
=
\sum_{i=1}^k \p(R_i) - \q(R_i).
\end{align*}
Thus, this gives $\snorm{\Ak}{\p - \q} \geq \eps/2$.
Therefore, we can distinguish between the two cases 
by performing $\Ak$-closeness testing between $\p, \q$ 
with accuracy parameter $\eps/2$.
\end{proof}

%% file: lower.tex
\section{Sample Complexity Lower Bound} \label{sec:lb}

In this section, we prove our sample complexity lower bound.
Specifically, we show that the task of $\mathcal A_k$-closeness testing
gets information-theoretically harder as we go from one dimension to two 
dimensions. For the one-dimensional case, it was shown in~\cite{diakonikolas2015optimal} that 
the sample complexity of $\mathcal A_k$-closeness testing is 
$\Theta \lp( \max  \lp( k^{4/5} \eps^{-6/5}, k^{1/2} \eps^{-2} \rp)  \rp)$. Perhaps surprisingly, for two-dimensional distributions, 
we prove a sample complexity lower bound of 
$\Omega\lp( k^{6/7} / \eps^{8/7}\rp)$ in the sublinear regime, 
where $\eps > k^{-1/8}$. 
This lower bound clearly dominates the sample complexity of one-dimensional $\Ak$ testing in the same regime.

At a very high level, we build on the lower bound framework of \cite{diakonikolas2015optimal}.
In particular, our lower bound proof consists of two steps. 
First, we argue that, if the domain size is a sufficiently large function of $d, k$, 
we can assume without loss of generality that the output of the tester 
only depends on the {\em relative order} of samples ranked in each coordinate. 
This is shown in Section~\ref{sec:order-base}.

Then, for such ``order-based'' testers, we present two explicit families 
of pairs of two-dimensional distributions such that a random pair of distributions 
from the first family are identical, 
and a random pair of distributions from the second family 
are far from each other in $\Ak$-distance.
Moreover, a random pair of distributions from the first family 
is hard (i.e., requires many samples)
to distinguish from a random pair from the second.
This step requires a carefully designed gadget consisting 
of distributions over $\R^2$ supported on the edges of a square.
We present the construction and analyze its key properties in Section~\ref{sec:square-edge}.

Next we appropriately replicate the gadget many times 
to create the full hard-instance of $2$-dimensional $\Ak$-closeness testing. 
The description of the hard instance and its detailed analysis can be found in Section~\ref{sec:construction}.

Finally, we provide an alternative way to prove a sample complexity 
lower bound against general $\Ak$ testers, 
while requiring the domain size to be at most 
doubly exponential in {$k$}. 
This involves a careful application of randomly chosen 
monotonic transformations to the $x$ and $y$ coordinates 
of all points in order to hide extra ``non-order based'' 
information that a tester can retrieve from the numerical values 
of the sample coordinates. 
This more refined construction and its analysis are presented 
in Section~\ref{sec:domain-optimize}.

\subsection{Order-Based Testers} \label{sec:order-base}
{Here we define the class of order-based testers 
and show that we can translate lower bounds against order-based testers 
to general testers at the cost of increasing the domain size.}
More formally, we consider algorithms which are restricted 
to obtain information from what we call the \emph{Order Sampling} process, 
as opposed to the usual direct sampling.
{This can be thought of as follows. 
We first draw \iid samples from the unknown distributions. 
Then, instead of feeding them directly to the algorithm, 
we perform an appropriate pre-processing to extract 
only the information related to the order of the coordinates 
of the samples, and reveal only the order information to the algorithm.}
% \subsection{Reduction to Order Sampling}
\begin{definition}[Order Sampling]
\label{def:order-sampling}
Let $\p, \q$ be a pair of distributions in $\R^2$. 
Let $\{(x_i, y_i), \ell_i\}_{i=1}^m$ be $m$ \iid samples, where
$(x_i, y_i)$ are sampled from $(1/2)(\p + \q)$ 
and $\ell_i$ records whether the sample comes from $\p$ or $\q$.
Let $\sigma(x), \sigma(y) \in \mathbb S_m$ be the permutation representing the rank of the $x$-coordinates and $y$-coordinates accordingly.
The \emph{Order Tuple} associated with the $m$ samples is given by $ 
\Order( \{x_i, y_i, \ell_i\}_{i=1}^m ) = ( \sigma(x), \sigma(y), \ell )$.
Furthermore, we will use 
$\mathcal D(\p, \q, m)$ to denote the distribution over the tuple $( \sigma(x), \sigma(y), \ell )$ obtained through this process. 
\end{definition}

As our first structural lemma, we show that if an algorithm is able to perform $\mathcal A_k$-closeness testing with direct sample access on a domain of size $N \times N$, then we can always use it to build another algorithm which performs the test with only the order tuple of the same number of samples --- albeit on a smaller domain of  size $n \times n$. 
The proof uses a Ramsey-theoretic argument 
and generalizes Theorem 13 in \cite{diakonikolas2015optimal}.
\begin{lemma} \label{lem:order-reduction}
For all $n, m, k \in \Z^+$ where $m < n$ and $\eps > 0$, there exist  $N_1,N_2 \in \Z^+$ such that the following holds: If there
exists an algorithm $A$ that for every pair of distributions $\p,\q$ over $[N_1] \times [N_2]$ distinguishes the case $\p = \q$ from the case  $\snorm{\mathcal A_k}{\p-\q} > \eps$ with probability at least $4/5$ while taking $m$ samples from $\p$ and $\q$, then there exists an algorithm $A'$ that for every pair of distributions $\p', \q'$ over $[n] \times [n]$ distinguishes the case $\p' = \q'$ versus $\snorm{ \mathcal A_k }{\p'-\q'} > \eps$ with probability at least $2/3$ given a tuple $T$ from the order sampling process $\D(\p', \q', m)$.
\end{lemma}
\begin{proof}
Suppose we are given the algorithm $A$ 
which can perform $\mathcal A_k$-closeness testing 
over the domain $[N_1] \times [N_2]$ given direct \iid sample access to $\p, \q$. 
We show that we can use $A$ to construct another algorithm $A'$ 
which performs the test with only tuples obtained from the order sampling process 
over the domain $[n] \times [n]$. 

Let $\{(x_i, y_i), \ell_i\}_{i=1}^m$ be the samples drawn by $A$.
We will write $A( \{ (x_i, y_i), \ell_i\}_{i=1}^m )$ to denote the probability 
that $A$ outputs ``YES'' given these samples.
Before we specify our construction, we remark that we can without loss of generality assume 
that the image of $A( \{(x_i, y_i), \ell_i\}_{i=1}^m )$ has size at most {$11$}. 
This is because we can always round the probability to the nearest multiples of {$1/10$} 
and lose only {$1/10$} in the overall success probability.

Let $\p,\q$ be the unknown distributions supported on $[n] \times [n]$.
The key step is to argue the existence of two monotonic transformations 
$f_x: [n] \mapsto [N_1]$, $f_y: [n] \mapsto [N_2]$, 
{where $N_2$ is chosen to be a sufficiently large function of $n$, 
and  $N_1$ is chosen to be a sufficiently large function of $n$ and $N_2$}, 
such that if one feeds the samples $  \{ \lp( f_x(x_i), f_y(y_i) \rp), \ell_i \} $ to $A$, 
the output of $A$ becomes a function only of  $\Order( \{ \lp(x_i, y_i\rp), \ell_i\} )$. 
In other words, we 
want to find two mappings $f_x, f_y$ such that
$$
A( \{ \lp( f_x(x_i), f_y(y_i) \rp), \ell_i \}  ) =  A( \{ \lp( f_x(x_i'), f_y(y_i') \rp), \ell_i' \}  ) \;,
$$ 
as long as $ \Order \lp( \{ (x_i, y_i), \ell_i\} \rp) = \Order\lp( \{(x_i', y_i'), \ell_i'\} \rp) $. 
Given such mappings, we can then define 
$A'(\{ (x_i, y_i), \ell_i) \}  ) := A(\{ \lp(f_x(x_i), f_y(y_i) \rp), \ell_i \}  ) $.
Then, it is easy to see that $A'$ is an order-based tester.
Furthermore, since $f_x, f_y$ are both monotonic, the domain transformation 
will preserve the $\Ak$ distance between $\p$ and $\q$. 
Hence, $A'$ enjoys the same guarantee and gives the correct answer 
with probability at least $2/3$.
% claim algorithm $A'$ can simply do the following.
% \begin{enumerate}
%     \item Upon receiving the order tuple $T \sim \D(\p, \q, m)$, arbitrarily choose a set of samples $ \{x_i, y_i, \ell_i\}_{i=1}^m $  such that $T = \Order \lp( \{x_i, y_i, \ell_i\} \rp)$.
%     \item Apply the transformation and feeds the sample points to $A$, i.e. output $A \lp( \{ f_x \lp( x_i \rp), f_y \lp( y_i \rp), \ell_i \} \rp)$.
% \end{enumerate}
% the correctness of the above construction follows from two observations. 
% Firstly, since $f_x, f_y$ are both monotonic, the domain transformation will preserve the $\mathcal A_k$ distance between $\p$ and $\q$. Therefore, if one samples $\{(x_i, y_i, \ell_i)\}_{i=1}^m$ from $\p,\q$ and then feeds $\{ (f(x_i), f(y_i), \ell_i) \}$ to $A$, the output will  be correct with constant probability.
% On the other hand, for any $\{x_i', y_i', \ell_i'\}$ such that 
% $\Order\lp( \{ x_i', y_i', \ell_i'\} \rp)  = \Order\lp( \{ x_i, y_i, \ell_i\} \rp)$,the output of algorithm $A$ running on $\{f_x(x_i'), f_y(y_i'), \ell_i'\}$ should be the same as running it on $\{f_x(x_i), f_y(y_i), \ell_i\}$ by our choice of the transformations $f_x, f_y$. 
% Combining the two observations then gives the correctness of algorithm $A'$.

We next show the existence of such a pair of transformations $f_x, f_y$. 
We do so in two steps. First, we show the existence of the transformation 
$f_x$ which will make the output of algorithm $A$
independent of the actual values of the $x$-coordinate.
% depend on only the rank information of the $x$-coordinate, the $y$-coordinates and the labels. 
This then allows us to construct an algorithm $A_x$ that depends only on the 
rank information of the $x$-coordinates, the $y$-coordinates and the labels.
Then we show the existence of $f_y$, which is defined with respect to $A_x$, 
that makes the output of $A_x$ independent of the actual values of the $y$-coordinates. 
This then allows us to conclude the existence of the algorithm $A'$.

For convenience, we will rewrite the tuples $\{(x_i, y_i,\ell_i)\}_{i=1}^m$ 
as $( \mathcal X, \sigma(x), \{y_i\}_{i=1}^m, \{\ell_i\}_{i=1}^m )$, 
where $\mathcal X$ is the set of $x$-coordinates and $\sigma(x)$ 
is the permutation which maps $i \in [m]$ to the rank of $x_i$ among $\{x_i\}_{i=1}^m$. 
For each $\mathcal X \in [N_1]^m$, we can define a mapping 
$g_{\mathcal X}: \mathbb S_m \times [N_2]^m \times \{0,1\}^m \mapsto [0,1]$ 
induced by the algorithm $A$ as 
$g_{\mathcal X}(\sigma(x), \{y_i\}_{i=1}^m, \{\ell_i\}_{i=1}^m  ) := A( \mathcal X, \sigma(x), \{y_i\}_{i=1}^m, \{\ell_i\}_{i=1}^m )$.
Notice that {the set of values that $g_{\mathcal X}$ has size at most $11$}, 
since we assume the acceptance probability of $A$ conditioned 
on any input can take at most $11$ different values.
We note that there can be at most 
{${11}^{m! N^m 2^m}$} many different types of mapping $g_{\mathcal X}$.

If we view $\mathcal X$ as a hyper-edge of the hypergraph $ {{[N_1]} \choose m}$ 
and the associated mapping $g_{\mathcal X}$ as the coloring of the hyper-edge, 
by Ramsey's theorem, there exists a subset of vertices $V$ of size $n$ 
such that the coloring of the hyper-edges in the sub-graph ${V \choose m}$ 
are all the same as long as $N_1$ is sufficiently large compared to $N_2$ and $m$. 
In other words, there exists a subdomain $V \subset [N_1]$ such that 
if the $x$ coordinates of the samples are all from this subdomain, 
the acceptance probability of algorithm $A$ becomes a function of only $y_i, \ell_i, \sigma(x)$ 
and independent of the actual $x$-coordinates $\mathcal X \subseteq V$. 
We will then choose $f_x$ as the order-preserving mapping from $[n]$ to $[N_1]$, 
where the image is exactly $V$. 

We next consider the algorithm $A_x$ which first applies the transformation 
$f_x$ and then runs the testing algorithm $A$ on the resulting samples. 
From the argument above, we know that algorithm $A_x$ 
depends only on $\sigma(x), \{y_i\}_{i=1}^m, \{\ell_i\}_{i=1}^m$. 
Similarly, we can rewrite the tuple as $\sigma(x), \sigma(y), \mathcal Y, \{\ell_i\}_{i=1}^m$, 
where $\mathcal Y$ is the set of $y$-coordinates and $\sigma(y)$ is the permutation 
which maps $i$ to the rank of $y_i$. With a similar argument, 
as long as $N_2$ is sufficiently large compared to $n, m$, 
we can show the existence of an order-preserving mapping $f_y$ such that 
if we apply the mapping $f_y$ first and then run $A_x$, 
the output of $A_x$ becomes only a function of $\sigma_x, \sigma_y, \{\ell_i\}_{i=1}^m$ 
and independent of the actual set of $y$ coordinates $\mathcal Y$. 
Notice that $\sigma_x, \sigma_y, \{\ell_i\}_{i=1}^m$ is exactly the order tuple 
$\Order\lp( \{ (x_i, y_i), \ell_i\}_{i=1}^m \rp)$. 
Hence, such a pair of transformations $f_x, f_y$ are exactly what we need 
to construct algorithm $A'$.
{Setting $A'(\{ (x_i, y_i), \ell_i) \}  ) := A(\{ \lp(f_x(x_i), f_y(y_i) \rp), \ell_i \}  )$ then concludes the proof}.
\end{proof}

\subsection{Square-Edge Distributions } \label{sec:square-edge}
We now present the building block of our lower bound construction, which consists of distributions supported on the edges of a square.
{Notice that though the domain is $\R^2$, the supports of such distributions are lower-dimensional. 
We will use $\t, \r$ to represent such distributions} and
one can refer to \Cref{fig:square} for a visual illustration.
\begin{definition}[Square-Edge Distributions] \label{def:square-edge}
Consider a square in $\R^2$ whose diagonals are parallel to the $x$-axis and $y$-axis.
We define $\t$ as the uniform distribution supported on the upper-left and lower-right edges and $\r$ as the uniform distribution supported on the remaining two edges.
\end{definition}
Let $(a,b)$ be a point lying on the edges of the square. The space can be divided into four regions by drawing one horizontal and one vertical lines across $(a,b)$.
% based on $a,b$: . 
The most important property that we will rely on in our analysis is the following: For any such point $(a,b)$, any of the resulting four regions have the same mass under $\t$ as under $\r$.
\begin{fact} \label{clm:equal-mass}
Let $\t, \r$ be the square-edge distributions defined as in \Cref{def:square-edge}.
Consider a point $(a,b) \in \supp(\t) \cup \supp(\r)$.
Denote the four regions as
$
 R_{a,b}^{(1)} = \{ x > a, y > b | (x,y) \in \R^2  \},
 R_{a,b}^{(2)} = \{ x < a, y < b | (x,y) \in \R^2  \},
 R_{a,b}^{(2)} = \{ x > a, y < b | (x,y) \in \R^2  \},
 R_{a,b}^{(2)} = \{ x < a, y > b | (x,y) \in \R^2  \}
$.
Then, it holds
$
\t \lp( R_{a,b}^{(i)} \rp)
= \r \lp( R_{a,b}^{(i)} \rp)
$ for all $i$.
\end{fact}

\begin{figure}
    \centering
    \includegraphics[width=10cm]{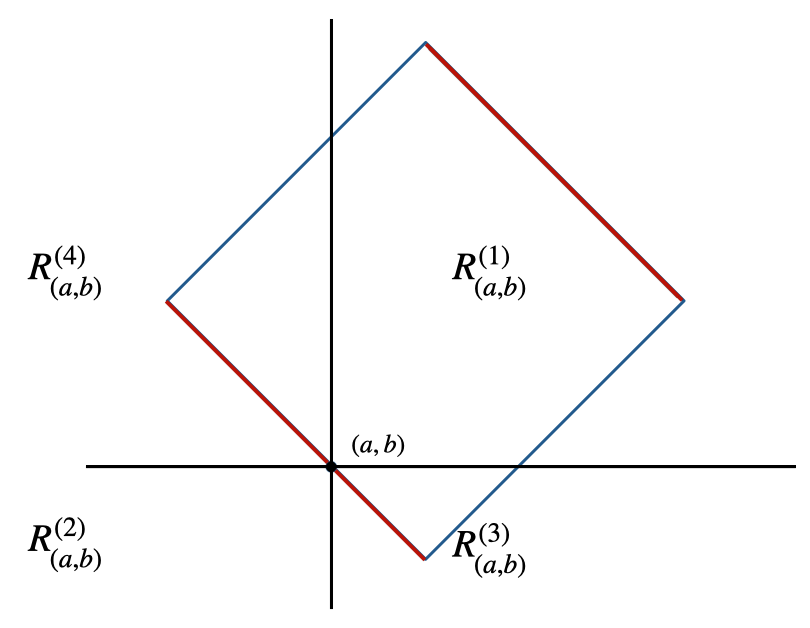}
    \caption{Square Edge Distributions}
    \label{fig:square}
    \medskip
\small
The red lines represent the distribution $\t$ and the blue lines represent the distribution $\r$. 
For any point $(a,b)$ on the edges of the square, it is easy to verify that the four regions 
$R^{(i)}_{(a,b)}$ in \Cref{clm:equal-mass} have the same probability mass under $\t$ as under $\r$.
\end{figure}
Intuitively, the above fact says that if one partitions the space based 
on one sample $(a,b)$, the tester cannot distinguish between $\t$ and 
$\r$ simply based on their mass on any of the regions $R_{a,b}^{(i)}$. 
As a consequence, to distinguish $\t$ and $\r$, one needs to take more 
samples to partition the space into finer pieces (for example, taking 
two samples and considering the rectangle formed by the two samples). 

To formalize this intuition, we will consider the  distribution obtained by performing order sampling under a pair of distributions composed of the square-edge distributions.
{In particular, imagine the following scenario, which can be thought of 
as a toy example of $\Ak$ closeness testing for $k=4$.
In the YES case, we have
$\p_{\Yes} = \q_{\Yes} = (\t + \r)/2$. 
Then we obtain order sampling with $m$ samples drawn 
from $\p_{\Yes}, \q_{\Yes}$, 
according to \Cref{def:order-sampling}. 
The resulting order tuple will then have the distribution
$\mathcal D( (\t + \r)/2, (\t + \r)/2, m)$ 
over $\mathbb S_m \times \mathbb S_m \times \{0,1\}^m$.
In the NO case, with probability $1/2$, we have $\p_{\No} = \t, \q_{\No} = \r$. 
Otherwise, we have $\p_{\No} = \r$ and $\q_{\No} = \t$.
Then, if we perform order sampling with $m$ samples from $\p_{\No}, \q_{\No}$, 
we obtain an order tuple following the uniform mixture of 
$\frac{1}{2} \lp( \mathcal D(  \t, \r, m ) + \mathcal D(\r, \t, m) \rp)$.
Notice that in the YES case, we have $\p_{\Yes} = \q_{\Yes}$; 
in the NO case, we have $\snorm{\Ak}{ \p_{\Yes} - \q{\Yes} } = 1$ deterministically, 
even for $k = 4$.
Yet, we show in the next lemma that the distributions over order-tuples 
in the two cases are the same when $m$ is no more than $3$.
This immediately gives us that no order-based algorithm 
can distinguish between the two cases with fewer than $4$ samples.
}
\begin{lemma}
\label{lem:diamond-order-match}
We have that
$
\mathcal D( (\t + \r)/2, (\t + \r)/2,  m) = 
\lp( \D(  \t,  \r, m ) +  \D(\r, \t, m) \rp)/2$ for $m = 1,2,3$.
\end{lemma}
\begin{proof}
Let $(\sigma(x), \sigma(y), \ell)$ be an order tuple.
We remark that the tuple can be decomposed into two parts: (i) the permutation patterns $\sigma(x), \sigma(y) \in \mathbb S_3$, which encodes the ``geometric pattern'' of the three points sampled and (ii) a bit string $\ell \in \{0,1\}^3$, which indicates whether the samples come from $\p$ or $\q$.
Now let $(\sigma(x)_{\Yes}, \sigma(y)_{\Yes}, \ell_{\Yes}) \sim \mathcal D( (\t + \r)/2, (\t + \r)/2,  m)$
and $(\sigma(x)_{\No}, \sigma(y)_{\No}, \ell_{\No}) \sim \lp( \D(  \t,  \r, m ) +  \D(\r, \t, m) \rp)/2$.
We begin with the following observations.
\begin{enumerate}
    \item The marginal distribution over the ``geometric pattern'' is identical for the two cases, i.e.
    $
    \Pr\lp[  \sigma(x)_{\Yes} =  \pi, \sigma(y)_{\Yes} = \pi' \rp] = 
    \Pr \lp[  \sigma(x)_{\No} =  \pi, \sigma(y)_{\No} = \pi' \rp]
    $
    for all $\pi, \pi' \in \mathbb S_m$.
    This is because the samples, ignoring the labels, in both cases come from the distribution supported uniformly on the four edges of the square.
    \item The distribution of $\ell_{\Yes}$ conditioned on any ``geometric pattern'' will be uniform over all possible bit strings, i.e. $\Pr \lp[ \ell_{\Yes} = \beta |  \sigma(x)_{\Yes} =  \pi, \sigma(y)_{\Yes} = \pi' \rp]$ is the same for all $ \beta  \in \{0,1\}^{m}$ and $\pi, \pi' \in \mathbb S_m$. {This is because $(\sigma(x)_{\Yes}, \sigma(y)_{\Yes}, \ell_{\Yes})$ is obtained by performing order sampling from two identical distributions (both are $(\t + \r)/2$).}
\end{enumerate}

Hence, it suffices to show that the distribution over the label vector $\ell_{\No}$ conditioned on any geometric patterns $\sigma(x)_{\No}, \sigma(y)_{\No}$ is uniform.

With this observation in mind, the $m = 1$ case is trivial since there is only $1$ geometric pattern and it is clear that the label $\ell_{\No}$ is uniform.
For $m = 2$, let the coordinate of the first sample be $(a,b)$, which divides the space into four quadrants.
Then, by \Cref{clm:equal-mass}, it holds that no matter which of the four quadrants the second sample fall into,
the probability that the point comes from $\t$ is the same as it comes from $\r$. Hence, the uniformity of $\ell_{\No}$ follows.
% Hence, we also have $\D(  \p_{\Yes}, \q_{\Yes}, 2 ) = \frac{1}{2} \D(  \p_{\No}^{(1)}, \q_{\No}^{(1)}, 2) + \frac{1}{2} \D(  \p_{\No}^{(2)}, \q_{\No}^{(2)}, 2)$.
% Let $\lp(  \sigma(x)_{\Yes}, \sigma(y)_{\Yes},\ell_{\Yes} \rp) \sim \D(  \p_{\Yes}, \q_{\Yes}, 3 )$ 
% and 
% $\lp(  \sigma(x)_{\No}, \sigma(y)_{\No}, \ell_{\No} \rp)  \sim \frac{1}{2} \D(  \p_{\No}^{(1)}, \q_{\No}^{(1)}, 3 ) + \frac{1}{2} \D(  \p_{\No}^{(2)}, \q_{\No}^{(2)}, 3 ).
% $

For $m=3$, we make some preliminary simplifications.
Let $\{(x_i, y_i), \ell_i\}_{i=1}^m$ be three \iid samples drawn. 
Since they are all identically distributed and independent, the sampling order does not matter.
Hence, we can without loss of generality just examine the case $x_1 < x_2 < x_3 $ (and accordingly $\sigma(x) = (1,2,3)$).
% Then, we see there are $6$ patterns left for $\sigma(y)$. 
Secondly, observe that our construction is invariant under reflections over $x$- or $y$-axis, and rotations of angle $\pi/4, \pi/2, 3 \pi /4$.
After reflection over the $x$-axis, any three points that have the pattern $\sigma(y) = (1,2,3)$ ($x_1 < x_2 < x_3, y_1 < y_2 < y_3$) then becomes $\sigma(y) = (3,2,1)$ ($x_1 < x_2 < x_3, y_1 > y_2 > y_3$). After rotations, the pattern $(1,3,2)$ yields $(2,3,1)$, $(2,1,3)$ and $(3,1,2)$. Hence, by symmetry, we can simply focus on the argument for $\sigma(x) = (1,2,3), \sigma(y) = (1,2,3)$ and $
\sigma(x) = (1,2,3),
\sigma(y) = (1,3,2)$.

We will begin with $\sigma(x) = (1,2,3), \sigma(y) = (1,2,3)$ and show that $\ell_{\No}$ is uniform conditioned on that.
% Let $(\beta_1, \beta_2, \beta_3) \in \{0,1\}^3$.
% We will begin by show that 
% $\sigma(y)_{\No}$
% $$
% \Pr [ \sigma_{\No} = (1,2,3),  \ell_{\No} = (\beta_1, \beta_2, \beta_3) ] 
% $$ 
% can be simplified to an expression that is independent of $(\beta_1, \beta_2, \beta_3)$.
% Hence, we can write
% \begin{align*}
%     \Pr [ \sigma_{\No} = (1,2,3),  \ell_{\No} = (\beta_1, \beta_2, \beta_3) ] &= \frac{1}{2} \Pr [ \sigma_{\No} = (1,2,3),  \ell_{\No} = (\beta_1, \beta_2, \beta_3) |\p_{\No} = \t, \q_{\No} = \r ] \\
%     &+\frac{1}{2} \Pr [ \sigma_{\No} = (1,2,3),  \ell_{\No} = (\beta_1, \beta_2, \beta_3) |\p_{\No} = \r, \q_{\No} = \t ].
% \end{align*}
% We will now focus on the case
% first term $ \Pr [ \sigma_{\No} = (1,2,3),  \ell_{\No} = (\beta_1, \beta_2, \beta_3) |\p_{\No} = \t, \q_{\No} = \r ]$.
% Let $ \{(x_i, y_i, \ell_i)\}_{i=1}^3$ be  the three \iid samples drawn.
We claim that this is true even 
if we further condition on the coordinates of the ``middle point'': we will condition on that $x_2 = x, y_2 = y$ for some arbitrarily chosen point $(x,y)$ from the support.
It is easy to see that the marginal distribution of $\ell_2$ is uniform since it only depends on whether we are sampling from $\mathcal D( \t, \r, 3)$ or $\mathcal D( \r, \t, 3)$.
For the same reason, further conditioning on the value of $\ell_2$ then completely determines whether we are sampling from $\mathcal D( \t, \r)$ or  $\mathcal D( \r, \t, 3)$.
Consequently, $(x_1, y_1, \ell_1), (x_3, y_3, \ell_3)$ are now independent samples from the lower left quadrant $R_{x,y}^{(2)}$ and upper right quadrant $R_{x,y}^{(4)}$ of the point $(x,y)$ respectively.
% If we further condition on the value of $\ell_2$, whether we are sampling from $\mathcal D( \p^{(1)}_{\No}, \q^{(1)}_{\No}, 3)$ or  $\mathcal D( \p^{(2)}_{\No}, \q^{(2)}_{\No}, 3)$ is then determined. 
% Then, the conditional distributions of $(x_1, y_1)$ and $(x_3, y_3)$ become independent.  
% Moreover, $x_1, y_1$ now becomes a random point from the lower left quadrant $R_{x,y}^{(2)}$.
By \Cref{clm:equal-mass}, the amount of mass from $\p^{(i)}_{\No}$ and from $\q^{(i)}_{\No}$ in $R_{x,y}^{(2)}$ is the same. Hence, the conditional distribution for $\ell_1$ is uniform (and similarly for $\ell_3$).

Next, we will show that $\ell_{\No}$ conditioned on $\sigma(x)_{\No} = (1,2,3), \sigma(y)_{\No} = (1,3,2)$ is also uniform.
% $\Pr [ \sigma(y)_{\No} = (1,3,2), \ell_{\No} = (\beta_1, \beta_2, \beta_3) ]$ is also independent of $\beta_i$.
Notice that it actually suffices for us to show the uniformity of $\ell_{\No}$ conditioned on the more general event 
{$\sigma(x)_{\No} = (1,2,3)$ and $\sigma(y)_{\No}$ is either $(1,3,2)$ (the case we are analyzing now) or $(1,2,3)$ (the case analyzed in the previous paragraph).} 
If this is true, we can then combine it with the fact that $\ell_{\No}$ is uniform conditioned on $\sigma(x)_{\No} = (1,2,3), \sigma(y)_{\No}= (1,2,3)$ to {conclude that 
$\ell_{\No}$ must be uniform conditioned on $\sigma(x)_{\No} = (1,2,3), \sigma(y)_{\No} = (1,3,2)$.}
Notice that this more general event happens if and only if 
$x_1 < \min (x_2, x_3)$ and $y_1 < \min(y_2, y_3)$.
We can then use techniques similar to the analysis of the last case. In particular, we claim that this is true even if we further condition on the coordinates of the first point: $x_1 = x, y_1 = y$ for some arbitrary point $(x,y)$ from the support. The analysis is then almost the same: After we have conditioned on the value of $(x_1, y_1, \ell_1)$, $(x_2, y_2, \ell_2), (x_3, y_3, \ell_3)$ now both become independent samples from the upper right quadrant $R_{x,y}^{(4)}$. Applying \Cref{clm:equal-mass} then allows us to conclude the uniformity of $\ell_2, \ell_3$ after the conditioning.
{This finishes the argument that $\ell_{\No}$ conditioned on any geometric patterns $\sigma(x)_{\No}, \sigma(y)_{\No}$ and  
concludes the proof.}
% We have already shown that $\ell_{\No}$ is uniform conditioned on $\sigma(x)_{\No} = (1,2,3), \sigma(y)_{\No}= (1,2,3)$
% Combining this with the observation that $\ell_{\No}$ is uniform conditioned on $\sigma(x)_{\No} = (1,2,3), \sigma(y)_{\No}= (1,2,3)$ then gives it is also uniform conditioned on $\sigma(x)_{\No} = (1,2,3), \sigma(y)_{\No}= (1,2,3)$
% $\Pr [ \sigma_{\No} = (1,2,3), \ell_{\No} = (\beta_1, \beta_2, \beta_3) ]$ is independent of  $(\beta_1, \beta_2, \beta_3)$. We then have 
% $\Pr [ \sigma_{\No} = (1,3,2), \ell_{\No} = (\beta_1, \beta_2, \beta_3) ]$ is also independent of $(\beta_1, \beta_2, \beta_3)$.
\end{proof}
% \subsection{Replication among the diagonal}
% \label{construction}

\subsection{$\Ak$ Closeness Lower Bound Construction} \label{sec:construction}
We will now use $\t , \r$ as building blocks to construct 
the full hard instance of $2$-dimensional $\Ak$ closeness testing
and establish the desired sample complexity lower bound 
$\Omega\lp( \min\lp ( k^{6/7} \eps^{-8/7} \, , \, k\rp) \rp)$.

{
We will readily apply the ``Poissonization trick'', which is a standard technique in proving 
lower bounds for distribution testing problems.
In particular, instead of drawing a fixed number of $m$ samples, 
we make the testers draw $\Poi(m)$ many samples. 
It is easy to translate any lower bound in the Poisson sampling model 
to the standard sampling model where the testers draw a fixed number of samples, 
since with probability at least $99\%$ 
the testers will receive at least $\Omega(m)$ many samples.

Furthermore, we will relax $\p, \q$ to be non-negative measures 
whose total mass is $\Theta(1)$ rather than equal to $1$. 
Clearly, taking samples from a non-negative measure $\mu$ is no longer a sensible concept.
Instead, we can take $\Poi(m \snorm{1}{\mu})$ samples from the normalized distribution 
$\mu / \snorm{1}{\mu}$.
We will slightly abuse the definition of sampling to describe the above the process 
as ``taking $\Poi(m)$ samples from $\mu$''. 
}

Lastly, since we are only proving a sample complexity lower bound that is sublinear 
with respect to $k$, we can safely assume $m < k/2$ throughout the section.
% Our general approach follows the framework from~\cite{DiakonikolasK16}: We exhibit two family of distributions 
% the testing problem when the algorithm is restricted to take inputs from the order sampling process with $\Poi(m)$ \iid~samples assuming $m< k$.

Now we are ready to describe the hard instance.
We will first partition the domain into $r^2$ squares with equal size,  for some $r = \Theta(k)$ that will be specified later. 
Most of the squares will be left blank: 
$\p,\q$ will have all their probability mass 
supported within the squares along one diagonal of the square grids. 
For each square on the diagonal, we will make it a ``heavy'' square 
with probability $m/k$ (this is a well-defined probability since $m<k$) 
and a ``light'' square otherwise, whose purpose will become clear later. 

Now consider the following random process for generating a pair of measures $\p, \q$. 
Let $X$ be a random variable that takes $0$ or $1$ 
each with probability $1/2$. If $X=0$, we will randomly generate 
a pair of measures $\p = \q$, which belongs to the YES instance. 
If $X=1$, we randomly generate a pair of measures satisfying 
$\snorm{\mathcal A_k}{\p-\q} > \Omega(\eps)$, which belong to the NO instance. 

When $X=1$, $\p,\q$ restricted to one square (after normalization) 
will be both $(\t + \r)/2$, which is the uniform distribution 
supported on the edges of a diagonal square.
Moreover, the mass of $\p$ will be $1/m$ if the square is ``heavy'' 
and $\eps/k$ if the square is ``light'' (and the same for $\q$ as well).

When $X=0$, the mass of $\p, \q$ restricted to a square 
will be the same as the case $X=1$. Yet, the conditional 
distributions within a square for $\p,\q$ will be different.
\begin{itemize}[leftmargin=*]
    \item For a ``heavy'' square, the  conditional distributions of 
    $\p,\q$ restricted to the square are still both $(\t + \r)/2$. 
    Intuitively, samples produced by the ``heavy'' squares behave the 
    same in the NO instance as in the YES instance, serving as noise to  
    ``confuse'' the algorithm.
    \item For a ``light'' square, 
    the conditional distributions of $\p,\q$ restricted to the square 
    are respectively $\t, \r$ with probability $1/2$ and $\r, \t$ 
    otherwise.
    These squares contribute to the $\Ak$ discrepancy between $\p, \q$ 
    but remain hard to distinguish from the YES case.
\end{itemize}
% We will refer to such a pair of measures generated from the random process above as a \emph{Order }
We first argue that the measures $\p, \q$ constructed from the random process described above qualify for basic properties of $\Ak$ closeness testing.
\begin{lemma}
\label{lem:basic-properties}
Suppose $m < k/2$. 
It holds that $\p, \q$ are positive measures with mass $\Theta(1)$ 
with probability $99\%$. Moreover, if $X=1$, we have $\p = \q$. 
If $X = 0$, we have 
{$\snorm{\Ak}{ \p/\snorm{1}{\p} - \q/\snorm{1}{\q} } > \Omega(\eps)$} 
with probability $99\%$.
\end{lemma}
\begin{proof}
We first  verify that $\p, \q$ are both measures with mass $\Theta(1)$ with probability $99\%$. 
By Chebyshev's inequality, we have that
the number of heavy squares is $
r \; \frac{m}{k} \pm \Theta(1) \; \sqrt{r \; \frac{m}{k}} = 
\Theta(1) \; \frac{rm}{k}
$ with probability $99\%$.
Conditioned on that, the contribution of the heavy squares to mass is
$\Theta(1) \; \frac{rm}{k} \; \frac{1}{m} = \Theta(1)$ given that $r = \Theta(k)$.
The contribution of the light squares is at most $r \; \frac{\eps}{k} = O(\eps)$.
Hence, we have the total mass will be $\Theta(1)$.

If $X=1$, it is easy to see that $\p = \q$. 
If $X=0$, for each light square $R$, recall that $\p, \q$ restricted to $R$ are exactly the square edge distributions after normalization.
By the definition of the square edge distribution, 
there exists $4$ sub-squares $R_1, R_2, R_3, R_4$ such that 
for each $R_i$, exactly one of $\p(R_i), \q(R_i)$ is $0$ and the other one is $\eps/(2k)$.
We have seen that $\snorm{1}{\p}, \snorm{1}{\q}$ are both $\Theta(1)$.
Hence, we have
$\sum_{i=1}^4 \abs{\p(R_i) / \snorm{1}{\p} - \q(R_i) /  / \snorm{1}{\q}} \geq \Omega(\eps / k)$.
Moreover, with probability $99\%$,
the number of light squares is
$ r \; \lp(1 - \frac{m}{k} \rp) \pm  \Theta(1) \; \sqrt{ r \; \lp(1 - \frac{m}{k} \rp) }
= \Theta(1) \; r 
$ since $m/k < 1/2$.
Conditioned on this, if we choose $r =  c \; k$ for a sufficiently small constant $c$, 
we ensure that there are $r' = \Theta(1) \; c \; k$ light squares.
Notice that if $c$ is chosen appropriately, we can ensure $\Omega(k) < r' < k/4$.
Therefore, there exists $k' = 4 \; r' < k$ rectangles such that
$
\sum_{i=1}^{k'} \abs{ \p(R_i)/\snorm{1}{\p} - \q(R_i)/\snorm{1}{\q} }
= \eps/k  \; r' = \Omega(\eps).
$
\end{proof}

Let $T$ be the tuple obtained from the order sampling process $\D(\p, \q, m')$, 
where $\p, \q$ are the pair of random measures described above and $m' \sim \Poi(m)$. 
We will bound above the mutual information $I(X:T)$, 
{implying} that $T$ reveals little information of the random variable $X$.
The {implication} argument is standard, see, e.g., the proof of Theorem 16 from \cite{diakonikolas2015optimal}.
In particular, we try to bound the information about $X$ obtained from samples 
falling in each of the squares.
In \cite{diakonikolas2015optimal}, we have that squares with fewer than two samples are uninformative.
By \Cref{lem:diamond-order-match}, we can further ignore the squares with three samples, 
therefore allowing
us to obtain a stronger lower bound.

{Our key technical lemma is the following:}

\begin{lemma} \label{lem:mutual-info-bound}
We have that $I(X:T) = O( m^7 \eps^8 / k^6  )$.
\end{lemma}
\begin{proof}
Let $Y = \{ (x_i, y_i), \ell_i \}_{i=1}^{m'}$, where $m' \sim \Poi(m)$, be the sample points drawn. Namely, $T = \Order(Y)$. Denote by $Y_i$ the set of points in the $i$-th square along the diagonal and define the tuple $T_i = \Order(Y_i)$. 
One can easily reconstruct $T$ from $\{T_1, \cdots, T_r\}$: Given $i < j$, all points from the $i$-th square will be ranked after points from the $j$-th square in both $x$ and $y$ coordinates in $T$. This hence gives us that $I(X : T) \leq \sum_{i=1}^r I(X : T_i)$. Next, we will bound $I(X:T_i)$ by $ O(m^7 \eps^8 / k^7) $. 
Our lemma easily follows from that since we also have $r = \Theta(k)$. 
We first bound the mutual information as a summation over all possible order tuples grouped by the size of the order tuple (recall that for an order tuple $t$, the size of the order tuple, denoted as $\abs{t}$, is simply the number of samples from which the order tuple is derived). We have that
\begin{align*}
    I(X: T_i)
    \leq O(1) \; \sum_{\lambda=0}^{\infty} \hspace{0.5em}
    \underset{\text{order tuple } t: \abs{t} = \lambda}  \sum
    \frac{
    \lp( \Pr \lp[ T_i = t | X = 0 \rp]
    - \Pr \lp[ T_i = t | X = 1 \rp]
    \rp)^2
    }
    {\Pr \lp[ T_i = t \rp]}.
\end{align*}
We will use the indicator variable $H_i$ to denote whether the $i$-th square is chosen to be a ``heavy'' square.
Notice that 
$\Pr[H_i = 0] = 1 - \frac{m}{k} = O(1)$ and $H_i$ is independent of $X$.
Furthermore, if the $i$-th square is chosen to be a heavy square, the distribution of $T_i$ conditioned on $X = 0$ and $X = 1$ is exactly the same. This gives us that
\begin{align*}
I(X: T_i)
    \leq
    O(1)
    \cdot
    \sum_{\lambda=0}^{\infty}
    \sum_{ t: \abs{t} = \lambda }
    \frac{
    \lp( \Pr \lp[ T_i = t | H_i = 0, X = 0 \rp]
    - \Pr \lp[ T_i = t | H_i = 0, X = 1\rp]
    \rp)^2
    }
    {\Pr \lp[ T_i = t \rp]}.
\end{align*}
Next, we note that
$\Pr \lp[ T_i = t | X = 0, H_i = 0\rp]$ for $|t| = \lambda$ 
is given by the distribution 
$\frac{1}{2} 
\D(\t, \r, \lambda) 
+ 
\frac{1}{2} 
\D(\r, \t, \lambda) 
$.
On the other hand,
$\Pr \lp[ T_i = t | X = 1, H_i = 0\rp]$ for $|t| = \lambda$
is given by the distribution 
$
\D( (\t + \r)/2 , (\t + \r)/2, \lambda).
$
Hence, by Lemma \ref{lem:diamond-order-match},
it holds
$$
\Pr \lp[ T = t| H_i = 0, X = 0\rp] = \Pr \lp[ T = t| H_i = 0, X = 1\rp]
$$ 
for any $t$ satisfying $ |t| \leq 3$. 
This allows us to discard the summation over any $t$ with $\abs{t} \leq 3$.
Hence, the expression can be further upper bounded by
\begin{align}
    &O(1)  
    \sum_{\lambda=4}^{\infty}
    \sum_{ t: \abs{t} = \lambda }
    \frac{
    \lp( \Pr \lp[ T_i = t|H_i = 0, X = 0 \rp]
    - \Pr \lp[ T_i = t|H_i = 0, X = 1 \rp]
    \rp)^2
    }
    {\Pr \lp[ T_i = t \rp]}  \nonumber
    \\
    &\leq O(1)
    \sum_{\lambda=4}^{\infty}
    \sum_{ t: \abs{t} = \lambda }
    \frac{
    \lp( 
    \Pr[T_i = t| H_i = 0]
    \rp)^2
    }
    {\Pr \lp[ T_i = t, H_i = 1 \rp]}  \nonumber \\
     &\leq O(1)
     \sum_{\lambda=4}^{\infty}
     \max_{ t: |t| = \lambda }
     \frac{ \Pr[T_i = t | H_i = 0] }
     { \Pr[T_i = t, H_i = 1] }
     \sum_{ t: \abs{t} = \lambda }
     \Pr[ T_i = t |H_i = 0 ]   \nonumber \\
     &= O(1)
     \sum_{\lambda=4}^{\infty}
     \max_{ t: |t| = \lambda }
     \frac{ \Pr[T_i = t | H_i = 0] }
     { \Pr[T_i = t, H_i = 1] } \; 
     \Pr[ |T_i| = \lambda | H_i = 0 ]  
     \label{eq:mutual-info-bound-1}
     \, ,
\end{align}
where in the second line above we upper bound
the difference in the numerator by their sum and upper bound the denominator by $\Pr[T_i = t, H_i = 1]$, in the third line above we use that $\sum_i a_i \; b_i \leq \lp(\max_i a_i\rp) \; \lp(\sum_i b_i \rp)$ when $a_i, b_i \geq 0$
and in the final equality we note that the summation over the probability of $T_i = t$ for each $|t| = \lambda$ is exactly 
that of $|T_i| = \lambda$.
Next we claim that
\begin{align}
\label{eq:uniform-l}
\max_{t: |t| = \lambda} \frac{ \Pr[T_i = t | H_i = 0] }
{ \Pr[T_i = t, H_i = 1] }
\leq O(1) \; 2^{\lambda} \;  \frac{ \Pr[ |T_i| = \lambda | H_i = 0 ] } { \Pr[|T_i| = \lambda, H_i = 1] } \;.
\end{align}
To show this, we first remark that
\begin{align} \label{eq:condition-chain-rule}
\max_{t: |t| = \lambda} \frac{ \Pr[T_i = t | H_i = 0] }
{ \Pr[T_i = t | H_i = 1] }
=
\max_{t: |t| = \lambda} \frac{ \Pr[T_i = t | H_i = 0, |T_i| = \lambda] }
{ \Pr[T_i = t | H_i = 1, |T_i| = \lambda] }
\; 
 \frac{ \Pr[ |T_i| = \lambda  | H_i = 0]  }{ \Pr[ |T_i| = \lambda  | H_i = 1]}.    
\end{align}
Then recall that $T_i$ can be decomposed into a binary vector $\ell_i$ representing the labels and a permutation tuple $\sigma_i \in S_{\lambda} \times S_{\lambda}$ representing the rank information of $x$ and $y$ coordinates. 
We note that $\sigma_i | H_i = 0, |T_i| = \lambda$ has the same distribution as $\sigma_i | H_i = 1, |T_i| = \lambda$. 
Then, conditioned on $\sigma_i$, the distribution of $\ell_i$ is uniform when $H_i = 1$.
This then gives
$$
\max_{t: |t| = \lambda} \frac{ \Pr[T_i = t | H_i = 0, |T_i| = \lambda] }
{ \Pr[T_i = t | H_i = 1, |T_i| = \lambda] }
\leq O(2^{\lambda}).
$$
Combining this with Equation~\eqref{eq:condition-chain-rule} and 
multiplying both sides by $\frac{1}{ \Pr[H_i = 1] }$
then gives~\eqref{eq:uniform-l}.
Substituting Equation~\eqref{eq:uniform-l} into 
Equation~\eqref{eq:mutual-info-bound-1} then gives us
\begin{align*}
     I(X:T_i) 
    &\leq
    O(1) \cdot
    \sum_{\lambda=4}^{\infty}
    2^{\lambda} \;\frac{     \lp( 
    \Pr\lp[ \abs{T_i} = \lambda \big| H_i = 0  \rp]
    \rp)^2 }
    { \Pr \lp[ |T_i| = \lambda, H_i = 1\rp] }.  
\end{align*}
Finally, notice that $\Pr[|T_i| = \lambda | H_i = 1] = \Poi( 
1, \lambda )
= \Theta(1) / \lambda!
$, $\Pr[H_i = 1] = m/k$, and 
$\Pr[ |T| =  \lambda | H_i = 0 ] = \Poi( \eps \;m/k, \lambda  ) \leq  (\eps \;m/k)^{\lambda} / \lambda!$. This further gives
\begin{align*}
I(X:T_i)
&\leq O(1) \;\sum_{\lambda=4}^{\infty}
\frac{2^\lambda}{\lambda!}
 \;
\frac{ k }
{ m  }
\;\lp( \frac{\eps m}{k} \rp)^{2\lambda}\\
&= O(1) \;\frac{k}{m} \;\sum_{\lambda=4}^{\infty} \lp( \sqrt{2} \;\frac{\eps \;m}{k} \rp)^{2 \lambda}
\leq O(1) \;\lp( \frac{m}{k} \rp)^7 \;\eps^8.
\end{align*}
This concludes the proof of \Cref{lem:mutual-info-bound}.
\end{proof}
We are now ready to conclude the proof of our main lower bound result.

\begin{proof}[Proof of Lower Bound in Theorem~\ref{thm:main-intro}]
By Lemma~\ref{lem:basic-properties}, given that $m < k/2$, 
it holds that both $\p, \q$ are measures of mass $\Theta(1)$ with probability 
at least $99\%$ and if $X = 0$, it holds 
$\snorm{\Ak}{\p - \q} > \Omega(\eps)$ with probability at least $99\%$.
By Lemma~\ref{lem:mutual-info-bound}, we have that the mutual information 
between the random bit $X$ and the ordering tuple 
$T \sim \D(\p, \q, m')$, for $m' \sim \Poi(m)$, 
is at most $O( m^7 \eps^8 / k^6 )$. 
{This means that no algorithm, given $T$ as input, can reliably predict
the value of $X$ with probability more than $2/3$ 
unless $m > \Omega(1) \;\min \lp( k^{6/7} / \eps^{8/7}, k\rp) $.
By \Cref{lem:basic-properties}, it holds that 
$\p / \snorm{1}{\p}$, $\q / \snorm{1}{\q}$ are a pair of identical distributions 
if $X = 0$ and a pair of distributions that are $\Omega(\eps)$ far in $\Ak$ distance 
with probability at least $99\%$ if $X=1$.
Furthermore, with probability $99\%$, $T$ is an order-tuple of at most $O(m)$ many samples.
Therefore, we conclude that the sample complexity of $\Ak$ testing 
is at least $\Omega(1) \;\min \lp( k^{6/7} / \eps^{8/7}, k\rp)$.}

Even though the distributions $\p, \q$ used in the construction are continuous, 
we next show that they can be easily ``rounded'' to discrete distributions 
that remain hard for the testing algorithm. In particular,  we can construct 
a grid $\mathcal G$ which splits the domain into $\Theta \lp( m^6 \rp)$ squares 
such that the mass of any square $R \in \mathcal G$ under $(1/2)(\p+\q)$ is bounded by $m^3$. 
Then, we consider the discrete distributions $\p', \q'$ which round 
the points falling in the square $R \in \mathcal G$ to its top-left vertex. 
It is easy to see that if $\p = \q$, then $\p' = \q'$.
Moreover, for an arbitrary rectangle $R \subset \R^2$, 
we have $ \p(R) - \q(R)  =  \p'(R) - \q'(R) \pm \Theta( \frac{1}{m^3} )$.
Hence, the effect of rounding to the $\Ak$ distance between $\p,\q$ 
is at most $\Theta(k/m^3) $, which can be safely ignored when $m<k$.
On the other hand, $\D(\p, \q, m')$ is nearly the same as $\D(\p', \q', m')$, 
since the distributions over the order tuples are the same as long 
as no two points fall in the same square 
(which happens with probability at most $O(1/m)$).
Hence,  the cases $\p' = \q'$ and $\snorm{\mathcal A_k}{\p'-\q'} > \eps$  
are also hard to distinguish given tuples from the order sampling process 
unless $m > \Omega(1) \;\min \lp( k^{6/7} / \eps^{8/7}, k\rp) $.
Finally, by Lemma~\ref{lem:order-reduction}, we can translate any lower bound 
under order sampling back to the usual direct sampling.
\end{proof}

\subsection{Domain Size Optimization} \label{sec:domain-optimize}
The lower bound of Theorem~\ref{thm:main-intro} holds only when the domain size $N$ is substantially larger than the other parameters. 
In particular, the statement does not quantitatively characterize the sample complexity 
as a function of the domain size. 
The bottleneck of the analysis lies in 
\Cref{lem:order-reduction}, which offers 
an inefficient (in terms of the size of the domain after the transformation) way of transforming the 
domain to ``hide'' the extra information  that an algorithm can extract from the samples in addition to 
their relative order. In this section, we provide a more efficient and constructive way to disguise the 
information in the values of each samples' coordinates and build on it to provide a tighter lower bound 
in terms of the domain size. 

{The main result of this section is the following:}

\begin{theorem}[Stronger Lower Bound for Discrete Distributions] \label{thm:lower-bound-refined}
Fix an integer $V > 0$.
Let $\p$ and $\q$ be distributions on $[V] \times  [V]$ and let $\eps > 0$ be less than a sufficiently small
constant. 
Any tester that distinguishes between $\p = \q$ and $\snorm{\Ak}{ \p - \q } \geq \eps$ 
for some $k \leq V$ 
with probability at least $2/3$
must use at least $m$ many samples for some $m$ with 
$$
m \geq \Omega(1) \cdot \min \lp( 
k^{2/3} \eps^{-4/3} \cdot \lp( \frac{ \log \log V }{ \log \log \log V} \rp)^{1/3} \, , \, 
k^{6/7} \eps^{-8/7} \, , \, k
\rp).
$$
\end{theorem}

Before presenting the transformation formally, we provide some high level intuition. 
Recall that in the lower bound construction from \Cref{sec:construction} 
the domain is partitioned into $r \times r$ many squares 
where $r = \Theta(k)$ and the distributions are supported on squares lying on the diagonal. 
The argument then proceeds to bound the order information of samples coming from each of the squares.
Now suppose that the algorithm is allowed to look at the absolute coordinates of the samples. 
If only $1$ or $2$ points fall in some square, 
the only extra information we need to hide is its absolute position 
and the distance between the points. 
To do so, we can generalize the techniques developed in \cite{DKN17} 
to randomly scale and shift the square in both the $x$ and $y$-axis.

For $2$-dimensional $\Ak$ closeness testing, 
if the algorithm takes $\Theta( k^{6/7} \eps^{-8/7} )$ many samples, 
since there are $\Theta(k)$ many squares in total, 
$3$ or more samples could fall in the same square.
Then the algorithm also gets to see the ratio of distances 
between different pairs of points, which remains invariant 
even if the coordinates of the points are scaled uniformly within the square.
To handle this, we will instead apply an uneven scaling on different parts of the square. 
In particular, we map points with $x$-coordinate $a$ to $\exp(a \;\lambda)$ 
with some randomly chosen $\lambda$ (and the same for the $y$-coordinate), 
which then makes the ratio of distances also noisy.

To formalize this idea, we first define a distribution over monotonic mappings, 
which we will then use to transform the points.
\begin{definition}[Distribution over monotonic mappings] \label{def:monotonic-mapping}
Let $W > 0$.
We define $\mathcal M(W)$ as a distribution over monotonic mappings of the form 
$f: [0,1] \mapsto \R_+$. To sample a mapping from $\mathcal M$, 
we first sample three parameters $\lambda_1, \lambda_2, \lambda_3$ which are uniform variables 
over the intervals $[\log \log W, 2 \log \log W], [0, \log^3 W], [0, \exp\lp( 2\log^3 W \rp)]$ respectively.
Then, the mapping $f \sim \mathcal M(W)$ is given by
$ f(x) = \exp(x \;\exp(\lambda_1) ) \;\exp(\lambda_2) + \lambda_3  $.
\end{definition}
Let $a < b <c$ be three points lying on $[0,1]$. 
% We can see that $a,b,c$ contains the same amount of information as
% $\frac{c-a}{b-a}, b-a, a$
% since the former can easily be constructed from the latter.
Here we show that, as long as $a,b,c$ are sufficiently separated, transforming the points by some random 
mapping $f$ from $\mathcal M(W)$ helps obfuscate the information a tester can retrieve from them.
In particular, we argue the distribution of $(f(a), f(b), f(c))$ (where the randomness is over $f$) is close to some fixed distribution $D$ for any choice of well-separated points $a,b,c$.
% Let $f$ be a mapping drawn from $\mathcal M(W)$ such that
% $f(x) = \exp(\lambda_1 \;x ) \;\lambda_2 + \lambda_3 $.
% We can define the quantities
% $A \eqdef 
% \frac{ f(a) - f(b) }{ f(b) - f(c)  }  \, ,
% B \eqdef f(a) - f(b)
% \, ,
% C \eqdef f(a).
% $
% Notice that the tuple $(A,B,C)$ is information theoretically equivalent to $ (f(a), f(b), f(c) )$ since we can easily reconstruct the latter from the former. Moreover, $(A,B,C)$ corresponds exactly to the three kinds of information a tester can retrieve from the points in addition to their order information.
% Here, we claim the transformation helps disguise the information by making them closed to the uniform distribution up to some logarithmic transformations.
\begin{lemma} \label{lem:tv-to-u}
Let $f \sim \mathcal M(W)$
such that $f(x) = \exp( \exp(\lambda_1) \;x ) \;\exp(\lambda_2) + \lambda_3$.
Then, there exists some fixed distributions $D$ over the domain $\R_{+}^3$ such that
for any three points $a < b < c$ from $[0, 1]$ satisfying
\begin{align} \label{eq:separated-condition}
  \min(c-b, b-a) > 1 /\log \log W \;,
\end{align}
we always have
\begin{align*}
&\dtv( ( f(a), f(b), f(c) ), D ) \leq O \lp(  
\frac{\log \log \log W}{\log \log W}\rp).
\end{align*}
\end{lemma}
\begin{proof}
Define $A \eqdef 
\frac{ f(c) - f(a) }{ f(b) - f(a)  }  \, ,
B \eqdef f(b) - f(a)
\, ,
C \eqdef f(a)
$.
First, we note that it suffices to show 
$(\log \log A, \log B,C)$ is close in total variation distance to some distribution $D'$ 
for an arbitrary choice of $a,b,c$ satisfying the condition in Equation~\eqref{eq:separated-condition}
since is a bijection between  $(f(a), f(b), f(c))$ and $(\log \log A, \log B,C)$.

In particular, let $U_1, U_2, U_3$ are uniform distributions over the intervals
$[\log \log W, 2 \log \log W]$, 
$[0, \log^3 W]$ and $[0, \exp \lp(  2 \;\log^3 W \rp)]$ respectively.
We argue $(\log \log A, \log B,C)$ is close to the distribution $U_1 \times U_2 \times U_3$.
The proof strategy is the following. We first bound the total variation distance between $\log \log A$ and $U_1$. Then, conditioned on $\log \log A$ and $U_1$, we show $\log B$ is close to $U_2$. Finally, conditioned on everything other variables, we show $C$ is close to $U_3$.

% \begin{align}
% \label{eq:loglog}
% &\E_{a,b,c} \lp[ \dtv( \log \log A, U_1 ) \rp] \leq O \lp ( \frac{1}{ \log \log W } \rp) \, , \\
% \label{eq:log}
% &\E_{a,b,\lambda_1} \lp[ \dtv( \log B , U_2 ) \rp] \leq O \lp ( \frac{1}{ \log \log W } \rp)\, , \\
% \label{eq:}
% &\E_{a,\lambda_1, \lambda_2} \lp[ \dtv( C   , U_3 ) \rp] \leq O \lp ( \frac{1}{ W } \rp)
% \end{align}
% where we define $A \eqdef 
% \frac{ f(c) - f(a) }{ f(b) - f(a)  }  \, ,
% B \eqdef f(b) - f(a)
% \, ,
% C \eqdef f(a)
% $, and $U_1, U_2, U_3$ are uniform distributions over the intervals
% $[\log \log W, 2 \log \log W]$, 
% $[\log^2 W \;\log \log W, 2 \log W]$,
% $[0, W]$ respectively.

Suppose $f(x) = \exp( \exp(\lambda_1) \;x + \lambda_2  ) + \lambda_3$.
We have that
$ \log \log A = g_{a,b,c}( \lambda_1 )$, 
where
$$
g_{a,b,c}(x) = \log \log \lp(   
\frac{ \exp(c \exp(x) ) - \exp(a \exp(x)) }{  \exp(b \exp(x)) - \exp(a \exp(x)) }
\rp).
$$
It is easy to verify that $g_{a,b,c}$ is monotonically increasing as a function of $x$ for any $a < b < c$.
Since $\lambda_1$ is uniform over $[\log \log W, 2 \log \log W]$, the support of $\log \log A$ will be 
$[ g_{a,b,c}( \log \log W ), g_{a,b,c}(2 \log \log W) ]$. By the change of variable rule of probability density functions, we have
% the PDF of $\log \log A$ is given by
\begin{align*}
    \Pr[ \log \log A = x ] = \begin{cases}
    \Pr[  \lambda_1 = g_{a,b,c}^{-1}(x)  ]
\;\frac{1}{  g_{a,b,c}'( g_{a,b,c}^{-1}(x) ) } \text{ , if } x \in [ g_{a,b,c}( \log \log W ), g_{a,b,c}(2 \log \log W) ] \, , \\
0 \text{ otherwise.}
 \end{cases}
\end{align*}
Before we bound the total variation distance between $\log \log A$ and $U_1$, we discuss some useful properties of $g_{a,b,c}$.
\begin{claim} \label{lem:function-approx}
Given $a < b <c \in  [0,1]$ are well separated (satisfying Equation~\eqref{eq:separated-condition}), it holds that
(i) $g_{a,b,c}(x) \leq x$ (ii) $\abs{ g_{a,b,c}(x) - x } \leq \log \log \log W + O(1)$ (iii) $\abs{ g'_{a,b,c}(x) - 1 } \leq  O \lp( \frac{1}{ \log \log W} \rp)$ for $x \in [\log \log W, 2 \log \log W]$.
\end{claim}
\begin{proof}
For the proof of this claim, we will temporarily drop the subscript of $g_{a,b,c}$ and write only~$g$.
For property (i), we have
$$
g(x) \leq 
\log  \log \lp(
\frac{ \exp(c \exp(x) )  }{  \exp(b \exp(x)) }
\rp)
= \log(c-b) + x
\leq x \, ,
$$
where the last inequality is true since $c - b \in [0,1]$, which follows from $b < c$ and $c,b \in [0,1]$.

For properties (ii) and (iii), our strategy is to show that $g(x)$ is approximately just $x + \log(c-b)$ for sufficiently large $W$.
To do so, we consider the function $h(x) := \exp( g(x) ) = \log \lp(
\frac{ \exp(c \exp(x) ) - \exp(a \exp(x)) }{  \exp(b \exp(x)) - \exp(a \exp(x)) }
\rp)$.
Our goal now is to show $h(x)$ is approximately $(c-b) \exp(x)$.
Denote $L_{\theta}(x) = \log( 1 - \exp( -\theta \; \exp(x) ) )$.
We then have
\begin{align} \label{eq:h-error}
    &h(x) 
=
\log\bigg( \exp \big( c \exp(x) \big) - \exp\big(a \exp(x) \big) \bigg)    
-
\log\bigg( \exp \big( b \exp(x) \big) - \exp \big(a \exp(x) \big) \bigg)   \nonumber \\
&=
\log\bigg( \exp \big( c \exp(x) \big) \lp(1- \exp \big( (a-c) \exp(x) \big) \rp) \bigg)
-
\log\bigg( \exp \big( b \exp(x)  \big) \lp(1- \exp \big( (a-b) \exp(x) \big) \rp) \bigg) \nonumber\\
&=
c \exp(x) + 
\log\bigg(  1- \exp \big( (a-c) \exp(x) \big)  \bigg)
-
b \exp(x) - 
\log\bigg(  1- \exp \big( (a-b) \exp(x) \big)  \bigg) \nonumber\\
&=  (c-b) \; \exp(x) + L_{c-a}(x) - L_{b-a}(x).
\end{align}
For $\theta \in [1 / \log \log W, 1]$ and $x \in [\log \log W, 2 \log \log W]$, we claim $L_{\theta}$(x) becomes almost the $0$ function in terms of its function values and its derivative when $W$ grows.
Using the inequality 
$- \log( 1 - \exp(-z) ) \leq 1/z$ for $z >0$,
we have
\begin{align} \label{eq:L-bound}
\abs{L_{ \theta } (x)} \leq \frac{1}{ \theta \exp(x) }
< \frac{1}{  \log \log W }.
\end{align}
Furthermore, the derivative of $L_{\theta}$ can be bounded by
\begin{align} \label{eq:L-derivative-bound}
L'_{\theta}(x) = 
\frac{\exp(x) \; \theta }{  \exp( \exp(x) \; \theta ) +1 }
\leq \frac{ \log^2 W  }{ W^{1 / \log \log W} }   
<  \frac{1}{ \log \log W }.
\end{align}
Combining Equations~\eqref{eq:h-error} and~\eqref{eq:L-bound}, we then have
$$
h(x) = (c-b) \; \exp(x)  \pm O \lp( \frac{1}{  \log \log W } \rp).
$$
% Hence, for sufficiently large $W$, it holds
% $$
% g(x) = \log h(x) 
% = \log \lp( (c-b) \; \exp \lp(x\rp) \pm O \lp( \frac{1}{ \log \log W } 
% \rp)\rp)
% = x \pm O( \log \log \log W ) \, ,
% $$
Notice that $h(x)$ is at most $\exp(x) + O(1/\log \log W)$. Then, we have 
$$
g(x) \leq \log \lp( \exp(x) + O \lp( \frac{1}{ \log \log W} \rp)  \rp)
= x + \log\lp(  1 + O(1) \; \frac{1}{\exp(x) \log \log W  } \rp)
\leq x + O(1).
$$
On the other hand, 
since $(c-b)$ is at least $\frac{1}{\log \log W}$, 
$h(x)$ is at least 
$ \frac{ \exp(x) }{ \log \log W} - O(1 / \log \log W) $.
Then, we have
\begin{align*}
g(x) &\geq
\log \lp(  \frac{ \exp(x) }{ \log \log W} - O\lp( \frac{1}{\log \log W} \rp) \rp)  \\
&= 
\log( \exp(x)- O(1) ) - \log \log W \\
&= x + \log\lp(1 - O(\exp(-x))  \rp) - \log \log \log W \\
&\geq
x + \log\lp((1 - O(\exp(- \log \log W ))  \rp)) - \log \log \log W \\
&\geq x - O(1) - \log \log \log W \, ,
\end{align*}
where the last inequality holds since for sufficiently large $W$, we have $O( \exp(-\log \log W) ) \leq 1/2.$
This then gives us property (ii).

Using Equations~\eqref{eq:L-derivative-bound} and~\eqref{eq:h-error}, we then have
\begin{align*}
    h'(x) = (c-b) \; \exp(x)
    + L'_{c-a}(x)
    + L'_{b-a}(x)
    = (c-b) \; \exp(x)  \pm O \lp( \frac{1}{ \log \log W}\rp).
\end{align*}
Hence, we can bound the derivative of $g(x)$ as
\begin{align*}
    &g'(x) = \frac{1}{h(x)} \; h'(x)
    =  \frac{ (c-b) \; \exp(x)  \pm O \lp( \frac{1}{ \log \log W}\rp)}{  (c-b) \; \exp(x) \pm O \lp( \frac{1}{ \log \log W}\rp)} \\
    &=  \frac{ 1  \pm O \lp( \frac{1}{ \log \log W \; (c-b) \; \exp(x)}\rp)}{  1 \pm O \lp( \frac{1}{ \log \log W \; (c-b) \; \exp(x)}\rp)} 
    =  \frac{ 1  \pm O \lp( \frac{1}{ \log \log W }\rp)}{  1 \pm O \lp( \frac{1}{ \log \log W }\rp)} \\    
    &= 1 \pm O \lp( \frac{1}{ \log \log W}\rp) \, ,
\end{align*}
where in the second last equality, we use the fact that $(c-b) \exp(x)$ is at least $\log W / \log \log W$ and hence lower bounded by a constant for sufficiently large $W$.
This concludes the proof of \Cref{lem:function-approx}.
\end{proof}

To bound the total variation distance between $\log \log A$ and $U_1$, we will introduce $U_{a,b,c}$, which denotes the uniform distribution over $[ g_{a,b,c}( \log\log W ), g_{a,b,c}(2 \log \log W ) ]$.
Then, by the triangle inequality, we have that
$\dtv( \log \log A, U_1 )
\leq \dtv ( \log \log A, U_{a,b,c} )
+ \dtv( U_{a,b,c}, U_1 ).
$
Notice that the second term is just the total variation distance between two uniform variables - one over the interval $[\log \log W, 2 \log \log W]$ and the other over $[ g_{a,b,c} \lp( \log \log W \rp), g_{a,b,c} \lp( 2 \log \log W \rp) ]$. 
By \Cref{lem:function-approx}, it holds 
$\abs{ g_{a,b,c}(x) - x } = O(\log \log \log W)$ and $g_{a,b,c}(x) \leq x$. 
We thus have
$$
g_{a,b,c} \lp( \log \log W \rp) \leq \log \log W \leq  g_{a,b,c} \lp( 2 \log \log W \rp) \leq  2 \log \log W.
$$
Hence, the total variation distance between $U_1$ and $U_{a,b,c}$ is exactly 
\begin{align*}
\frac{1}{2} \; \bigg( &\frac{  \log \log W - g_{a,b,c}(\log \log W)  }{ g_{a,b,c}(2 \log \log W) - g_{a,b,c}(\log \log W) } 
+ \frac{ 2 \log \log W - g_{a,b,c}(2 \log \log W) }
{ \log \log W } \\
&+ 
\lp( g_{a,b,c}(2 \log \log W) - \log \log W  \rp)
\; 
\abs{ \frac{1}{ \log \log W } - \frac{1}{  g_{a,b,c}(2 \log \log W) - g_{a,b,c}(\log \log W) } } \bigg) 
\, ,
\end{align*}
where the first two terms capture the difference between $U_1$ and $U_{a,b,c}$ on the domain such that exactly one of $U_1$ and $U_{a,b,c}$ is supported on, and the last term captures the difference on the domain they are commonly supported on.
For the first two terms, the numerators are of size $O( \log \log \log W )$ and the denominators are at least $\log \log W - O(\log \log \log W)$ since $\abs{g_{a,b,c}(x) - x} = O(\log \log \log W)$.
Therefore, both of them are of order $O( \log \log \log W / \log \log W )$.
For the last term, we have $g_{a,b,c}(2 \log \log W) - \log \log W \leq \log \log W + O(\log \log \log W)$ and 
$$
\abs{ \frac{1}{ \log \log W } - \frac{1}{  g_{a,b,c}(2 \log \log W) - g_{a,b,c}(\log \log W) } } 
\leq O(\log \log \log W).
$$
Hence, in total, we have $\dtv(U_1, U_{a,b,c}) \leq O\lp( \log \log \log W / \log \log W\rp)$.

For the term $\dtv ( \log \log A, U_{a,b,c} )$, one can see that the two variables have the same support.
We will first show the PDF of $\log \log A$ and $U_{a,b,c}$ are point-wise close.
In particular, for $x \in [ g(\log \log W), g(2 \; \log \log W) ]$, we have 
\begin{align*}
&\abs{\Pr[ \log \log A = x ] -
\Pr[ U_{a,b,c} = x ]}\\
&= 
\Big |
\Pr[ \lambda_1 = g_{a,b,c}^{-1}(x) ]
\; \frac{1}{  g_{a,b,c}'( g_{a,b,c}^{-1}(x) ) }
- \frac{1}{ g_{a,b,c}( 2 \; \log \log W ) - g_{a,b,c}( \log \log W )  } \Big |\\
&= \Big | \frac{1}{  \log \log W }
\; \frac{1}{  1 \pm O \lp(\frac{1}{ \log \log W } \rp) }
- \frac{1}{ \log \log W \pm O( \log \log \log W)  }
\Big | \\
&=  O\lp(\frac{ \log \log \log W}{ \lp(\log \log W \rp)^2 }\rp) \;,
% &=  O\lp(\frac{1}{ \lp(\log \log W \rp)^2 }\rp) \; ( \log \log W +O(1) ) = O \lp(\frac{1}{ \log \log W } \rp)\\
\end{align*}
where in the second equality we use the fact 
$\lambda_1$ is a uniform variable over an interval of length $\log \log W$ and that $g'(a,b,c) = 1 \pm O(1 / \log \log W)$ by \Cref{lem:function-approx}.
Then, since the interval where $\log \log A, U_{a,b,c}$ are supported on is of length at most $O( \log \log W )$.
We then have $ \dtv( \log \log A, U_{a,b,c} ) \leq O(\log \log \log W / \log \log W)$. Hence, overall, we then have
$$ \dtv( \log \log A, U_1 ) \leq O(\log \log \log W / \log \log W).
$$
% This concludes the proof for Equation~\eqref{eq:loglog}.

Next, we show $\log B$ conditioned on $\log \log A$ is close to $U_2$.
We can simplify the expression of $\log B$ and arrive at
$$
\log B = \lambda_2 + \log \lp(  \exp( a  \exp \lp( \lambda_1 \rp) ) - \exp(b  \exp \lp( \lambda_1 \rp)  ) \rp).
$$ 
Notice that that since $\log \log A$ depends only on $a,b,c,\lambda_1$ (since $\lambda_2, \lambda_3$ are cancelled in the expression of $A$), conditioning on $\log \log A$ only makes $\lambda_1$ fixed while $\lambda_2$ is still the uniform distribution over $[0, \log^3 W]$, which is the same as $U_2$.
Hence, to show that $\log B$ is close to $U_2$, it suffices to show $\log B - \lambda_2$ is small after fixing any valid choice of $\lambda_1,a,b$.
We can write
\begin{align*}
&\abs{\log \lp(  \exp( a \exp(\lambda_1) ) - \exp(b \exp(\lambda_1)) \rp)}
= a \exp(\lambda_1) + 
\abs{\log \lp( 1 - \exp \lp( (a-b) \; \exp(\lambda_1) \rp) \rp)} \\
&\leq a \exp(\lambda_1)
+ \frac{ \exp(-\lambda_1) }{b-a}
\leq O(\log^2 W) + O( \log \log W / \log W ) 
\leq O(\log^2 W)\, ,
\end{align*}
where in the first inequality we again use that
$\abs{\log ( 1 - \exp(-z) )} \leq 1/z$ for $z > 0$, and in the second inequality we use $a \leq 1$, $\log \log W \leq \lambda_1 \leq 2 \log \log W$, $b-a \geq 1/ \log \log W$.
Then, recall that $\log B$ and $U_2$ are both
uniform variables supported on intervals with the same lengths but different offsets (differ by $O(\log^2 W)$).
Thus, conditioned on any value of $\lambda_1$, we have
$$
\dtv \lp( U_2, \log B \rp) \leq O(1 / \log W).
$$

Lastly, consider the random variables 
$C \eqdef \exp( a  \exp( \lambda_1 ) ) \; \exp( \lambda_2 ) + \lambda_3 $.
Again, we remark that conditioning on $B$ and $A$ only fixes $\lambda_1, \lambda_2$. So $\lambda_3$ is still a uniform random variable over 
$[0,  \exp( 2 \; \log^3 W )  ]$, just like $U_3$.
Hence,  the total variation distance between $C$ and $U_3$ can be bounded by 
$$
\frac{1}{ \exp \lp( 2 \log^3 W  \rp)  }
\exp( a  \exp( \lambda_1 ) ) \; \exp( \lambda_2 )
\leq 
\exp(  \log^3 W + \log^2 W - 2 \log^3 W )
\leq O(1 / W) \, ,
$$
where we use the fact $ a \leq 1, \lambda_1 \leq 2 \log \log W $, $\lambda_2 \leq \log^3 W$.
This concludes the proof of \Cref{lem:tv-to-u}.
\end{proof}
Let $\t, \r$ be the square edge distributions defined in \Cref{def:square-edge}.
In the lower bound construction from \Cref{sec:construction}, within each square, $(\p, \q)$ is  either $ \big( (\t + \r)/2, (\t + \r)/2 \big)$, $(\t, \r)$ or $(\r, \t)$.  
\Cref{lem:diamond-order-match} states that, if the tester is only given the order information of three samples, it cannot tell whether the samples are taken from $\big( (\t + \r)/2, (\t + \r)/2 \big)$ or a random pair from $(\t, \r)$ and $(\r, \t)$.
In order to hide the extra information, one need to apply the transformation specified in \Cref{lem:order-reduction}, which increases the domain size substantially.
Here, we argue that applying transformations sampled from $\mathcal M(W)$ also eliminates most of the extra information in addition to the order information.
\begin{lemma}
\label{lem:three-point-info}
Let $\t, \r$ be the square edge distributions defined in \Cref{def:square-edge}.
Let $\{u_i, v_i, b_i\}_{i=1}^m$ be  samples drawn from the pair of  distributions $\big((\t+\r)/2, (\t + \r)/2 \big)$.
With probability $1/2$, we draw $\{x_i, y_i, \ell_i\}_{i=1}^m$ from $(\t, \r)$.
Otherwise, we draw $\{x_i, y_i, \ell_i\}_{i=1}^m$ from $(\r, \t)$.
Let $f_1, f_2, f_3, f_4$ be four random mappings drawn independently from $\mathcal M(W)$.
Then,  the quantity 
$$ \dtv \lp( 
\{f_1(u_i), f_2(v_i), b_i\}_{i=1}^m,
\{f_3(x_i), f_4(y_i), \ell_i\}_{i=1}^m
\rp)
$$ is $0$
for $m = 1$ and $O(\log \log \log W / \log \log W)$ for $m =2 ,3$.
\end{lemma}
\begin{proof}
We first analyze the case for $m=1$.
We claim that the tuple $(u_1, v_1, b_1)$ has the same distribution as $(x_1, y_1, \ell_1)$.
since conditioned on any values of $(u_1, v_1)$, the distribution of $b_1$ is uniform (and similarly for $x_1, y_1, \ell_1$).
Then, since $f_1, f_2, f_3, f_4$ are all identically distributed, it follows the distributions in the two cases are the same.
% First, we claim that it suffice for us to consider the case $m=3$ since one can always explicitly drop the extra samples and this operation will only decrease the total variation distance.

% For convenience, we will denote 
% $P_{\Yes} := \{ u_i, v_i, b_i \}_{i=1}^3$,
% and $P_{\No} := \{ x_i, y_i, \ell_i \}_{i=1}^3$.
We then proceed to prove the cases $m=2,3$.
We remark that the total variation distance for $m = 2$ is at most that for $m=3$ since one can always explicitly drop the extra sample and this operation will only decrease the total variation distance. Thus, we only need to consider the case $m=3$.
By \Cref{lem:diamond-order-match},
we have $ \Order( \{ u_i, v_i, b_i \}_{i=1}^m )$
has the same distribution as $\Order( \{ x_i, y_i, \ell_i \}_{i=1}^m )$.
Hence, there exists a coupling $J$ between $ \{ u_i, v_i, b_i \}_{i=1}^m $ and $ \{ x_i, y_i, \ell_i \}_{i=1}^m $ such that if we sample from $J$ we always have $\Order( \{ u_i, v_i, b_i \}_{i=1}^m  ) = \Order( \{ x_i, y_i, \ell_i \}_{i=1}^m ) $.
Hence, we can bound the overall total variation distance by
\begin{align*}
\underset{ J } \E
\lp[ 
\dtv \lp( 
\{ f_1(u_i), f_2(v_i), b_i \}_{i=1}^m\, ,
\{ f_3(x_i), f_4(y_i), \ell_i \}_{i=1}^m
\rp) \rp].
\end{align*}
% Subject to $P_{\Yes} = T_{\No}$, we always have that $b_i = \ell_i$.
% Furthermore, 
% Hence, there exists a joint distribution $J$ such that if we sample $(\vec u^{(1)}, \vec u^{(2)}, \vec u^{(3)}), (\vec v^{(1)}, \vec v^{(2)}, \vec v^{(3)})$ from $J$ then we always have 
% $Order( (\vec u^{(1)}, \vec u^{(2)}, \vec u^{(3)})  ) = Order( (\vec v^{(1)}, \vec v^{(2)}, \vec v^{(3)})  )$. Moreover, the marginal distribution over $(\vec u^{(1)}, \vec u^{(2)}, \vec u^{(3)})$ is exactly the same as $T_1 | X=0, p(I_1) = \eps/k, |T_1| = 3$ (and the same for $(\vec v^{(1)}, \vec v^{(2)}, \vec v^{(3)})$). 
% Then, it is easy to see that 
% \begin{align*}
% & \dtv \lp( 
% T_1 | X=0, p(I_1) = \eps/k, |T_1| = 3,
% T_1 | X=1, p(I_1) = \eps/k, |T_1| = 3,
% \rp) \\
% &\leq 
% \E_{ ( \vec u^{(i)} ), (\vec v^{(i)}) \sim J } 
% \lp[ \dtv \lp( 
% \{f( \vec u^{(i)}) \} \, ,
% \{f( \vec v^{(i)}) \}
% \rp)    \rp] \, ,
% \end{align*}
% where the variation distance is over the random choice of the transformation $f$.
% Let $\vec u^{(i)} = (u^{(i)}_1,  u^{(i)}_2   )$. 
We remark that the total variation distance inside the expectation is now for fixed values of $\{u_i, v_i, b_i\}_{i=1}^m$ and $\{x_i, v_i, \ell_i\}_{i=1}^m$ that share the same order information and over the random choice of the transformations $f_1, f_2, f_3, f_4$.
Since the transformations along the two dimensions are picked independently
and $b_i = \ell_i$ under the coupling $J$, we thus have
\begin{align*}
&\dtv \lp( 
\{ f_1(u_i), f_2(v_i), b_i \}_{i=1}^m \, ,
\{ f_3(x_i), f_4(y_i), \ell_i \}_{i=1}^m
\rp) \\
&=
\dtv \lp( 
\{ f_1(u_i) \}_{i=1}^m \, ,
\{ f_3(x_i) \}_{i=1}^m
\rp)
+ 
\dtv \lp( 
\{ f_2(v_i) \}_{i=1}^m \, ,
\{ f_4(y_i) \}_{i=1}^m
\rp).
\end{align*}
The arguments for bounding the total variation distance over the two different dimensions are identical. We will therefore just focus on the first dimension.
Then consider the event $E$ such that 
$\min \lp( u_i - u_{i-1}, x_i - x_{i-1}  \rp) \geq 1 / \log \log W$ for any $i$. By the union bound, it is easy to see that $E$ does not hold under $J$ with probability at most $O(1 / \log \log W)$.
By the triangle inequality, we have
\begin{align*}
\dtv \lp( 
\{ f_1(u_i) \}_{i=1}^m \, ,
\{ f_3(x_i) \}_{i=1}^m
\rp)
\leq
\dtv \lp( 
\{ f_1(u_i) \}_{i=1}^m \, , D
\rp)
+ 
\dtv \lp( 
\{ f_3(x_i) \}_{i=1}^m \, , D
\rp)
\end{align*}
where $D$ is the distribution defined in \Cref{lem:tv-to-u}.
Conditioned on the event $E$, we then have that the expression is bounded by $O( \log \log \log W / \log \log W )$.
Since the total variation distance is bounded by $1$ and $E$ does not hold with probability at most $O(1 / \log \log W)$. The overall total variation distance is at most $O( \log \log \log W / \log \log W )$.
{This completes the proof of \Cref{lem:three-point-info}}
\end{proof}

% If $X=1$, $\p, \q$ restricted to one square will be either $(\t, \r)$ or $(\r, \t)$, each with probability $1/2$.
Now, let $X$ be an unbiased binary variable. Let $\p, \q$ be a pair of measures generated by the random process described in Section~\ref{sec:construction}.
Recall that in the construction from the last section, we divide the domain into $\Theta(k^2)$ squares and $\p, \q$ are only supported on the $\Theta(k)$ squares along the diagonal.
We will then apply the following domain transformation.
For each square along the diagonal, we will independently generate two monotonic mappings $f_1, f_2 \sim \mathcal M$.
Then, we stretch the square along the $x$-axis by $f_1$ and stretch it along the $y$-axis by $f_2$.
We will denote the transformed measures as $\p', \q'$.
Let $P$ be the set of samples obtained by taking $\Poi(m)$ samples from $\p', \q'$.
We claim that $P$ reveals little information about $X$.
\begin{lemma} \label{lem:mutual-info-bound-2}
Suppose $m < k$.
The mutual information between $X$ and $P$ is at most
$$
I(X: P)  \leq O \lp (    
\frac{ m^3 \eps^2 }{k^2 } \; \frac{\log \log \log W }{ \log \log W} +
\frac{m^7 \eps^8}{ k^6 } \rp).
$$
\end{lemma}
\begin{proof}
Let $P_i$ be the samples taken from the $i$-th square. Notice that $P_i, P_j$ for $i \neq j$ are conditionally independent on $X$. 
Hence, we have $I(X:P) \leq O(k) \; I(X:P_1)$.
% As a slight abuse of notation, we will write
% $\Pr[ P_1 = s ]$ for some multisets $s$ representing the sample values to denote the probability density function of $P_1$. 
We will use a multiset $s$ made up of elements from $\R_{+}^2$ to represent the possible values $P_1$ can take.
Besides, we write $\abs{s}$ to represent the size of the multiset.
Then, it holds that
$$
I(X:P_1)
= O(1) \; \sum_{ \gamma = 1}^{\infty}
\int_{ |s| = \gamma }
\frac{ \lp( \Pr[ P_1 = s | X = 0  ] - \Pr[ P_1 = s | X = 1  ] \rp)^2 }{ \Pr[P_1 = s] }.
$$
Let $H_1$ be the indicator variable of whether the first square is selected as a heavy square.
We can use techniques similar to the proof of \Cref{lem:mutual-info-bound} to show that $I(X:P_1)$ can be bounded by
\begin{align*} 
    \sum_{\gamma=1}^{\infty}
    O \lp( 2^{\gamma} \rp) \; \frac{      
    \Pr\lp[ \abs{P_1} = \gamma \big| H_1 = 0  \rp]
     }
    { \Pr \lp[ |P_1| = \gamma, H_1 = 1\rp] }
    \; 
    \int_{ |s| = \gamma } \abs{ \Pr[ P_1 = s | X = 0, H_1 = 0 ] - \Pr[P_1 = s | X = 1, H_1 = 0] }.
\end{align*}
We will take a closer look at the integral in the expression. 
Given the observation that {$|P_1|$ and $X$ is conditionally independent on $H_1$}, 
it is not hard to see that
\begin{align*}
&\int_{|s| = \gamma} \abs{ \Pr[ P_1 = s | X = 0, H_1 = 0 ] - \Pr[P_1 = s | X = 1, H_1 = 0] } \\
&= 
2 \; \dtv \lp( P_1 | (X = 0, H_1 = 0, |P_1| = \gamma), P_1 | (X = 1, H_1 = 0, |P_1| = \gamma)\rp)
\; \Pr[ |P_1| = \gamma | H_1 = 0
].
\end{align*}
Notice that $P_1 | (X = 0, H_1 = 0, |P_1| = \gamma)$ and $P_1 | (X = 1, H_1 = 0, |P_1| = \gamma)$ correspond exactly to the distributions of $\{ f_1(x_i), f_2(y_i), \ell_i \}_{i=1}^{\gamma}$ and 
$\{ f_3(u_i), f_4(v_i), b_i \}_{i=1}^{\gamma}$ specified in \Cref{lem:three-point-info}.
Therefore, for $1 \leq \gamma \leq 3$, we can apply \Cref{lem:three-point-info} and bound the total variation distance by $0$ for $\gamma = 1$ and $O( \log \log \log W / \log \log W )$ for $\gamma = 2,3$.
For $\gamma \geq 4$, we will simply bound the total variation distance by $1$.
This then allows us to bound $I(X: P_1)$ by
\begin{align*} \\ 
O(1) \; 
\sum_{\gamma=2}^{3}
 \frac{      
\lp( 
\Pr\lp[ \abs{P_1} = \gamma \big| H_i = 0  \rp] \rp)^2
 }
{ \Pr \lp[ |P_1| = \gamma, H_i = 1\rp] }
\; 
\frac{ \log \log \log W }{ \log \log W}
+ 
O(1) \; 
\sum_{\gamma=4}^{\infty}
2^{\gamma} \; 
 \frac{      
\lp( 
\Pr\lp[ \abs{P_1} = \gamma \big| H_i = 0  \rp] \rp)^2
 }
{ \Pr \lp[ |P_1| = \gamma, H_i = 1\rp] }.
\end{align*}
Now, recall that $\Pr[H_1 = 1] = m/k$. 
When $H_1 = 1$, the mass of the square will be $1/m$ and hence $\abs{P_1} | H_1 = 1$ will be distributed as $\Poi(1)$. Therefore, we have
$\Pr[|P_1| = \gamma | H_i = 1] = \Poi( 
1, \gamma )
= \Theta(1) / \gamma!
$.
On the other hand, when $H_1 = 0$, the mass of the square is $\eps / k$. Hence, $\abs{P_1} | H_1 = 0$ is distributed as $\Poi( \eps \; m/k )$.
Therefore, we have
$\Pr[ |P_1| =  \gamma | H_i = 0 ] = \Poi( \eps \; m/k, \gamma  ) \leq  (\eps \; m/k)^{\gamma} / \gamma!$.  
Together with our assumption $m < k$, we can simplify the bound as
\begin{align*}
I(X:P_1)
\leq  O(1) \; \frac{ m^3 \eps^4 }{ k^3 }
\; \frac{  \log \log \log W }{ \log \log W}
+ 
O(1) \; \frac{ m^7 \eps^8 }{ k^7 }.
\end{align*}
This concludes the proof of \Cref{lem:mutual-info-bound-2}.
\end{proof}
We are now ready to conclude the proof of Theorem~\ref{thm:lower-bound-refined}.
\begin{proof}[Proof of Theorem~\ref{thm:lower-bound-refined}]
Throughout the proof, we assume that $m < k/2$ as this is the regime where we can use the random process described in Section~\ref{sec:construction} to generate measures.

Let $X$ be an unbiased binary variable and
$\p, \q$ be a pair of measures generated according to the random process described in Section~\ref{sec:construction} and $\p', \q'$ be the measures obtained after applying the random transformation defined by mappings sampled from $\mathcal M(W)$. 
Since the transformation is monotonic in both $x$ and $y$ axis, we thus have 
$ \snorm{\Ak}{ \p -  \q} = \snorm{\Ak}{ \p' - \q' }$.
Therefore, when $X = 0$, we have $\p' = \q'$; when $X=1$, we have $\snorm{\Ak}{\p' - \q'} > \eps$.
% By Lemma~\ref{lem:basic-properties}, given that $m < k/2$, it holds that both $\p, \q$ are measures of mass $\Theta(1)$ with probability at least $99\%$ and with if $X = 0$, it holds 
% $\snorm{\Ak}{\p - \q} > \Omega(\eps)$ with probability at least $99\%$.
By Lemma~\ref{lem:mutual-info-bound-2}, we have the mutual information between the random bit $X$ and the output of any algorithm that uses $\Poi(m)$ samples is at most 
$O \lp (    
\frac{ m^3 \eps^4 }{k^2 } \; \frac{\log \log \log W }{ \log \log W} +
\frac{m^7 \eps^8}{ k^6 } \rp)$.
Hence, no tester can reliably distinguish between the case that $\p' = \q'$ and $\snorm{\Ak}{\p'-\q'} > \eps$ with probability more than $2/3$ unless 
\begin{align} \label{eq:W-sample-bound}
m \geq \Omega(1) \; \min \lp( 
k^{2/3} \eps^{-4/3} \; \lp( \frac{ \log \log W }{ \log \log \log W} \rp)^{1/3} \, , \, 
k^{6/7} \eps^{-8/7}
\rp).    
\end{align}
% If we set $ W  = \exp ( \exp ( C \; k ) )$ for some sufficiently large constant $C$, it then holds that no tester can  reliably distinguish between the case that $\p =\q$ and $\snorm{\Ak}{\p-\q} > \Omega(\eps)$ with probability more than $2/3$ unless $m > \Omega(1) \; \min \lp( k^{6/7} / \eps^{8/7}, k\rp) $ provided that $\p, \q$ are continuous distributions over some $W' \times W'$ squares where $W' \leq \exp \lp( \exp \lp(  \Theta(k) \rp) \rp)$.

Note that the measures $\p', \q'$ are continuous.
The remaining step is to turn them into discrete measures $\tilde \p', \tilde \q'$
such that distinguishing between $\tilde \p' = \tilde \q'$ versus $\snorm{\Ak}{ \tilde \p' - \tilde \q' } \geq \eps$ is about as hard as $ \p' =  \q'$ versus $\snorm{\Ak}{  \p' -  \q' } \geq \eps$ .

% Now, we turn the lower bound into one for discrete distributions.
First, we argue that, for any horizontal or vertical strip of width at most $\frac{ \eps }{ 8}$, the mass of $\p' ,\q'$ is at most $\frac{\eps}{8 k}$. 
It is easy to see the claim is true for $\p, \q$ since their marginal distributions in any dimension is uniform over intervals whose lengths add up to at least $k$.
For the transformed distribution $\p', \q'$, the bound still holds since the transformation only stretches the distribution along $x$, $y$ axis.

Then, we can construct a grid $\mathcal G$ which splits the domain into small unit squares, each of size $\frac{ \eps }{  8 } \times \frac{\eps}{ 8}$.
Then, consider $\tilde \p', \tilde \q'$ which round the points falling in each square in $\mathcal G$ to its top-left vertex.
Then, for an arbitrary rectangle $R$, {$\abs{\p'(R) - \q'(R)} - \abs{\tilde \p'(R) - \tilde \q'(R)  }$ is at most} the mass of $\p'$ or $\q'$ in the two vertical strips and the two horizontal strips, each of width at most $\frac{\eps}{8}$.
Thus, for any $R$, it holds
$$
\abs{\tilde \p'(R) - \tilde \q'(R)  }
\geq 
\abs{ \p'(R) -  \q'(R)  } - \eps/(2k).
$$
Consequently, it holds $\snorm{ \Ak }{ \tilde \p' - \tilde \q' }  \geq \eps/2$ if $\snorm{ \Ak }{  \p' -  \q' } \geq \eps$. On the other hand, if $\p' = \q'$, it is easy to see that we still have $\tilde \p' = \tilde \q'$ after the rounding.
Thus, if there is an algorithm which can distinguish between the cases $\tilde \p' = \tilde \q'$ and $\snorm{\Ak}{\tilde \p' - \tilde \q'} > \eps / 2$, we can use it to distinguish between the cases $\p' =  \q'$ and $\snorm{\Ak}{\p' -  \q'} > \eps$ as well by simulating the rounding process. 
Hence, the sample complexity lower bound in Equation~\eqref{eq:W-sample-bound} applies to $\tilde \p', \tilde \q'$ as well.

Finally, we note the supports of the transformed measures $\p', \q'$ are always contained in some $W' \times W'$ square where $ W' =  k \; \exp( 4 \log^3 W )$. Hence, $\tilde \p', \tilde \q'$ are over a
$ V \times V  $ discrete grid where
$V = \Theta \lp( k \; \exp( 4 \log^3 W ) / \eps \rp)$.
If the first term in the sample complexity bound (Equation~\eqref{eq:W-sample-bound}) is dominating, we must have
$\log \log W \geq  \log(k / \eps) $. This then implies that
$V$ is at most $\exp( 5 \log^3 W )$, which further implies that $\log \log W \geq \Omega(1) \; \log \log V$. On the other hand, it is easy to see that $V > W$ and so 
$1/ \log \log \log W > 1 / \log \log \log V$.
We can the rewrite Equation~\eqref{eq:W-sample-bound} as
$$
m \geq \Omega(1) \; \min \lp( 
k^{2/3} \eps^{-4/3} \; \lp( \frac{ \log \log V }{ \log \log \log V} \rp)^{1/3} \, , \, 
k^{6/7} \eps^{-8/7}
\rp) \, ,
$$
which is indeed the desired lower bound.
\end{proof}

\input{conclusion}

%% file: conclusion.tex
\section{Conclusions and Open Problems} \label{sec:open}

{
In this work, we studied the problem of closeness testing between two 
multidimensional distributions under the $\Ak$ distance.
Our main contribution is the first tester for this task 
with sublinear sample complexity. 
The sample complexity of our tester is provably near-optimal as a function of the parameter $k$ (within logarithmic factors)
for any fixed dimension $d \geq 2$.

Conceptually, our sample complexity lower bound implies that the testing problem 
is provably harder in the multidimensional setting. 
In particular, there is a ``phase transition'' between 
the one-dimensional and the two-dimensional cases.
On the positive side, we show that as the dimension $d$ 
further increases the dependency of the sample complexity 
on $k$ --- the main parameter of our interest --- 
stays approximately the same. 

As immediate corollaries of our $\Ak$ closeness tester, 
we also obtain the first closeness tester for families of structured multidimensional distributions --- 
including $k$-histograms and uniform distributions over unions of axis-aligned rectangles --- under the total variation distance.
}

{
While \Cref{thm:main-intro} implies that our upper and lower bounds are nearly optimal 
in terms of their dependence on $k$, 
their dependence on $\eps$ do not match. 
In particular, the upper bound scales polynomially 
with $1/\eps$, where the degree of the polynomial 
depends on the dimension $d$. On the other hand, 
the lower bound applies to $2$-dimensional 
distributions, and hence has a constant exponent in its (polynomial) $\eps$-dependence. 
%There is certainly still room of improvement for the upper bound's dependency on $\eps$. For example, we can always learn the distributions in $\Ak$ distance with a sample complexity that scales quadratically with $\eps$. 
%However, it is unclear whether one can achieve any \emph{sub-learning} sample complexity in the regime that $1/\eps$ is significantly larger than the other parameters. 
This leads to the following question.
\begin{question} \label{q:eps-dependency}
{\em What is the optimal sample complexity as a function of $\eps$ for multidimensional $\Ak$ closeness testing?}
\end{question}
In the current and prior works, the multidimensional $\Ak$-distance is defined 
as the maximum discrepancy between two distributions 
over $k$ disjoint axis-aligned \emph{rectangles}. 
On the other hand, the $\Ak$-distance for univariate 
distributions is defined with respect to \emph{intervals}.
This definition inherently uses axis-aligned rectangles 
in $\R^d$, as the natural generalization of intervals in $\R$.
Yet, rectangles are not necessarily the only valid choice. 
More specifically, one can replace axis-aligned rectangles 
in the definition 
of multidimensional $\Ak$ distance with other geometric shapes whose $1$-dimensional projection 
corresponds to intervals. For example, we can use 
shapes like unit-balls, simplices, 
or any other convex set.
Such natural variants of multidimensional $\Ak$ distance 
can be used to build $\dtv$-closeness testers 
of other families of structured distributions, 
such as log-concave distributions. 
This leads to the following question.
\begin{question}
{\em Are there alternative definitions of multidimensional $\Ak$ distance 
for multivariate distributions that can lead to optimal $\dtv$-closeness/identity testers 
for other multivariate shape-restricted distributions? }
\end{question}
Exploring other notions of multidimensional $\Ak$ distance is of significant interest and 
may lead to a unified theory of testing multivariate structured distributions.
%that relates testing structured distributions in $\ell_1$ distance with testing in alternative metrics in the same way as the $\Ak$ distance for univariate distributions.
}

%% file: appendix.tex
% \section{Appendix}
% \begin{theorem}\label{cubeboundtight}
% Theorem \ref{forcedCubeUpper} is tight. That is, there exists a set $S$ of $2^{2^{d-1}}$ points in $\R^d$ such that no point is contain in the axis-aligned rectangle formed by two others.
% \end{theorem}
% \begin{proof}
% Let $\sgn(x) = (\sgn(x_1),\ldots,\sgn(x_d))$ for $x \in \R^d$. Additionally, let $v_1, \ldots, v_{2^{d-1}} \in \R^d$ be all vectors whose first coordinate is 1 and remaining coordinates are $\pm$ 1, and let $S = \{\sum_{i=1}^{2^{d-1}} \alpha_i 3^i v_i : \alpha_i \in \{\pm 1\}\}$. Note that $|S| = 2^{2^{d-1}}$.
% We claim $S$ does not have two points whose axis-aligned rectangle contains a third.

% Note that for $x$ and $y$ in $S$, $\sgn(x-y) = \pm v_j$, where $j$ is the largest index such that the coefficient of $v_j$ in the definitions of $x$ and $y$ disagree. From this, it is easy to see that for any distinct $x,y,z \in S$, we cannot have $\sgn(x-y)$, $\sgn(y-z)$, and $\sgn(z-x)$ all the same.
% \end{proof}